\documentclass[submission,copyright,creativecommons]{eptcs}
\providecommand{\event}{QPL 2021} 
\usepackage{breakurl}             
\usepackage{underscore}           
\input{intrels.sty}
\title{Generators and Relations for the Group $\OD$}
\author{Sarah Meng Li
\institute{Dalhousie University}
\email{sarah.li@dal.ca}
\and
Neil J. Ross
\institute{Dalhousie University}
\email{neil.jr.ross@dal.ca}
\and
Peter Selinger
\institute{Dalhousie University}
\email{selinger@dal.ca}
}
\def\titlerunning{Generators and Relations for the Group $\OD$}
\def\authorrunning{S. M. Li, N. J. Ross \& P. Selinger}
\begin{document}
\maketitle

\begin{abstract}
  We give a finite presentation by generators and relations for the
  group $\On(\Z[1/2])$ of $n$-dimensional orthogonal matrices with
  entries in $\Z[1/2]$. We then obtain a similar presentation for the
  group of $n$-dimensional orthogonal matrices of the form
  $M/\sqrt{2}{}^k$, where $k$ is a nonnegative integer and $M$ is an
  integer matrix. Both groups arise in the study of quantum
  circuits. In particular, when the dimension is a power of 2, the
  elements of the latter group are precisely the unitary matrices that
  can be represented by a quantum circuit over the universal gate set
  consisting of the Toffoli gate, the Hadamard gate, and the
  computational ancilla.
\end{abstract}

\section{Introduction}
\label{sec:intro}
There is a beautiful correspondence which relates certain quantum circuits and matrices over rings of algebraic integers \cite{AGR2019,fgkm15, GS13, KMM-exact, ky15}. A first instance of this correspondence arises in the study of circuits over the gate set $\s{CCX, H\otimes H}$, where $CCX$ is the Toffoli gate and $H\otimes H$ is the twofold tensor product of the Hadamard gate. In this case, the correspondence takes a particularly simple form: a unitary matrix $M$ can be exactly represented by an $n$-qubit quantum circuit over $\s{CCX, H\otimes H}$ if and only if $M\in \OD$, where $\OD$ is the group of \emph{orthogonal dyadic matrices}. A second instance of the correspondence follows as a corollary of this first one: circuits over the gate set $\s{CCX, H}$ correspond to orthogonal matrices of the form $M/\sqrt{2}{}^k$, where $M$ is an integer matrix and $k$ is a nonnegative integer. These matrices form the group of \emph{orthogonal scaled dyadic matrices}. The above gate sets are ubiquitous in the theory of quantum computation \cite{Aharonov03asimple, bjs2010, KLM07, Shi2003, Montanaro2017}.
 
The correspondence between quantum circuits and matrix groups exposes the mathematical structure underlying certain gate sets, thereby enabling exact and efficient manipulation of circuits. These insights, along with applications such as compiling \cite{glaudell2021optimal,kbry15,kmm-approx,r15,RS16} and verification
\cite{amy2019}, motivate the study of the relevant matrix groups.

In this paper, we give a finite presentation by generators and relations for the group $\OD$, following the approach initiated in \cite{Gr2014}. It was shown in \cite{AGR2019} that $\OD$ is generated by the collection of 1-, 2-, and 4-level operators of type $-1$, $X$, and $H\otimes H$. To give a presentation of $\OD$ we introduce a finite list of relations among these generators and show that two words over the generators denote the same element of $\OD$ if and only if one word can be converted into the other using a finite number of applications of the relations. Remarkably, the relations can be stated independently of $n$. As a corollary of our main result, we obtain a similar presentation for the group of matrices of the form $M/\sqrt{2}{}^k$ mentioned above.

The paper is structured as follows. In \cref{sec:gens}, we introduce
the generators, along with some basic definitions. In
\cref{sec:synth}, we give a detailed presentation of the exact
synthesis algorithm of \cite{AGR2019}. In \cref{sec:presentation}, we
introduce the relations and prove our main result: the relations are
sound and complete. In \cref{sec:supin}, we use the results of
\cref{sec:presentation} to give a presentation of the group of
orthogonal scaled dyadic matrices. We draw some final conclusions in
\cref{sec:conclusion}.

\section{Generators}
\label{sec:gens}
\begin{definition}
  \label{def:Ztwo}
  The ring of \emph{dyadic rationals} is defined as $\Ztwo=\left\{
  \frac{u}{2^k} \mid u\in \Z, k\in\N \right\}$.
\end{definition}

\begin{definition}
  \label{def:lde}
  Let $t$ be a dyadic rational. A natural number $k$ is a
  \emph{denominator exponent} of $t$ if $2^k t \in \Z$. The least such
  $k$ is called the \emph{least denominator exponent} of $t$ and is
  denoted by $\lde(t)$.
\end{definition}

We extend \cref{def:lde} to matrices as follows. A natural number $k$
is a denominator exponent of a matrix $M$ if it is a denominator
exponent of all of the entries of $M$. Similarly, the least
denominator exponent of $M$ is the least $k$ that is a denominator
exponent for all of its entries, which we write $\lde(M)$.

\begin{definition}
  \label{def:OZtwo}
  The $n$-dimensional group of \emph{orthogonal dyadic matrices}
  consists of the $n\times n$ orthogonal matrices with entries in
  $\Ztwo$. It is denoted $\OD$.
\end{definition}

\begin{definition}
  \label{def:basegens}
  The matrices $X$, $H$, and $K$ are defined as
  \[
  X = 
  \begin{bmatrix}
  0 & 1 \\ 1 & 0
  \end{bmatrix},
  \quad
  H = \frac{1}{\sqrt 2}
  \begin{bmatrix}
  1 & 1 \\ 1 & -1
  \end{bmatrix} , 
  \quad  
  \mbox{ and }
  \quad
  K = \frac{1}{2}
  \begin{bmatrix}
    1 &  1 &  1 &  1 \\
    1 & -1 &  1 & -1 \\
    1 &  1 & -1 & -1 \\
    1 & -1 & -1 &  1 \\
  \end{bmatrix}.
  \]
\end{definition}

The matrix $X$ is known as the \emph{Pauli $X$ gate} and the matrix
$H$ is known as the \emph{Hadamard gate}. We have $K= H\otimes H$,
where $\otimes$ is the usual tensor product. We now embed $X$, $H$,
and $K$ into larger matrices which will serve as our generators.

\begin{definition}
  \label{def:onetwolevel}
  Let $M$ be an $m\times m$ matrix, let $m\leq n$, and let $1 \leq
  a_1,\ldots,a_m\leq n$. The \emph{$m$-level matrix of type $M$} is
  the $n\times n$ matrix $M_{[a_1,\ldots,a_m]}$ defined by
  \[
    {M_{[a_1,\ldots,a_m]}}_{i,j} =
    \begin{cases}
      M_{i',j'} \mbox{ if } i=a_{i'} \mbox{ and } j = a_{j'}\\
      I_{i,j} \mbox{ otherwise.}
    \end{cases}
  \]
\end{definition}

\begin{definition}
  \label{def:gens}  
  The set $\gens_n$ of \emph{$n$-dimensional generators} is the subset
  of $\OD$ defined as
  \[
    \gens_n=\s{\mone{a},\xx{a,b},\hh{a,b,c,d}\mid 1\leq a<b<c<d
      \leq n}.
  \]
\end{definition}

\section{Constructive Membership for
  \texorpdfstring{$\OD$}{On(Z[1/2])}}
\label{sec:synth}

In this section, we present a solution to the constructive membership
problem for $\OD$, following \cite{AGR2019}. To this end, we describe
an algorithm which inputs an arbitrary element $M$ of $\OD$ and
outputs a sequence of elements of $\gens_n$ representing $M$. As is
common in the quantum computing literature, we refer to the algorithm
as the \emph{exact synthesis} algorithm. In addition to showing that
$\gens_n$ generates $\OD$, the algorithm will play a central role in
the rest of the paper.

\begin{lemma}
  \label{lem:twohs}
  Let $u_1,u_2,u_3,u_4$ be odd integers. Then there exist
  $\tau_1,\tau_2,\tau_3,\tau_4\in\Z_2$ such that
  \[
  \hh{1,2,3,4} \mone{1}^{\tau_1}\mone{2}^{\tau_2}
  \mone{3}^{\tau_3}\mone{4}^{\tau_4}
  \begin{bmatrix}
    u_1 \\ u_2 \\ u_3 \\ u_4
  \end{bmatrix} =
  \begin{bmatrix} 
    u_1' \\ u_2' \\ u_3'\\ u_4' 
  \end{bmatrix}
  \]
  where $u_1',u_2',u_3',u_4'$ are even integers.
\end{lemma}

\begin{proof}
Because $u_i\equiv 1 \pmod{2}$, we have $u_i\equiv 1,3 \pmod{4}$. And
since $-3\equiv 1 \pmod{4}$ there exists $\tau_i\in\Z_2$ such that
$(-1)^{\tau_i}u_i\equiv 1 \pmod{4}$. The claim then follows by
computation.
\end{proof}

\begin{lemma}
  \label{lem:indstep}
  Let $v \in \Ztwo^n$ be a unit vector. If $\lde(v)=k>0$, then there
  exists a sequence $G_1,\ldots,G_q$ of elements of $\gens_n$ such
  that $\lde(G_q\cdots G_1 v) <k$.
\end{lemma}

\begin{proof}
Let $w = 2^k v$, so that $w\in \Z^n $. Since $v^\intercal v = 1$, we
have $w^\intercal w =4^k$ and therefore $\sum w^2_j = 4^k$. Note that
$w_j^2 \equiv 1 \pmod{4}$ if and only if $w_j$ is odd and that $w_j^2
\equiv 0 \pmod{4}$ if and only if $w_j$ is even. Hence the number of
$w_j$ such that $w_j^2 \equiv 1 \pmod{4}$ is a multiple of 4. Let
$w_{a_1},\ldots,w_{a_{4q}}$ be the odd entries of $w$ in order of
increasing index. We can apply \cref{lem:twohs} to $w_{a_1},\ldots,
w_{a_4}$, then to $w_{a_5},\ldots, w_{a_8}$, and so on until the
entries of $w$ are all even. This yields a sequence $G_1,\cdots,G_q
\in \gens_n$ such that
\[
G_q\cdots G_1 v = G_q\cdots G_1 \frac{1}{2^k}w = \frac{2}{2^k}w' =
\frac{1}{2^{k-1}}w'
\]
where $w'\in\Z^n$.
\end{proof}

\begin{lemma}
  \label{lem:indbase}
  Let $v \in \Ztwo^n$ be a unit vector. If $\lde(v)=0$, then $v=\pm
  e_j$ for some $1\leq j \leq n$, where $e_j$ is the $j$-th standard
  basis vector.
\end{lemma}

\begin{proof}
If $k=0$ then $v\in\Z^n$. Since $v$ is a unit vector we then get $\sum
v^2_j = 1$. Since the $v_j$ are integers, there must be exactly one
$j$ such that $v_j=\pm 1$ while all the other entries of $v$ are 0.
\end{proof}

\begin{lemma}
  \label{lem:column}
  Let $v\in\Ztwo^n$ be a unit vector and let $1\leq j\leq n$. Then
  there exists a sequence of generators $G_1,\ldots,G_q\in\gens_n$
  such that $G_q\cdots G_1 v=e_j$.
\end{lemma}

\begin{proof}
By induction on $\lde(v)$. If $\lde(v)=0$ then $v=\pm e_{j'}$ for some
$j'$, by \cref{lem:indbase}. If $e_{j'}=e_{j}$ there is nothing to
do. Otherwise, we can map $v$ to $e_j$ by applying an optional $(-1)$
generator followed by an optional $X$ generator. Now if $\lde(v)=k>0$
then by \cref{lem:indstep} there exists a sequence $G_p,\ldots,G_1$ of
elements of $\gens_n$ such that $\lde(G_p \cdots G_1 v) < \lde(v)$. By
induction, there exists a sequence $G_{p+1},\ldots,G_{q}$ such that
$G_{q}\cdots G_{p+1}G_p \cdots G_1 v = e_j$.
\end{proof}

\cref{lem:column} can be used iteratively on the columns of an
arbitrary element of $\OD$ to reduce it to the identity matrix.

\begin{theorem}
  \label{thm:membership}
  Let $M$ be an $n\times n$ matrix. Then $M\in \OD$ if, and only if,
  $M$ can be written as a product of elements of $\gens_n$.
\end{theorem}

\begin{proof}
The right-to-left direction follows from the fact that
$\gens_n\subseteq\OD$. For the left-to-right direction, apply
\cref{lem:column} to reduce the rightmost column of $M$ to $e_n$, then
proceed recursively.
\end{proof}

The algorithm establishing the left-to-right direction of
\cref{thm:membership} is the exact synthesis algorithm. For future
reference, an explicit description is given in \cref{alg:algo}.

\begin{algorithm}[ht]
  \caption{Exact Synthesis\label{alg:algo}}
  \DontPrintSemicolon
  \SetAlgoLined
  \SetKwInOut{Input}{Input}\SetKwInOut{Output}{Output}
  \Input{An element $M$ of $\OD$}
  \Output{A sequence $\word{W}_1,\ldots,\word{W}_\ell$ of words over $\gens_n$ such that $\word{W}_\ell \cdots \word{W}_1 M = I$}
  $N \leftarrow M$\;
  \While{$N\neq I$}{
  Let $j$ be the greatest integer such that $Ne_j \neq e_j$\;
  Let $v = Ne_j$\;
  Let $k = \lde(v)$\;
  Let $w = 2^kv$\;
    \Case{$k=0$}{
      Let $v = (-1)^\tau e_a$ for some $a$ such that $1\leq a \leq j$ and some $\tau\in\Z_2$\;
      \lIf{$a = j$}{$\word{W} = (-1)_{[j]}^\tau$ \quad // note that $\tau=1$ in this case}
      \lIf{$a < j$}{$\word{W} = X_{[a,j]}\mone{a}^\tau$}
    }
    \Case{$k>0$}{
      Let $a, b, c, d$ be the indices of the first four odd entries of $w$\;
      Let $\tau_a,\tau_b,\tau_c,\tau_d \in \Z_2$ be such that $(-1)^{\tau_i}w_i\equiv 1 \pmod{4}$ for $i\in\s{a,b,c,d}$\;
      $\word{W} = \hh{a,b,c,d} \mone{a}^{\tau_a} \mone{b}^{\tau_b} \mone{c}^{\tau_c}\mone{d}^{\tau_d}$\;
    }
    Output $\word{W}$\;
    $N \leftarrow \word{W}N$\;
  }
\end{algorithm}

Intuitively, \cref{alg:algo} terminates because each iteration of the
algorithm rewrites the input matrix into one that is closer to the
identity. We introduce a notion of \emph{level} which makes this
intuition precise.

\begin{definition}
  \label{def:level}
  Let $M\in\OD$. The \emph{level} of $M$ is the triple $(j,k,\ell)$,
  where
  \begin{itemize}
    \item $j$ is the largest element of $[n]$ such that $Me_j\neq
      e_j$, or $j=0$ if no such index exists;
    \item $k = \lde(M e_j)$, or $k=0$ if $j=0$; and
    \item $\ell$ is the number of odd entries in $2^k(M e_j)$, or
      $\ell=0$ if $k=0$.
  \end{itemize}
  We denote the level of $M$ by $\level(M)$. If $\level(M)=(j,k,\ell)$
  we call $Me_j$ the \emph{pivot column} of $M$.
\end{definition}

Levels are ordered lexicographically and it can be verified that each
iteration of the algorithm strictly decreases the level of $N$.

\section{A Finite Presentation of \texorpdfstring{$\OD$}{On(Z[1/2])}}
\label{sec:presentation}
\cref{thm:membership} shows that the group generated by $\gens_n$ is
$\OD$. However, $\OD$ is not free over $\gens_n$ since there are
relations among the generators, such as $\mone{1}\mone{1}=I$. Our goal
is to give a presentation of $\OD$ by generators and relations,
adopting the approach of \cite{Gr2014}. We start by introducing some
useful terminology.

If $A$ is a set, we write $A^*$ for the collection of words over
$A$. We use $\word{W}$ to denote words, and we sometimes write
$\epsilon$ for the \emph{empty word}. If $\word{W}=A_1\ldots A_m$ is a
word over $A$ then the \emph{length} of $\word{W}$ is $m$. We will be
particularly interested in words over $\gens_n$. Any such word
$\word{W}$ can be \emph{interpreted} as an element $\interp{\word{W}}$
of $\OD$ by multiplying the generators that compose $\word{W}$. That
is, if $\word{W}=G_1\ldots G_m$ then
\[
\interp{\word{W}} = G_1 \cdot \ldots \cdot G_{m-1}\cdot G_m,
\]
where the product is the usual multiplication of matrices. This notion
of interpretation induces a first equivalence relation on $\gens_n^*$.

\begin{definition}
  \label{def:semeq}
  The relation $\sim$ on $\gens_n^*$ is defined by
  $\word{V}\sim\word{W}$ if $\interp{\word{V}}=\interp{\word{W}}$. Two
  words $\word{V}$ and $\word{W}$ such that $\word{V}\sim\word{W}$ are
  said to be \emph{semantically equivalent}.
\end{definition}

Intuitively, two words are semantically equivalent if they denote the
same element of $\OD$. In contrast to this semantic notion of
equivalence, we now introduce a \emph{syntactic} notion of equivalence
which does not rely on the interpretation of words as matrices.

\begin{table}
  \begin{subequations} 
    \renewcommand{\theparentequation}{\arabic{parentequation}}%
    \begin{align}
      \label{rel:orderx} \xx{a,b}^2\ &\approx\ \epsilon \\    
      \label{rel:ordermone} \mone{a}^2\ &\approx\ \epsilon \\  
      \label{rel:orderk} \hh{a,b,c,d}^2\ &\approx\ \epsilon \\
      \notag \\
      \resetparent
      \label{rel:disjoint1} \xx{a,b}\xx{c,d}\ &\approx\ \xx{c,d}\xx{a,b} \\    
      \label{rel:disjoint2} \xx{a,b}\mone{c}\ &\approx\ \mone{c}\xx{a,b} \\
      \label{rel:disjoint3} \xx{a,b}\hh{c,d,e,f}\ &\approx\ \hh{c,d,e,f}\xx{a,b} \\  
      \label{rel:disjoint4} \mone{a}\mone{b}\ &\approx\ \mone{b}\mone{a} \\
      \label{rel:disjoint5} \mone{a}\hh{b,c,d,e}\ &\approx\ \hh{b,c,d,e}\mone{a} \\
      \label{rel:disjoint6} \hh{a,b,c,d}\hh{e,f,g,h}\ &\approx\ \hh{e,f,g,h}\hh{a,b,c,d} \\
      \notag \\
      \resetparent
      \label{rel:rename1} \xx{a,a'}\xx{a,b}\ &\approx\ \xx{a',b}\xx{a,a'} \\    
      \label{rel:rename2} \xx{b,b'}\xx{a,b}\ &\approx\ \xx{a,b'}\xx{b,b'} \\    
      \label{rel:rename3} \xx{a,b}\mone{b}\ &\approx\ \mone{a}\xx{a,b} \\    
      \label{rel:rename4} \xx{a,a'}\hh{a,b,c,d}\ &\approx\ \hh{a',b,c,d}\xx{a,a'} \\    
      \label{rel:rename5} \xx{b,b'}\hh{a,b,c,d}\ &\approx\ \hh{a,b',c,d}\xx{b,b'} \\    
      \label{rel:rename6} \xx{c,c'}\hh{a,b,c,d}\ &\approx\ \hh{a,b,c',d}\xx{c,c'} \\
      \label{rel:rename7} \xx{d,d'}\hh{a,b,c,d}\ &\approx\ \hh{a,b,c,d'}\xx{d,d'} \\
      \notag \\
      \resetparent
      \label{rel:ksym1} \xx{a,b}\hh{a,b,c,d}\ &\approx\ \hh{a,b,c,d}\xx{b,d}\mone{b}\mone{d} \\
      \label{rel:swap1} \xx{b,c}\hh{a,b,c,d}\ &\approx\ \mone{a}\hh{a,b,c,d}\mone{a}\hh{a,b,c,d}\mone{a} \\      
      \label{rel:ksym3} \xx{c,d}\hh{a,b,c,d}\ &\approx\ \hh{a,b,c,d}\xx{b,d} \\
      \notag \\
      \resetparent
      \label{rel:kcom1} \hh{a,b,c,d}\hh{b,d,e,f}\ &\approx\ \hh{c,d,e,f}\hh{a,b,c,e} \\
      \notag \\
      \resetparent      
      \begin{split} \label{rel:x}
      \mone{a}\mone{e}\xx{a,e}\hh{e,f,g,h}\hh{a,b,c,d}&\xx{d,e}\hh{a,b,c,d}\hh{e,f,g,h}\xx{a,e}\mone{a}\mone{e} \\
      &\approx\ \\
      \hh{e,f,g,h}\hh{a,b,c,d}&\xx{d,e}\hh{a,b,c,d}\hh{e,f,g,h}
      \end{split}
    \end{align}
  \end{subequations}
  \caption{\label{tab:relations}Relations for $\OD$. One should assume
    that the indices are distinct and that the relations are
    well-formed. For example, in \cref{rel:kcom1} we have
    $a<b<c<d<e<f$.}
\end{table}
  
\begin{definition}
  \label{def:synteq}
  The relation $\approx$ on $\gens_n^*$ is the smallest equivalence
  relation on $\gens_n^*$ containing the relations of
  \cref{tab:relations} and such that if $\word{V}\approx \word{V'}$
  and $\word{W}\approx \word{W'}$ then $\word{V}\word{W} \approx
  \word{V'}\word{W'}$. Two words $\word{V}$ and $\word{W}$ such that
  $\word{V}\approx\word{W}$ are said to be \emph{syntactically
  equivalent}.
\end{definition}

The relation $\approx$ is the smallest congruence relation on
$\gens_n^*$ containing the relations of
\cref{tab:relations}. Intuitively, two words are syntactically
equivalent if one word can be rewritten into the other through a
finite number of applications of the relations contained in
\cref{tab:relations}.

We want to show that two words $\word{V}$ and $\word{W}$ are
semantically equivalent if and only if they are syntactically
equivalent. This is achieved by establishing the two implications
below.
\begin{description}
  \item[Soundness:] Let $\word{G}$ and $\word{H}$ be words over
    $\gens_n$. Then $\word{G} \approx \word{H}$ implies $\word{G} \sim
    \word{H}$.
  \item[Completeness:] Let $\word{G}$ and $\word{H}$ be words over
    $\gens_n$. Then $\word{G} \sim \word{H}$ implies $\word{G} \approx
    \word{H}$.
\end{description}
Soundness and completeness together imply that the semantic and
syntactic relations coincide. This yields a presentation of $\OD$ by
generators and relations. We prove soundness in \cref{ssec:soundness}
and completeness in \cref{ssec:completeness}.
 
\subsection{Soundness}
\label{ssec:soundness}

\begin{theorem}[Soundness]
 \label{thm:sound}
  Let $\word{G}$ and $\word{H}$ be words over $\gens_n$. Then
  $\word{G} \approx \word{H}$ implies $\word{G} \sim \word{H}$.
\end{theorem}

\begin{proof}
  It suffices to show that the relations in \cref{tab:relations} are
  sound. This can be verified by direct computation.
\end{proof}

\subsection{Completeness}
\label{ssec:completeness}

\cref{alg:algo} associates a word over $\gens_n$ to every element of
$\OD$. Because the algorithm is deterministic, the word it associates
to an element $M$ of $\OD$ can be viewed as a \emph{normal form} for
$M$. Our strategy to prove completeness is to show that the relations
of \cref{tab:relations} suffice to rewrite an arbitrary word over
$\gens_n$ into its normal form.

\subsubsection{The State Graph}
\label{sssec:stateG}

We start by introducing a useful graph representation for $\OD$. This
graph representation is akin to a Cayley graph for $\OD$ but is
intended to highlight the words produced by \cref{alg:algo}. Recall
that steps 9, 10, and 15 of \cref{alg:algo} produce short words over
$\gens_n$ of the form
\[
\mone{a}, \quad \xx{a,b}\mone{a}^{\tau_a}, \quad \mbox{ and } \quad
\hh{a,b,c,d} \mone{a}^{\tau_a} \mone{b}^{\tau_b}
\mone{c}^{\tau_c}\mone{d}^{\tau_d}
\]
for appropriately chosen $a,b,c,d$ and $\tau_a, \tau_b, \tau_c,
\tau_d$. We refer to these words as \emph{syllables}.

\begin{definition}
  \label{def:stateg}
  The \emph{state graph} is the directed graph whose vertices and
  edges are defined as follows.
  \begin{itemize}
  \item The vertices are the elements of $\OD$ and are referred to as
    \emph{states}.
  \item There are two types of edges:
    \begin{itemize}
    \item \emph{simple edges}, which are triples $\ang{s', G, s}$
      where $s,s'\in\OD$, $G\in\gens_n$ and $s' = G s$;
    \item \emph{normal edges}, which are triples $\ang{s^\prime, N,
      s}$ where $s,s'\in\OD$, $N$ is the unique first syllable output
      by \cref{alg:algo} on input $s$, and $s^\prime = N s$.
    \end{itemize}
  \end{itemize}
\end{definition}

We denote the edge $\ang{s', G, s}$ by $s \xrightarrow{G} s^\prime$ or
$G:s\to s'$. We use a double line to indicate that an edge is normal,
as in $N:s\too s'$. When the source and target of an edge $\ang{s', G,
  s}$ are clear from context we sometimes simply denote the edge by
$G$.

We note that for every state $s \neq I$, there exists a unique normal
edge originating at $s$. Moreover, if $N:s \too s'$ is normal, then
$\level(s') < \level(s)$. As a result, for every state $s \neq I$,
there exists a unique sequence of normal edges from $s$ to $I$.

\begin{definition}
  \label{def:levels}
  Let $\word{G}$ be the following sequence of simple edges
  \[
  \word{G}= s_0 \xrightarrow{G_1} s_1 \ldots s_{n-1}
  \xrightarrow{G_n}s_n.
  \]
  The \emph{level} of $\word{G}$, denoted $\level(\word{G})$, is the
  maximum of the levels of the states $s_0, \ldots, s_n$. That is,
  $\level(\word{G}) = \max\s{\level(s_0), \ldots, \level(s_n)}$.
\end{definition}

Intuitively, the level of a sequence of edges is the largest level
reached by a state along that sequence.

\begin{definition}
  Let $\word{G},\word{G}':s\to t$ be two sequences of edges. We say
  that the diagram
  \[
    \begin{tikzcd}[column sep=large, row sep=large]
      s \arrow[r, bend right=30, "\word{G}'",swap] \arrow[r, bend right=-30, "\word{G}"] & t
    \end{tikzcd}
  \]
  \emph{commutes equationally} if $\word{G}\approx\word{G}'$.
\end{definition}  

\subsubsection{The Main Lemma and the Proof of Completeness}
\label{sssec:mainlemcomp}

\begin{lemma}[Main Lemma]
\label{lem:main}
  Let $s$, $t$, and $r$ be states, $N:s\too t$ be a normal edge, and
  $G:s\to r$ be a simple edge. Then there exist a state $q$, a
  sequence of normal edges $\word{N}':r\too q$, and a sequence of
  simple edges $\word{G}':t\to q$ such that the diagram
  \[
    \begin{tikzcd}[column sep=large, row sep=large]
      s \arrow[d, Rightarrow, "N" left] 
        \arrow[r, "G"] & 
        r\arrow[d, Rightarrow, "\word{N}'"] \\
      t\arrow[r, "\word{G}'" below] & q
    \end{tikzcd}
  \]
  commutes equationally and $\level(\word{G}')<\level(s)$.
\end{lemma}

The proof of \cref{lem:main} is a very long case distinction which can be
found in \cref{app:main}. We now show how \cref{lem:main} can be used
to derive completeness.

\begin{lemma}
  \label{lem:NormalCommute}
  Let $G:s\to r$ be a simple edge, $\word{N}:s\too I$ be the unique
  sequence of normal edges from $s$ to $I$, and $\word{M}:r \too I$ be
  the unique sequence of normal edges from $r$ to $I$. Then $\word{M}G
  \approx \word{N}$.
\end{lemma}

\begin{proof}
  We proceed by induction on the level of $s$. When $\level(s) =
  (0,0,0)$, then $s = I$ and $\word{N} = \epsilon$. In this case,
  $r=G$ so that $\word{M}=G$ and $\word{M}G\approx\word{N}$ by
  relations \eqref{rel:orderx}, \eqref{rel:ordermone}, or
  \eqref{rel:orderk}. Now suppose that $\level(s) > (0,0,0)$. Then $s
  \neq I$, so that $\word{N}$ can be written as $\word{N} =
  \word{N}'N$ where $N:s\too t_0$ is a normal edge and
  $\word{N}':t_0\too I$ is a sequence of normal edges. By
  \cref{lem:main}, there exist a state $t_k$, a sequence of normal edges
  $\word{L}:r\Rightarrow t_k$, and a sequence of simple edges
  $\word{G}': t_0 \rightarrow t_k$ such that $\word{L}G \approx
  \word{G}'N$, $\level(\word{G}') < \level(s)$, and
  $\word{M}=\word{L}'\word{L}$ for some sequence of normal edges
  $\word{L}'$. Write the sequence $\word{G}'$ as $\word{G}'=G_k\ldots
  G_1$, where $G_\ell:t_{\ell-1}\to t_\ell$ is a simple edge for
  $1\leq \ell \leq k$. For each $\ell$, let $\word{N}_\ell: t_\ell
  \too I$ be the unique sequence of normal edges from $t_\ell$ to
  $I$. Note that, by uniqueness, $\word{N}_k=\word{L}'$. Since
  $\level(t_\ell) < \level(s)$, then, by the induction hypothesis,
  $\word{N}_\ell G_\ell \approx \word{N}_{\ell-1}$. Thus, since
  $\approx$ is a congruence relation, we get $\word{N}' \approx
  \word{N}_k\word{G}'$. Hence, $\word{N}'N \approx \word{N}_k
  \word{G}'N = \word{L}'\word{G}'N\approx \word{L}'\word{L}G$ and, by
  the uniqueness of normal edges, we conclude that $\word{N} =
  \word{N}'N \approx \word{N}_k MG = \word{M}G$.
\end{proof}

\begin{lemma}
  \label{lem:BasicNormalEquiv}
  Let $\word{G}: s\to I$ be any sequence of simple edges with final
  state $I$ and $\word{N}:s \too I$ be the unique sequence of normal
  edges from $s$ to $I$. Then $\word{G} \approx \word{N}$.
\end{lemma}

\begin{proof} 
  We proceed by induction on the length of $\word{G}$. When $\word{G}
  = \epsilon$, then $s = I$ and $\word{N} = \epsilon$. Thus, in the
  base case, we have $\word{G} \approx \word{N}$. Now suppose that
  there is a state $r$ such $\word{G} = \word{G}'G$ for some simple
  edge $G:s\to r$ and some sequence of simple edges $\word{G}':r \to
  I$. Let $\word{M}:r\too I$ be the unique sequence of normal edges
  from $r$ to $I$. By the induction hypothesis, we have
  $\word{G}'\approx \word{M}$, and, by \cref{lem:NormalCommute},
  $\word{M}G \approx \word{N}$. It follows, since $\approx$ is a
  congruence relation, that $\word{G}'G \approx \word{M}G$. Thus
  $\word{G} \approx \word{N}$.
\end{proof}

\begin{theorem}[Completeness]
  \label{thm:completeness}
  Let $\word{G}$ and $\word{H}$ be words over $\gens_n$. Then
  $\word{G} \sim \word{H}$ implies $\word{G} \approx \word{H}$.
\end{theorem}

\begin{proof}
  Since $\word{G} \sim \word{H}$, we have $\interp{\word{G}} =
  \interp{\word{H}}$. Let $s = \interp{\word{G}}^{-1} =
  \interp{\word{H}}^{-1}$ and let $\word{N}:s\too I$ be the unique
  sequence of normal edges from $s$ to $I$. By
  \cref{lem:BasicNormalEquiv}, $\word{G} \approx \word{N}$ and
  $\word{H}\approx \word{N}$ so that, since $\approx$ is an
  equivalence relation, $\word{G} \approx \word{H}$.
\end{proof}

\section{Orthogonal Scaled Dyadic Matrices}
\label{sec:supin}

As discussed in \cref{sec:intro}, the elements of $\OD$ correspond
exactly to quantum circuits over the gate set $\{CCX,H \otimes
H\}$. Replacing the $H \otimes H$ gate with the $H$ gate results in a
more familiar gate set. In this final section, we give a presentation
of the corresponding matrix group.

\begin{definition}
  \label{def:scaledmat}
  The $n$-dimensional group of \emph{orthogonal scaled dyadic
  matrices} consists of the $n\times n$ orthogonal matrices of the
  form $M/\sqrt{2}{}^k$, where $M$ is an integer matrix and $k$ is a
  nonnegative integer. It is denoted $\supin_n$.
\end{definition}

The notions of denominator exponent and least denominator exponent, as
introduced for dyadic matrices in \cref{sec:gens}, also apply to
scaled dyadic matrices. For elements of $\supin_n$, however, one
should consider powers of $1/\sqrt{2}$, rather than powers of $1/2$. As
a result, in this final section, (least) denominator exponents are
considered with respect to $1/\sqrt{2}$.
 
Note that $\OD\subseteq\supin_n$. It is known from
\cite[Lemma~5.9]{AGR2019} that $\supin_n=\OD$ when $n$ is odd. When
$n$ is even, $\OD$ is a proper subgroup of $\supin_n$ of index 2. As a
consequence, we focus on the case of even $n$ in what follows.

To obtain a set of generators for $\supin_n$ when $n$ is even, it
suffices to add $I_{n/2}\otimes H$ to $\gens_n$, where $I_{n/2}\otimes
H$ is the $n\times n$ block-diagonal matrix
\[
I_{n/2}\otimes H = \diag(H,\ldots,H).
\]
For simplicity, when $n$ is clear from context, we write $I\otimes H$
for $I_{n/2}\otimes H$. Note that, unlike the other generators,
$I\otimes H$ is a global matrix which acts non-trivially on entries of
a vector or matrix.

\begin{definition}
  Let $n$ be even. The set of \emph{$n$-dimensional generators} is the
  subset of $\supin_n$ defined as
  \[
    \gensup_n=\s{\mone{a},\xx{a,b},\hh{a,b,c,d},I\otimes H \mid 1\leq
      a,b,c,d \leq n}.
  \]
\end{definition}

The relation of semantic equivalence is defined on $\gensup_n^*$ as in
\cref{def:semeq}. We adapt the relation of syntactic equivalence on
$\gensup_n^*$ by adding further relations to account for the
additional generator.

\begin{table}
  \begin{subequations} 
    \renewcommand{\theparentequation}{\arabic{parentequation}}%
    \begin{align}
      \label{itm:relh0} (I\otimes H)^2\ &\approx\ \epsilon \\
      \label{itm:relh5} (I\otimes H)\hh{1,2,3,4}(I\otimes H)\ &\approx\ \hh{1,2,3,4}\\
      \label{itm:relh1} (I\otimes H) \mone{1} (I\otimes H)\ &\approx\ \mone{1}\xx{1,2}\mone{1} \\
      \label{itm:relh3} (I\otimes H) \xx{a,a+1} (I\otimes H)\ &\approx\ \mone{a+1}^{a+1}\xx{a,a+1}^{a}\hh{a-1,a,a+1,a+2}^a 
    \end{align}
  \end{subequations}
  \caption{\label{tab:relationssup}Relations for $\supin_n$}
\end{table}

\begin{definition}
  \label{def:synteqsup}
  The relation $\approx$ on $\gensup_n^*$ is the smallest equivalence
  relation on $\gensup_n^*$ containing the relations of
  \cref{tab:relations,tab:relationssup} and such that if
  $\word{V}\approx \word{V'}$ and $\word{W}\approx \word{W'}$ then
  $\word{V}\word{W} \approx \word{V'}\word{W'}$. Two words $\word{V}$
  and $\word{W}$ such that $\word{V}\approx\word{W}$ are said to be
  \emph{syntactically equivalent}.
\end{definition}

To obtain a presentation of $\supin_n$, we establish soundness and
completeness. As for $\OD$, soundness is proved by
computation and is therefore stated without proof. For completeness,
we leverage \cref{thm:completeness}.

\begin{theorem}[Soundness]
 \label{thm:soundsup}
  Let $n$ be even. Let $\word{G}$ and $\word{H}$ be words over
  $\gensup_n$. Then $\word{G} \approx \word{H}$ implies $\word{G} \sim
  \word{H}$.
\end{theorem}

\begin{lemma}
  \label{lem:supcomp}
  Let $n$ be even. For every word $\word{G}$ over $\gens_n$ there
  exists a word $\word{G}'$ over $\gens_n$ such that $\left(I \otimes
  H\right) \word{G} \approx \word{G}' \left(I\otimes H\right)$.
\end{lemma}

\begin{proof}
  By \cref{lem:simp} and \cref{thm:completeness}, every word in
  $\gens_n^*$ is syntactically equivalent to one that uses only
  $\mone{1}$, $\hh{1,2,3,4}$ and $\xx{a,a+1}$. The claim then follows
  from the relations in \cref{tab:relationssup}.
\end{proof}

\begin{corollary}
  \label{cor:supcomp}
  Let $n$ be even and let $\word{G} \in \gensup_n^*$. If the least
  denominator exponent of $\interp{\word{G}}$ is even, there exists
  $\word{G}' \in \gens_n^*$ such that $\word{G} \approx \word{G}'$. If
  the least denominator of $\interp{\word{G}}$ is odd, there exists
  $\word{G}' \in \gens_n^*$ such that $\word{G} \approx
  \word{G}'(I\otimes H)$.
\end{corollary}

\begin{proof}
  Let $k$ be the least denominator exponent of $\interp{\word{G}}$
  (with respect to $1/\sqrt{2}$). Through repeated application of
  \cref{lem:supcomp}, we can push all of the occurrence of $I\otimes
  H$ in $\word{G}$ to the right in order to rewrite $\word{G}$ as
  $\word{G}'(I\otimes H)^\ell$ for some $\ell\in\N$ such that
  $\ell\equiv k\pmod{2}$. The result then follows from
  \cref{itm:relh0}.
\end{proof}

\begin{theorem}[Completeness]
  \label{thm:completenesssup}
  Let $\word{G}$ and $\word{H}$ be words over $\gensup_n$. Then
  $\word{G} \sim \word{H}$ implies $\word{G} \approx \word{H}$.
\end{theorem}

\begin{proof}
  Let $k=\lde(\interp{\word{G}})=\lde(\interp{\word{H}})$. If $k$ is
  even, by \cref{cor:supcomp}, $\word{G}\approx \word{G}'$ and
  $\word{H} \approx\word{H}'$ for some
  $\word{G}',\word{H}'\in\gens_n^*$. Thus $\word{G}'\sim \word{H}'$
  and by \cref{thm:completeness} $\word{G}'\approx\word{H}'$. Hence,
  $\word{G} \approx \word{H}$. If $k$ is odd, by \cref{cor:supcomp},
  $\word{G}\approx \word{G}'(I\otimes H)$ and $\word{H}
  \approx\word{H}'(I\otimes H)$ for some
  $\word{G}',\word{H}'\in\gens_n^*$. Thus $\word{G}'\sim \word{H}'$
  and by \cref{thm:completeness} $\word{G}'\approx\word{H}'$. Hence,
  $\word{G} \approx \word{H}$.
\end{proof}

\section{Conclusion}
\label{sec:conclusion}
In this paper, we gave a finite presentation of the groups $\OD$ and
$\supin_n$, which arise in the study of so-called restricted
Clifford+$T$ circuits. A natural extension of this work is to study
the matrix groups which correspond to alternative restrictions of the
Clifford+$T$ gate set. Another avenue for future research is to
interpret the relations of \cref{tab:relations,tab:relationssup} as
relations between quantum circuits and to use them to optimize
restricted Clifford+$T$ circuits.

\appendix

\section{Proof of the Main Lemma}
\label{app:main}

This appendix contains a proof of the Main Lemma (\cref{lem:main}). We
first record some important properties of $\hh{a,b,c,d}$ in
\cref{apps:propk}. Then, in \cref{apps:derivedrels}, we introduce
derived relations which are helpful in establishing that certain
diagrams commute. In \cref{apps:edges}, we distinguish between
\emph{simple edges} and \emph{basic edges} in order to simplify the
proof of \cref{lem:main}. The proof of the Main Lemma, a long case
distinction, can be found in \cref{apps:main}.

\subsection{Properties of \texorpdfstring{$\hh{a,b,c,d}$}{K[a,b,c,d]}}
\label{apps:propk}

We start by recording a few useful properties of $\hh{a,b,c,d}$. To
this end, it will be useful to consider the vector of residues
associated to a vector of integers. For brevity, we will sometimes
write $u\equiv r_1 \cdots r_n \pmod{m}$ if $u_i\equiv r_i \pmod{m}$
for $1\leq i \leq n$.

Let $u\in\Z^4$ and define the vectors $v$ and $w$ as
\[
v =
\begin{bmatrix}
  u_1+u_2+u_3+u_4 \\
  u_1-u_2+u_3-u_4 \\
  u_1+u_2-u_3-u_4 \\
  u_1-u_2-u_3+u_4     
\end{bmatrix}
\]
and $w=v/2$. Then $w=\hh{1,2,3,4}u$. Note that while $v\in\Z^4$, for
$w$ we have $w\in\Z^4$ or $w\in\D{}^4$

\begin{lemma}
  \label{applem:honemod4}
  Let $u\in \Z^4$ and suppose that $u_1 + u_2 + u_3 + u_4 \equiv 0
  \pmod{2}$. Then $\hh{1,2,3,4} u = w$ for some $w \in \Z^4$.
\end{lemma}

\begin{proof}
  Write $v$ as above. Then, since $u_1+ u_2 + u_3 + u_4\equiv
  0 \pmod{2}$ and $u_i\equiv -u_i\pmod{2}$, we have $v_i\equiv
  0 \pmod{2}$. The result then follows by setting $v_i=2w_i$ and
  noting that $\hh{1,2,3,4}u=v/2=w$.
\end{proof}

\begin{lemma}
  \label{applem:evenoddsOddodds} Let $u\in \Z^4$ and suppose that
  $u\equiv 1111 \pmod{2}$. Then
  \begin{itemize}

  \item if the number of entries in $u$ that are congruent to 1 modulo
  4 is even, then $\hh{1,2,3,4} u =w$ for some $w\in \Z^4$ such that
  $w\equiv 0000 \pmod{2}$, and

  \item if the number of entries in $u$ that are congruent to 1 modulo
  4 is odd, then $\hh{1,2,3,4} u =w$ for some $w\in \Z^4$ such that
  $w\equiv 1111 \pmod{2}$ and the number of entries in $w$ that are
  congruent to 1 modulo 4 is odd.

  \end{itemize}
\end{lemma}

\begin{proof}
We know from \cref{applem:honemod4} that $w\in\Z^4$. Now write $v$ as
above. It can then be verified that if there are evenly many $u_i$
such that $u_i\equiv 1\pmod{4}$, then $v\equiv 0000 \pmod{4}$, so that
$w\equiv 0000 \pmod{2}$. Similarly, if there are oddly many $u_i$ such
that $u_i\equiv 1\pmod{4}$, then $v\equiv 2222 \pmod{4}$, so that
$w\equiv 1111 \pmod{2}$.

Finally, suppose that $u\equiv 1111 \pmod{2}$, that the number of
$u_i\equiv 1 \pmod{4}$ is odd, and that the number of $w_i\equiv
1 \pmod{4}$ is even. Then by the first part of the lemma we have
$\hh{1,2,3,4}w\equiv 0000 \pmod{2}$. But this is a contradiction since
\[
\hh{1,2,3,4} w = \hh{1,2,3,4}\hh{1,2,3,4}u = u
\]
and $u\equiv 1111 \pmod{2}$ by assumption.
\end{proof}

\begin{lemma}
  \label{prop:kevenodds}
  Let $u\in \Z^4$ and suppose that $u^\intercal u\equiv
  2 \pmod{4}$. Then $u$ has exactly two odd entries and $\hh{1,2,3,4}
  u =w$ for some $w\in \Z^4$. Moreover,
  \begin{itemize}
  \item if $u\equiv 1100 \pmod{2}$ then $w\equiv 1010 \pmod{2}$ or
  $w\equiv 0101 \pmod{2}$,
  \item if $u\equiv 1010 \pmod{2}$ then $w\equiv 1100 \pmod{2}$ or
  $w\equiv 0011 \pmod{2}$,
  \item if $u\equiv 1001 \pmod{2}$ then $w\equiv 1001 \pmod{2}$ or
  $w\equiv 0110 \pmod{2}$,
  \item if $u\equiv 0110 \pmod{2}$ then $w\equiv 1001 \pmod{2}$ or
  $w\equiv 0110 \pmod{2}$,
  \item if $u\equiv 0101 \pmod{2}$ then $w\equiv 1100 \pmod{2}$ or
  $w\equiv 0011 \pmod{2}$, and
  \item if $u\equiv 0011 \pmod{2}$ then $w\equiv 1010 \pmod{2}$ or
  $w\equiv 0101 \pmod{2}$.
  \end{itemize}
\end{lemma}

\begin{proof}
Since $u^\intercal u\equiv 2 \pmod{4}$, $u$ has exactly two odd
entries. Thus, by \cref{applem:honemod4}, $\hh{1,2,3,4}u=w$ for some
$w\in\Z^4$. Now suppose that $u\equiv 1100\pmod{2}$. Then $u_1\equiv
u_2\equiv 1 \pmod{2}$ and $u_3\equiv u_4\equiv 0 \pmod{2}$. Note that
$(\pm u_3) + (\pm u_4) \equiv 2u_3 \pmod{4}$. If $u_1\equiv
u_2 \pmod{4}$ we get
\[
  v =
  \begin{bmatrix}
    u_1+u_2+u_3+u_4 \\
    u_1-u_2+u_3-u_4 \\
    u_1+u_2-u_3-u_4 \\
    u_1-u_2-u_3+u_4     
  \end{bmatrix} =
  \begin{bmatrix}
    2u_1+2u_3 \\
    2u_3 \\
    2u_1+2u_3 \\
    2u_3 \\
  \end{bmatrix}
\]  
so that $v\equiv 2020 \pmod{4}$ and $w=v/2\equiv 1010 \pmod{2}$. And
if $u_1\not\equiv u_2 \pmod{4}$ we get
\[
  v =
  \begin{bmatrix}
    u_1+u_2+u_3+u_4 \\
    u_1-u_2+u_3-u_4 \\
    u_1+u_2-u_3-u_4 \\
    u_1-u_2-u_3+u_4     
  \end{bmatrix} =
  \begin{bmatrix}
    2u_3 \\
    2u_1+2u_3 \\
    2u_3 \\
    2u_1+2u_3 \\
  \end{bmatrix}
\]
so that $v\equiv 0202 \pmod{4}$ and $w=v/2\equiv 1010 \pmod{2}$. The
remaining cases are proved similarly.
\end{proof}

\begin{lemma}
  \label{prop:koddodds}  
  Let $u\in \Z^4$ and suppose that $u^\intercal u\equiv
  1 \pmod{2}$. Then $u$ has exactly one or three odd entries and
  $\hh{1,2,3,4} u =w$ for some $w\not\in \Z^4$. Moreover, for
  $v=2w\in\Z^4$, we have \begin{itemize}
  
  \item if $u\equiv 1000 \pmod{2}$ or $u\equiv 0111 \pmod{2}$ then
  $v\equiv 1111 \pmod{4}$ or $v\equiv 3333 \pmod{4}$,

  \item if $u\equiv 0100 \pmod{2}$ or $u\equiv 1011 \pmod{2}$ then
  $v\equiv 1313 \pmod{4}$ or $v\equiv 3131 \pmod{4}$,

  \item if $u\equiv 0010 \pmod{2}$ or $u\equiv 1101 \pmod{2}$ then
  $v\equiv 1133 \pmod{4}$ or $v\equiv 3311 \pmod{4}$, and

  \item if $u\equiv 0001 \pmod{2}$ or $u\equiv 1110\pmod{2}$ then
  $v\equiv 1331 \pmod{4}$ or $v\equiv 3113 \pmod{4}$.

  \end{itemize}
\end{lemma}

\begin{proof}
Since $u^\intercal u\equiv 1 \pmod{2}$, $u$ has oddly many odd
entries. Writing $v$ and $w$ as above, we see that $v\equiv
1111 \pmod{2}$ so that $w\not\in \Z^4$.

Now, if $u\equiv 1000 \pmod{2}$, then
\[
(\pm u_2) + (\pm u_3) + (\pm u_4) \equiv 3u_2 \pmod{4}.
\]
Hence, we either have $v \equiv 1111 \pmod{4}$ when $3u_2\equiv
0 \pmod{4}$ or $v \equiv 3333 \pmod{4}$ when $3u_2\equiv
2 \pmod{4}$. This proves the first item. The remaining items are
proved similarly.
\end{proof}

\begin{lemma}
  \label{applem:twohsone} Let $u\in \Z^4$ and suppose that $u\equiv
  1111 \pmod{4}$. Then $\hh{1,2,3,4} u =2w'$ for some $w'\in \Z^4$ such
  that $w'\equiv 1000 \pmod{2}$ or $w'\equiv 0111 \pmod{2}$.
\end{lemma}

\begin{proof}
  Let $v\in\Z^4$ be defined as above. Since $u_i \equiv 1 \pmod{4}$,
  we have $v_i\equiv 0 \pmod{4}$. Moreover, $u_i \equiv 1 \pmod{4}$
  also implies that $u_2+u_4\equiv 2 \pmod{4}$, so that $u_2+u_4\equiv
  -(u_2+u_4) \pmod{4}$, and thus that $u_2+u_4\equiv -(u_2+u_4)
  +4 \pmod{8}$. As a result, $v_2 \equiv v_1 +4 \pmod{8}$. Reasoning
  similarly we find that $v_3 \equiv v_1 +4 \pmod{8}$ and that
  $v_4 \equiv v_1 +4 \pmod{8}$. The result then follows by setting
  $v_i=4w_i'$ and noting that $\hh{1,2,3,4}u=v/2=2w'$.
\end{proof}

\begin{lemma}
  \label{prop:kevens1}
  Let $u\in \Z^4$ and suppose that $u\equiv 0000 \pmod{2}$ and that
  $u^\intercal u \equiv 0 \pmod{8}$. Then 
  \[
  \hh{1,2,3,4} u =w
  \]
  for some
  $w\in \Z^4$ such that $w\equiv 0000 \pmod{2}$.
\end{lemma}

\begin{proof}
We have $u^\intercal u \equiv 0 \pmod{8}$. Since the square of an even
integer is congruent to 0 or 4 modulo 8 there must be evenly many
$u_i$ such that $u_i^2\equiv 4 \pmod{8}$. Therefore, there must be
evenly many $u_i$ such that $u_i\equiv 2 \pmod{4}$. The result then
follows by computation, as in the proof of \cref{applem:honemod4}.
\end{proof}

\begin{lemma}
  \label{prop:kevens2}
  Let $u\in \Z^4$ and suppose that $u\equiv 0000 \pmod{2}$ and that
  $u^\intercal u \equiv 4 \pmod{8}$. Then 
  \[
  \hh{1,2,3,4} u =w
  \]
  for some
  $w\in \Z^4$ such that $w\equiv 1111 \pmod{2}$. Moreover, evenly many
  of the entries of $w$ are congruent to 1 modulo 4.
\end{lemma}

\begin{proof}
The first statement is shown as in \cref{prop:kevens1}. For the second
statement, suppose that oddly many of the entries of $w$ were
congruent to 1 modulo 4. Then $w_1+w_2+w_3+w_4\equiv 2 \pmod{4}$. Then
$(w_1+w_2+w_3+w_4)/2\equiv 1 \pmod{2}$. But this is a contradiction
since $(w_1+w_2+w_3+w_4)/2=v_1$ and $v_1\equiv 0 \pmod{2}$ by
assumption.
\end{proof}

\begin{lemma}
  \label{applem:norm}
  Let $u \in \Z^8$ and suppose that $u\equiv 11111111 \pmod{2}$. Then
  either $u^\intercal u\equiv 0 \pmod{16}$ or $u^\intercal u\equiv 8
  \pmod{16}$.
\end{lemma}

\begin{proof}
Since the square of an odd integer is either 1 or 9 modulo 16, then
$u^\intercal u \equiv x + 9y \pmod{16}$ where $x$ is the number of
entries whose square is congruent to 1 and $y$ is the number of
entries whose square is congruent to 9. But $x+y=8$, so that
$u^\intercal u\equiv 0 \pmod{16}$ or $u^\intercal u \equiv 8
\pmod{16}$ as desired.
\end{proof}

\begin{lemma}
  \label{applem:normresidue1}
  Let $u \in \Z^8$ and suppose that $u\equiv 11111111 \pmod{4}$.  If
  $u^\intercal u\equiv 0 \pmod{16}$ then
  \[
  \hh{1,2,3,4}\hh{5,6,7,8}u=2w
  \]
  for some $w\in\Z^8$ such that $
  w\equiv 10000111 \pmod{2}$ or $w\equiv 01111000 \pmod{2}$.
\end{lemma}

\begin{proof}
We know by \cref{applem:twohsone} that $\hh{1,2,3,4}\hh{5,6,7,8}u=2w$ for
some $w\in\Z^8$ such that the vector of residues of $w$ is one of
\[
10001000, \quad 10000111, \quad 01111000, \quad \mbox{ or } \quad 01110111.
\]
But, since $K$ is orthogonal and $u^\intercal u\equiv 0 \pmod{16}$, we
have $4(w^\intercal w) \equiv u^\intercal u \equiv 0\pmod{16}$ and
therefore $w^\intercal w\equiv 0 \pmod{4}$ so that $ w\equiv 10000111
\pmod{2}$ or $w\equiv 01111000 \pmod{2}$ as claimed.
\end{proof}

\begin{lemma}
  \label{applem:normresidue2}
  Let $u \in \Z^8$ and suppose that $u\equiv 11111111 \pmod{4}$.  If
  $u^\intercal u\equiv 8 \pmod{16}$ then
  \[
  \hh{1,2,3,4}\hh{5,6,7,8}u=2w
  \]
  for some $w\in\Z^8$ such that $
  w\equiv 10001000 \pmod{2}$ or $w\equiv 01110111 \pmod{2}$.
\end{lemma}

\begin{proof}
  Similar to the proof of \cref{applem:normresidue1}.
\end{proof}

\subsection{Derived Relations}
\label{apps:derivedrels}

In this section, we show that certain convenient relations can be
derived from the relations given in \cref{tab:relations}. In the
derivations, we sometimes use certain relations implicitly: we remove
adjacent pairs of identical generators using
\cref{rel:orderx,rel:ordermone,rel:orderk}, we commute generators
acting on distinct indices using
\cref{rel:disjoint1,rel:disjoint2,rel:disjoint3,rel:disjoint4,rel:disjoint5,rel:disjoint6},
and we change indices using
\cref{rel:rename1,rel:rename2,rel:rename3,rel:rename4,rel:rename5,rel:rename6,rel:rename7}.

\begin{proposition}
  \label{derivedrelsxk}
  The relations below are derivable.
  \begin{subequations}
    \begin{align}
      \label{rel:ksym4} \xx{a,c}\hh{a,b,c,d}\ &\approx\ \hh{a,b,c,d}\xx{c,d}\mone{c}\mone{d} \\
      \label{rel:ksym5} \xx{a,d}\hh{a,b,c,d}\ &\approx\ \hh{a,b,c,d}\xx{b,d}\xx{c,d}\xx{b,d}\mone{b}\mone{c} \\
      \label{rel:ksym2} \xx{b,c}\hh{a,b,c,d}\ &\approx\ \hh{a,b,c,d}\xx{b,c} \\      
      \label{rel:ksym6} \xx{b,d}\hh{a,b,c,d}\ &\approx\ \hh{a,b,c,d}\xx{c,d}
    \end{align}
  \end{subequations}
\end{proposition}

\begin{proof}
  For \cref{rel:ksym2}, using \cref{rel:swap1}, we have
  \begin{align*}
    \xx{b,c}\hh{a,b,c,d}\ &\approx\ \mone{a}\hh{a,b,c,d}\mone{a}\hh{a,b,c,d}\mone{a}\hh{a,b,c,d}\hh{a,b,c,d}\\
    &\approx\ \mone{a}\hh{a,b,c,d}\mone{a}\hh{a,b,c,d}\mone{a} \\
    &\approx\ \hh{a,b,c,d}\hh{a,b,c,d}\mone{a}\hh{a,b,c,d}\mone{a}\hh{a,b,c,d}\mone{a} \\
    &\approx\ \hh{a,b,c,d}\xx{b,c}.
  \end{align*}
  For \cref{rel:ksym4}, using \cref{rel:ksym2,rel:ksym1}, we have
  \begin{align*}
    \xx{a,c}\hh{a,b,c,d}\ &\approx\ \xx{b,c}\xx{a,b}\xx{b,c}\hh{a,b,c,d} \\
    &\approx\ \xx{b,c}\xx{a,b}\hh{a,b,c,d}\xx{b,c} \\
    &\approx\ \xx{b,c}\hh{a,b,c,d}\xx{b,d}\mone{b}\mone{d}\xx{b,c} \\
    &\approx\ \hh{a,b,c,d}\xx{b,c}\xx{b,d}\mone{b}\mone{d}\xx{b,c} \\
    &\approx\ \hh{a,b,c,d}\xx{c,d}\mone{c}\mone{d}.
  \end{align*}
  For \cref{rel:ksym5}, using \cref{rel:ksym4,rel:ksym3}, we have
  \begin{align*}
    \xx{a,d}\hh{a,b,c,d}\ &\approx\ \xx{c,d}\xx{a,c}\xx{c,d}\hh{a,b,c,d} \\
    &\approx\ \xx{c,d}\xx{a,c}\hh{a,b,c,d}\xx{b,d} \\
    &\approx\ \xx{c,d}\hh{a,b,c,d}\xx{c,d}\mone{c}\mone{d}\xx{b,d} \\
    &\approx\ \hh{a,b,c,d}\xx{b,d}\xx{c,d}\mone{c}\mone{d}\xx{b,d} \\
    &\approx\ \hh{a,b,c,d}\xx{b,d}\xx{c,d}\xx{b,d}\mone{b}\mone{c}.
  \end{align*}
  Finally, \cref{rel:ksym6} is the adjoint of \cref{rel:ksym3}.
\end{proof}

Along with \cref{rel:ksym1,rel:ksym2}, the relations
of \cref{derivedrelsxk} will allow us to move an $x$ generator passed
a $K$ generator when the $X$ generator acts on two of the indices on
which the $K$ generator acts. The next proposition shows how to move
evenly many occurrences of a $(-1)$ generator passed a $K$ generator.

\begin{proposition}
  \label{prop:derivedrels1}
  The relations below are derivable.
  \begin{subequations}
    \begin{align}
      \label{rel:k12} \mone{a}\mone{b}\hh{a,b,c,d} &\approx\ \hh{a,b,c,d}\xx{a,c}\xx{b,d}\mone{a}\mone{b}\mone{c}\mone{d}\\
      \label{rel:k22} \mone{a}\mone{c}\hh{a,b,c,d} &\approx\ \hh{a,b,c,d}\xx{a,b}\xx{c,d}\mone{a}\mone{b}\mone{c}\mone{d}\\
       \label{rel:k32} \mone{a}\mone{d}\hh{a,b,c,d} &\approx\ \hh{a,b,c,d}\xx{a,d}\xx{b,c}\mone{a}\mone{b}\mone{c}\mone{d}\\
      \label{rel:k31} \mone{b}\mone{c}\hh{a,b,c,d} &\approx\ \hh{a,b,c,d}\xx{a,d}\xx{b,c}\\
      \label{rel:k21} \mone{b}\mone{d}\hh{a,b,c,d} &\approx\ \hh{a,b,c,d}\xx{a,b}\xx{c,d}\\
      \label{rel:k11} \mone{c}\mone{d}\hh{a,b,c,d} &\approx\ \hh{a,b,c,d}\xx{a,c}\xx{b,d} \\
      \label{itm:relk4} \mone{a}\mone{b}\mone{c}\mone{d}\hh{a,b,c,d} &\approx\ \hh{a,b,c,d}\mone{a}\mone{b}\mone{c}\mone{d}
    \end{align}
  \end{subequations}
\end{proposition}

\begin{proof}
  For \cref{rel:k21}, using \cref{rel:ksym1,rel:ksym3}, we have
  \begin{align*}
    \mone{b}\mone{d}\hh{a,b,c,d}\ &\approx\  \mone{b}\mone{d}\hh{a,b,c,d}\xx{c,d}\xx{c,d}\\
    &\approx\ \mone{b}\mone{d}\xx{b,d}\hh{a,b,c,d}\xx{c,d}\\
    &\approx\ \hh{a,b,c,d}\xx{a,b}\xx{c,d}.
  \end{align*}
  For \cref{rel:k11}, using \cref{rel:k21,rel:ksym2}, we have
  \begin{align*}
    \mone{c}\mone{d}\hh{a,b,c,d}\ &\approx\  \xx{b,c}\mone{b}\xx{b,c}\mone{d}\hh{a,b,c,d}\\
    &\approx\ \xx{b,c}\mone{b}\mone{d}\hh{a,b,c,d}\xx{b,c}\\
    &\approx\ \xx{b,c}\hh{a,b,c,d}\xx{a,b}\xx{c,d}\xx{b,c}\\
    &\approx\ \hh{a,b,c,d}\xx{b,c}\xx{a,b}\xx{c,d}\xx{b,c}\\
    &\approx\ \hh{a,b,c,d}\xx{b,c}\xx{a,b}\xx{b,c}\xx{b,c}\xx{c,d}\xx{b,c}\\
    &\approx\ \hh{a,b,c,d}\xx{a,c}\xx{b,d}.
  \end{align*}
  For \cref{rel:k31}, using \cref{rel:k11,rel:ksym3}, we have
  \begin{align*}
    \mone{b}\mone{c}\hh{a,b,c,d}\ &\approx\  \mone{c}\xx{b,d}\mone{d}\xx{b,d}\hh{a,b,c,d}\\
    &\approx\ \mone{c}\xx{b,d}\mone{d}\hh{a,b,c,d}\xx{c,d}\\
    &\approx\ \xx{b,d}\mone{c}\mone{d}\hh{a,b,c,d}\xx{c,d}\\
    &\approx\ \xx{b,d}\hh{a,b,c,d}\xx{a,c}\xx{b,d}\xx{c,d}\\
    &\approx\ \hh{a,b,c,d}\xx{c,d}\xx{a,c}\xx{b,d}\xx{c,d}\\
    &\approx\ \hh{a,b,c,d}\xx{c,d}\xx{a,c}\xx{c,d}\xx{c,d}\xx{b,d}\xx{c,d}\\
    &\approx\ \hh{a,b,c,d}\xx{a,d}\xx{b,c}.
  \end{align*}
 For \cref{itm:relk4}, using \cref{rel:k11}, we have
  \begin{align*}
    \mone{a}\mone{b}\mone{c}\mone{d}\hh{a,b,c,d}\ &\approx\ \mone{a}\mone{b}\hh{a,b,c,d}\xx{a,c}\xx{b,d} \\
    &\approx\ \xx{a,c}\xx{b,d}\mone{c}\mone{d}\xx{a,c}\xx{b,d}\hh{a,b,c,d}\xx{a,c}\xx{b,d} \\
    &\approx\ \xx{a,c}\xx{b,d}\mone{c}\mone{d}\hh{a,b,c,d}\mone{c}\mone{d}\xx{a,c}\xx{b,d} \\
    &\approx\ \xx{a,c}\xx{b,d}\hh{a,b,c,d}\xx{a,c}\xx{b,d}\mone{c}\mone{d}\xx{a,c}\xx{b,d} \\
    &\approx\ \xx{a,c}\xx{b,d}\hh{a,b,c,d}\mone{a}\mone{b} \\
    &\approx\ \hh{a,b,c,d}\mone{a}\mone{b}\mone{c}\mone{d}.
  \end{align*}
  For \cref{rel:k22}, using \cref{rel:k21,itm:relk4}  and multiplying the right-hand side by
  \[
  \mone{a}\mone{b}\mone{c}\mone{d}\mone{a}\mone{b}\mone{c}\mone{d}
  \]
  we get
  \begin{align*}
    \mone{a}\mone{c}\hh{a,b,c,d}\ &\approx\ \mone{b}\mone{d}\hh{a,b,c,d}\mone{a}\mone{b}\mone{c}\mone{d}\\
     &\approx\ \hh{a,b,c,d}\xx{a,b}\xx{c,d}\mone{a}\mone{b}\mone{c}\mone{d}.
  \end{align*}
  For \cref{rel:k12}, using \cref{rel:k11,itm:relk4}  and multiplying the right-hand side by
  \[
  \mone{a}\mone{b}\mone{c}\mone{d}\mone{a}\mone{b}\mone{c}\mone{d}
  \]
  we get
  \begin{align*}
    \mone{a}\mone{b}\hh{a,b,c,d}\ &\approx\ \mone{c}\mone{d}\hh{a,b,c,d}\mone{a}\mone{b}\mone{c}\mone{d}\\
     &\approx\ \hh{a,b,c,d}\xx{a,c}\xx{b,d}\mone{a}\mone{b}\mone{c}\mone{d}.
  \end{align*}  
  For \cref{rel:k32}, using \cref{rel:k31,itm:relk4} and multiplying the right-hand side by
  \[
  \mone{a}\mone{b}\mone{c}\mone{d}\mone{a}\mone{b}\mone{c}\mone{d}
  \]
  we get
  \begin{align*}
    \mone{a}\mone{d}\hh{a,b,c,d}\ &\approx\ \mone{b}\mone{c}\hh{a,b,c,d}\mone{a}\mone{b}\mone{c}\mone{d}\\
     &\approx\ \hh{a,b,c,d}\xx{a,d}\xx{b,c}\mone{a}\mone{b}\mone{c}\mone{d}.\qedhere
  \end{align*}
\end{proof}

\begin{corollary}
  \label{prop:kevenk}
  Let $\word{W}$ be a word over $\gens$ of the form
  \[
  \hh{a,b,c,d}\mone{a}^{\tau_a}\mone{b}^{\tau_b}
  \mone{c}^{\tau_c} \mone{d}^{\tau_d}\hh{a,b,c,d}
  \]
  where $\tau_a,\tau_b,\tau_c,\tau_d\in\Z_2$ and evenly many of
  $\tau_a,\tau_b,\tau_c,\tau_d\in\Z_2$ are even. Then there exists a
  word $\word{V}$ over $\s{\mone{x},\xx{y,z}\mid x,y,z\in\s{a,b,c,d}}$ such
  that $\word{V}\approx \word{W}$.
\end{corollary}

\begin{proof}
By \cref{rel:ksym1,rel:ksym3,prop:derivedrels1}.
\end{proof}

\begin{corollary}
  \label{prop:koddkoddk}
  Let $\word{W}$ be a word over $\gens$ of the form
  \[
  \hh{a,b,c,d} \mone{a}^{\tau_a}\mone{b}^{\tau_b} \mone{c}^{\tau_c} \mone{d}^{\tau_d} \hh{a,b,c,d}
  \mone{a}^{\tau_a'}\mone{b}^{\tau_b'} \mone{c}^{\tau_c'} \mone{d}^{\tau_d'} \hh{a,b,c,d}  
  \]
  where
  $\tau_a, \tau_b, \tau_c,\tau_d, \tau_a',\tau_b', \tau_c', \tau_d'\in\Z_2$,
  oddly many of $\tau_a,\tau_b,\tau_c,\tau_d$ are even, and
  oddly many of $\tau_a',\tau_b',\tau_c',\tau_d'$ are
  even. Then there exists a word $\word{V}$ over $\s{\mone{x},\xx{y,z}\mid
  x,y,z\in\s{a,b,c,d}}$ such that $\word{V}\approx \word{W}$.
\end{corollary}

\begin{proof}
First note that by multiplying by
\[
  \mone{a}\mone{b}\mone{c}\mone{d}\mone{a}\mone{b}\mone{c}\mone{d}
\]
if required, we can ensure that exactly one of $\tau_a$, $\tau_b$,
$\tau_c$, or $\tau_d$ is odd, and similarly for for $\tau_a'$,
$\tau_b'$, $\tau_c'$, or $\tau_d'$. Moreover, conjugating the left
occurrence of $(-1)$ by $\xx{a,x}$ and commuting both $X$s, we can
ensure that the left occurrence of $(-1)$ is of the form
$\mone{a}$. Therefore, we can assume without loss of generality that
$\word{W}$ is of the form
\[
\hh{a,b,c,d} \mone{a}\hh{a,b,c,d}\mone{x}\hh{a,b,c,d}  
\]
for some $x\in \s{a,b,c,d}$. If $x=a$ we can conclude
by \cref{rel:swap1}. If $x=b$ we have
\begin{align*}
\begin{split}
\hh{a,b,c,d}& \mone{a}\hh{a,b,c,d}\mone{b}\hh{a,b,c,d} \\
& \approx \hh{a,b,c,d} \mone{a}\hh{a,b,c,d}\mone{b}\xx{a,b}\xx{a,b}\hh{a,b,c,d} \\
& \approx \hh{a,b,c,d} \mone{a}\hh{a,b,c,d}\xx{a,b}\mone{a}\xx{a,b}\hh{a,b,c,d} \\
& \approx \hh{a,b,c,d} \mone{a}\mone{d}\mone{b}\xx{b,d}\hh{a,b,c,d}\mone{a}\hh{a,b,c,d}\xx{b,d}\mone{b}\mone{d} \\
& \approx \xx{a,b}\hh{a,b,c,d} \mone{a}\hh{a,b,c,d}\mone{a}\hh{a,b,c,d}\xx{b,d}\mone{b}\mone{d}
\end{split}
\end{align*}
so that this case reduces to the case of $x=a$. Similarly, if $x=c$
\begin{align*}
\hh{a,b,c,d} \mone{a}\hh{a,b,c,d}\mone{c}\hh{a,b,c,d} & \approx \hh{a,b,c,d} \mone{a}\hh{a,b,c,d}\mone{b}\xx{b,c}\xx{b,c}\hh{a,b,c,d} \\
& \approx \xx{b,c}\hh{a,b,c,d} \mone{a}\hh{a,b,c,d}\mone{b}\hh{a,b,c,d}\xx{b,c}
\end{align*}
and if $x=d$
\begin{align*}
\hh{a,b,c,d} \mone{a}\hh{a,b,c,d}\mone{d}\hh{a,b,c,d} & \approx \hh{a,b,c,d} \mone{a}\hh{a,b,c,d}\mone{b}\xx{b,d}\xx{b,d}\hh{a,b,c,d} \\
& \approx \xx{b,d}\hh{a,b,c,d} \mone{a}\hh{a,b,c,d}\mone{b}\hh{a,b,c,d}\xx{c,d}.\qedhere
\end{align*}
\end{proof}

\begin{proposition}
  \label{prop:rewriterules}
  Let $\word{G}$ be one of the words below.
  \begin{enumerate}
    \item $\hh{1,3,c,d}\mone{1}^{\tau_1'}\mone{3}^{\tau_3'}\mone{c}^{\tau_c}\mone{d}^{\tau_d}
        \hh{1,2,3,4} \mone{d}^{\tau_d} \mone{c}^{\tau_c}\mone{2}^{\tau_2}\mone{1}^{\tau_1}\hh{1,2,c,d}$
    \item $\hh{2,4,c,d}\mone{2}^{\tau_2'}\mone{4}^{\tau_4'}\mone{c}^{\tau_c}\mone{d}^{\tau_d}
                \hh{1,2,3,4} \mone{d}^{\tau_d} \mone{c}^{\tau_c}\mone{2}^{\tau_2}\mone{1}^{\tau_1}\hh{1,2,c,d}$
    \item $\hh{1,2,c,d}\mone{1}^{\tau_1'}\mone{2}^{\tau_2'}\mone{c}^{\tau_c}\mone{d}^{\tau_d}
                \hh{1,2,3,4} \mone{d}^{\tau_d} \mone{c}^{\tau_c}\mone{3}^{\tau_3}\mone{1}^{\tau_1}\hh{1,3,c,d}$
    \item $\hh{3,4,c,d}\mone{3}^{\tau_3'}\mone{4}^{\tau_4'}\mone{c}^{\tau_c}\mone{d}^{\tau_d}
                \hh{1,2,3,4} \mone{d}^{\tau_d} \mone{c}^{\tau_c}\mone{3}^{\tau_3}\mone{1}^{\tau_1}\hh{1,3,c,d}$
    \item $\hh{1,4,c,d}\mone{1}^{\tau_1'}\mone{4}^{\tau_4'}\mone{c}^{\tau_c}\mone{d}^{\tau_d}
                \hh{1,2,3,4} \mone{d}^{\tau_d} \mone{c}^{\tau_c}\mone{4}^{\tau_4}\mone{1}^{\tau_1}\hh{1,4,c,d}$
    \item $\hh{2,3,c,d}\mone{2}^{\tau_2'}\mone{3}^{\tau_3'}\mone{c}^{\tau_c}\mone{d}^{\tau_d}
                \hh{1,2,3,4} \mone{d}^{\tau_d} \mone{c}^{\tau_c}\mone{4}^{\tau_4}\mone{1}^{\tau_1}\hh{1,4,c,d}$
    \item $\hh{1,4,c,d}\mone{1}^{\tau_1'}\mone{4}^{\tau_4'}\mone{c}^{\tau_c}\mone{d}^{\tau_d}
                \hh{1,2,3,4} \mone{d}^{\tau_d} \mone{c}^{\tau_c}\mone{3}^{\tau_3}\mone{2}^{\tau_2}\hh{2,3,c,d}$
    \item $\hh{2,3,c,d}\mone{2}^{\tau_2'}\mone{3}^{\tau_3'}\mone{c}^{\tau_c}\mone{d}^{\tau_d}
                \hh{1,2,3,4} \mone{d}^{\tau_d} \mone{c}^{\tau_c}\mone{3}^{\tau_3}\mone{2}^{\tau_2}\hh{2,3,c,d}$
    \item $\hh{1,2,c,d}\mone{1}^{\tau_1'}\mone{2}^{\tau_2'}\mone{c}^{\tau_c}\mone{d}^{\tau_d}
                \hh{1,2,3,4} \mone{d}^{\tau_d} \mone{c}^{\tau_c}\mone{4}^{\tau_4}\mone{2}^{\tau_2}\hh{2,4,c,d}$
    \item $\hh{3,4,c,d}\mone{3}^{\tau_3'}\mone{4}^{\tau_4'}\mone{c}^{\tau_c}\mone{d}^{\tau_d}
                \hh{1,2,3,4} \mone{d}^{\tau_d} \mone{c}^{\tau_c}\mone{4}^{\tau_4}\mone{2}^{\tau_2}\hh{2,4,c,d}$
    \item $\hh{1,3,c,d}\mone{1}^{\tau_1'}\mone{3}^{\tau_3'}\mone{c}^{\tau_c}\mone{d}^{\tau_d}
                \hh{1,2,3,4} \mone{d}^{\tau_d} \mone{c}^{\tau_c}\mone{4}^{\tau_4}\mone{3}^{\tau_3}\hh{3,4,c,d}$
    \item $\hh{2,4,c,d}\mone{2}^{\tau_2'}\mone{4}^{\tau_4'}\mone{c}^{\tau_c}\mone{d}^{\tau_d}
                \hh{1,2,3,4} \mone{d}^{\tau_d} \mone{c}^{\tau_c}\mone{4}^{\tau_4}\mone{3}^{\tau_3}\hh{3,4,c,d}$
  \end{enumerate}
  Then there exist words $\word{V}$ and $\word{W}$ over
  $\s{\mone{x},\xx{x,y}}$, with $x,y\in\s{1,2,3,4,c,d}$, such that
  \[
  \word{G} \approx \word{V}\hh{1,2,3,4,}\word{W}.
  \]
\end{proposition}

\begin{proof}
Let $\word{G}$ be one of the words above. Then $\word{G}$ has the form
\[
  \hh{\alpha,\beta,c,d}\mone{\alpha}^{\tau_\alpha}\mone{\beta}^{\tau_\beta}\mone{c}^{\tau_c}\mone{d}^{\tau_d}
  \hh{1,2,3,4} \mone{d}^{\tau_d}
  \mone{c}^{\tau_c}\mone{\gamma}^{\tau_\gamma}\mone{\delta}^{\tau_\delta}\hh{\delta,\gamma,c,d}
\]
for appropriate indices $\alpha$, $\beta$, $\gamma$, and $\delta$. We
want to show that there exist $\word{V}$ and $\word{W}$ over
$\s{\mone{x},\xx{x,y}}$,
$\word{G}\approx \word{V}\hh{1,2,3,4}\word{W}$. By \cref{prop:derivedrels1},
evenly many occurrences of $(-1)$ can be commuted passed $K$. Since
$\hh{1,2,3,4} \approx \mone{c}\hh{1,2,3,4}\mone{c}$, we can thus
assume without loss of generality that $\word{G}$ is in fact of the
form
\[
  \hh{\alpha,\beta,c,d}\mone{\alpha}\hh{1,2,3,4} \hh{\delta,\gamma,c,d} \quad \mbox{or} \quad \hh{\alpha,\beta,c,d}\mone{\beta}\hh{1,2,3,4} \hh{\delta,\gamma,c,d}.
\]  
Using this simplification, we illustrate the rewriting strategy for the first two words.
\begin{enumerate}

\item In this case, without loss of generality, $\word{G}$ is either 
\[
\hh{1,3,c,d} \mone{1}\hh{1,2,3,4}\hh{1,2,c,d} \quad \mbox{or} \quad 
\hh{1,3,c,d} \mone{3}\hh{1,2,3,4}\hh{1,2,c,d}.
\]
By \cref{prop:derivedrels1,derivedrelsxk},
in the first case we get
\begin{align*}
\word{G} &\approx   \hh{1,3,c,d}\mone{1}\hh{1,2,3,4}\hh{1,2,c,d} \\
&\approx \mone{2}\hh{1,3,c,d}\mone{1}\mone{2}\hh{1,2,3,4}\hh{1,2,c,d} \\
&\approx \mone{2}\hh{1,3,c,d}\hh{1,2,3,4}\xx{1,3}\xx{2,4}\mone{1}\mone{2}\mone{3}\mone{4}\hh{1,2,c,d} \\
&\approx \word{V}\hh{1,3,c,d}\hh{1,2,3,4}\hh{3,4,c,d}\word{W}.
\end{align*}
And in the second case we get
\begin{align*}
\word{G} &\approx   \hh{1,3,c,d}\mone{3}\hh{1,2,3,4}\hh{1,2,c,d} \\
&\approx \mone{4}\hh{1,3,c,d}\mone{3}\mone{4}\hh{1,2,3,4}\hh{1,2,c,d} \\
&\approx \mone{4}\hh{1,3,c,d}\hh{1,2,3,4}\xx{1,3}\xx{2,4}\hh{1,2,c,d} \\
&\approx \word{V}\hh{1,3,c,d}\hh{1,2,3,4}\hh{3,4,c,d}\word{W}.
\end{align*}
Hence, to complete the proof it suffices to show that
$\hh{1,3,c,d}\hh{1,2,3,4}\hh{3,4,c,d}$ can be written in the desired
form. This is a consequence of \cref{rel:kcom1} since
\begin{align*}
\hh{1,3,c,d}\hh{1,2,3,4}\hh{3,4,c,d} &\approx \xx{1,2}\xx{1,2}\xx{3,4}\xx{3,4}\hh{1,3,c,d}\hh{1,2,3,4}\hh{3,4,c,d} \\
&\approx \xx{1,2}\xx{3,4}\hh{2,4,c,d}\xx{1,2}\xx{3,4}\hh{1,2,3,4}\hh{3,4,c,d} \\
&\approx \xx{1,2}\xx{3,4}\hh{2,4,c,d}\hh{1,2,3,4}\mone{3}\mone{4}\hh{3,4,c,d} \\
&\approx \word{V}\hh{2,4,c,d}\hh{1,2,3,4}\hh{3,4,c,d}\word{W}.
\end{align*}

\item In this case, without loss of generality, $\word{G}$ is either 
\[
\hh{2,4,c,d} \mone{2}\hh{1,2,3,4}\hh{1,2,c,d} \quad \mbox{or} \quad 
\hh{2,4,c,d} \mone{4}\hh{1,2,3,4}\hh{1,2,c,d}.
\]
 By \cref{prop:derivedrels1,derivedrelsxk},
in the first case we get
\begin{align*}
\word{G} &\approx \hh{2,4,c,d}\mone{2}\hh{1,2,3,4}\hh{1,2,c,d} \\
&\approx \mone{1}\hh{2,4,c,d}\mone{1}\mone{2}\hh{1,2,3,4}\hh{1,2,c,d} \\
&\approx \mone{1}\hh{2,4,c,d}\hh{1,2,3,4}\xx{1,3}\xx{2,4}\mone{1}\mone{2}\mone{3}\mone{4}\hh{1,2,c,d} \\
&\approx \word{V}\hh{2,4,c,d}\hh{1,2,3,4}\hh{3,4,c,d}\word{W}.
\end{align*}
And in the second case we get
\begin{align*}
\word{G} &\approx   \hh{2,4,c,d}\mone{4}\hh{1,2,3,4}\hh{1,2,c,d} \\
&\approx \mone{3}\hh{2,4,c,d}\mone{3}\mone{4}\hh{1,2,3,4}\hh{1,2,c,d} \\
&\approx \mone{3}\hh{2,4,c,d}\hh{1,2,3,4}\xx{1,3}\xx{2,4}\hh{1,2,c,d} \\
&\approx \word{V}\hh{2,4,c,d}\hh{1,2,3,4}\hh{3,4,c,d}\word{W}.
\end{align*}
Hence, to complete the proof it suffices to show that
$\hh{2,4,c,d}\hh{1,2,3,4}\hh{3,4,c,d}$ can be written in the desired
form, which follows directly from \cref{rel:kcom1}.
\end{enumerate}
The remaining cases are treated similarly.
\end{proof}

\begin{proposition}
\label{prop:derivedrels1alt}
  The relation below is derivable.
  \begin{subequations}
    \begin{align}
     \begin{split}
      \hh{e,f,g,h}\hh{a,b,c,d}&\xx{d,e}\hh{a,b,c,d}\hh{e,f,g,h} \\
      &\approx\ \\
      \mone{a}\mone{h}\xx{a,h}\hh{e,f,g,h}\hh{a,b,c,d}&\xx{d,e}\hh{a,b,c,d}\hh{e,f,g,h}\xx{a,h}\mone{a}\mone{h}
    \end{split}\notag
    \end{align}    
  \end{subequations}
\end{proposition}

\begin{proof}
Using \cref{rel:ksym5,rel:x}, we get:
\begin{align*}
    &\hh{e,f,g,h}\hh{a,b,c,d}\hh{a,b,c,e}\hh{d,f,g,h}\xx{a,h}\mone{a}\mone{h}\\ 
    &\approx\ \xx{e,h}\xx{e,h}\hh{e,f,g,h}\hh{a,b,c,d}\hh{a,b,c,e}\hh{d,f,g,h}\xx{a,h}\mone{a}\mone{h}\\ 
    &\approx\ \xx{e,h}\hh{e,f,g,h}\xx{f,h}\xx{g,h}\xx{f,h}\mone{f}\mone{g}\hh{a,b,c,d}\hh{a,b,c,e}\hh{d,f,g,h}\xx{a,h}\mone{a}\mone{h}\\ 
    &\approx\ \xx{e,h}\hh{e,f,g,h}\hh{a,b,c,d}\hh{a,b,c,e}\xx{f,h}\xx{g,h}\xx{f,h}\mone{f}\mone{g}\hh{d,f,g,h}\xx{a,h}\mone{a}\mone{h}\\ 
    &\approx\ \xx{e,h}\hh{e,f,g,h}\hh{a,b,c,d}\hh{a,b,c,e}\xx{f,g}\mone{f}\mone{g}\hh{d,f,g,h}\xx{a,h}\mone{a}\mone{h}\\ 
    &\approx\ \xx{e,h}\hh{e,f,g,h}\hh{a,b,c,d}\hh{a,b,c,e}\hh{d,f,g,h}\xx{d,h}\xx{a,h}\mone{a}\mone{h}\\ 
    &\approx\ \xx{e,h}\hh{e,f,g,h}\hh{a,b,c,d}\hh{a,b,c,e}\hh{d,f,g,h}\mone{a}\mone{d}\xx{d,h}\xx{a,h}\\
    &\approx\ \xx{e,h}\hh{e,f,g,h}\hh{a,b,c,d}\hh{a,b,c,e}\hh{d,f,g,h}\xx{a,d}\xx{a,d}\mone{a}\mone{d}\xx{d,h}\xx{a,h}\\
    &\approx\ \xx{e,h}\hh{e,f,g,h}\hh{a,b,c,d}\hh{a,b,c,e}\hh{d,f,g,h}\xx{a,d}\mone{a}\mone{d}\xx{a,d}\xx{d,h}\xx{a,h}\\
    &\approx\ \xx{e,h}\hh{e,f,g,h}\hh{a,b,c,d}\hh{a,b,c,e}\hh{d,f,g,h}\xx{a,d}\mone{a}\mone{d}\xx{d,h}\\
    &\approx\ \xx{e,h}\xx{a,e}\mone{a}\mone{e}\hh{e,f,g,h}\hh{a,b,c,d}\hh{a,b,c,e}\hh{d,f,g,h}\xx{d,h}\\
    &\approx\ \xx{e,h}\xx{a,e}\mone{a}\mone{e}\hh{e,f,g,h}\hh{a,b,c,d}\hh{a,b,c,e}\mone{f}\mone{g}\xx{f,g}\hh{d,f,g,h}\\
    &\approx\ \xx{e,h}\xx{a,e}\mone{a}\mone{e}\hh{e,f,g,h}\mone{f}\mone{g}\xx{f,g}\hh{a,b,c,d}\hh{a,b,c,e}\hh{d,f,g,h}\\
    &\approx\ \xx{e,h}\xx{a,e}\mone{a}\mone{e}\xx{e,h}\hh{e,f,g,h}\hh{a,b,c,d}\hh{a,b,c,e}\hh{d,f,g,h}\\
    &\approx\ \xx{e,h}\xx{a,e}\xx{e,h}\mone{a}\mone{h}\hh{e,f,g,h}\hh{a,b,c,d}\hh{a,b,c,e}\hh{d,f,g,h}\\
    &\approx\ \xx{a,h}\mone{a}\mone{h}\hh{e,f,g,h}\hh{a,b,c,d}\hh{a,b,c,e}\hh{d,f,g,h}.
\end{align*}
\end{proof}

\subsection{Basic and Simple Edges}
\label{apps:edges}

\begin{definition}
The subset $\gens_n'\subseteq\gens_n$ is defined as
\begin{equation}
  \label{eq:gens2}
  \gens_n' = \s{\xx{a, a + 1}, \hh{1,2,3,4}, \mone{1} \mid 1\leq a \leq n-1}.
\end{equation}
The elements of $\gens_n'$ are called \emph{basic} generators and an
edge $G:s\to t$ is called a \emph{basic} edge if $G$ is basic.
\end{definition}

\begin{definition}
\label{def:extent}
Let $G\in\gens_n$. The \emph{extent} of $G$ is the largest subscript
appearing in $G$. That is, 
\[
\extent(\xx{a,b}) = b, \quad \extent(\hh{a,b,c,d}) = d, \quad \mbox{and} \quad \extent(\mone{a}) = a. 
\]
The
\emph{extent} of a sequence $\word{G}= G_1 \cdots G_n$ is $\max\s{\extent(G_i)\mid 1 \leq i \leq n}$.
\end{definition}

\begin{lemma}
  \label{lem:simp}
  For any simple edge $G$, there exists a sequence of basic edges
  $\word{G}'$ such that
  \begin{enumerate}
  \item $\word{G}' \approx G$,
  \item $\extent(\word{G}') = \extent(G)$, and
  \item $\level(\word{G}') = \level(G)$.
  \end{enumerate}
\end{lemma}

\begin{proof}
See \cite{Gr2014}.
\end{proof}

\subsection{The Proof}
\label{apps:main}

We start with a version of the Main Lemma for basic edges, from which
the full version of the Main Lemma will follow.

\begin{lemma}
\label{lem:mainbasic}
  Let $s$, $t$, and $r$ be states, $N:s\too t$ be a normal edge, and $G:s\to r$ be a basic edge. Then there exist a state $q$, a sequence of normal edges $\word{N^*}:r\too q$, and a sequence of
  simple edges $\word{G^*}:t\to q$ such that the diagram
  \[
    \begin{tikzcd}[column sep=large, row sep=large]
      s \arrow[d, Rightarrow, "N" left] 
        \arrow[r, "G"] & 
        r\arrow[d, Rightarrow, "\word{N^*}"] \\
      t\arrow[r, "\word{G^*}" below] & q
    \end{tikzcd}
  \]
  commutes equationally and $\level(\word{G^*})<\level(s)$.
\end{lemma}

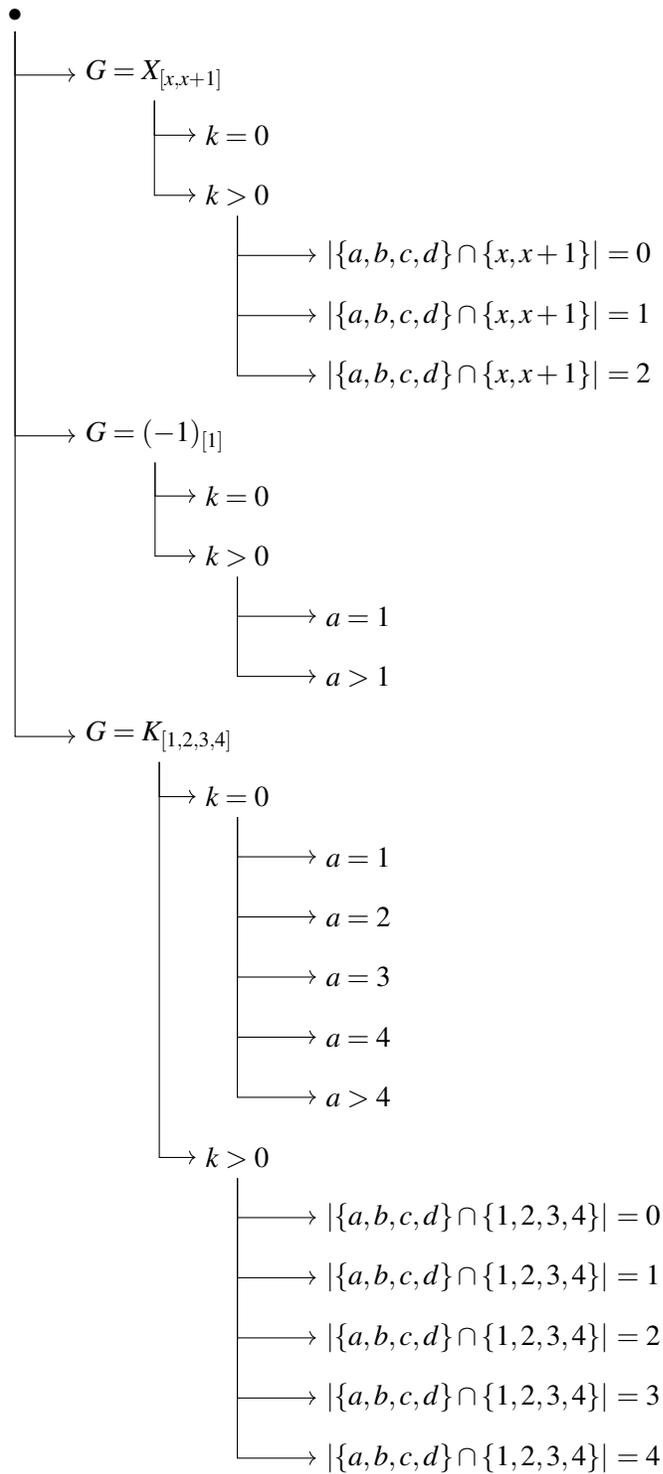
\begin{figure}
  \begin{tikzpicture}[scale=0.8]
  \node (init) at (-1,1) {$\bullet$};
    \node[anchor=west] (x) at (0,0) {$G=\xx{x,x+1}$};
      \node[anchor=west] (x0k) at (2,-1) {$k=0$};
      \node[anchor=west] (x1k) at (2,-2) {$k>0$};
        \node[anchor=west] (x1kabcd0) at (4,-3) {$|\s{a,b,c,d}\cap\s{ x , x +1}|=0$};
        \node[anchor=west] (x1kabcd1) at (4,-4) {$|\s{a,b,c,d}\cap\s{ x , x +1}|=1$};
        \node[anchor=west] (x1kabcd2) at (4,-5) {$|\s{a,b,c,d}\cap\s{ x , x +1}|=2$};    
    \node[anchor=west] (m) at (0,-6) {$G=\mone{1}$};
      \node[anchor=west] (m0k) at (2,-7) {$k=0$};
      \node[anchor=west] (m1k) at (2,-8) {$k>0$};
        \node[anchor=west] (m1ka0) at (4,-9) {$a=1$};
        \node[anchor=west] (m1ka1) at (4,-10) {$a>1$};
    \node[anchor=west] (k) at (0,-11) {$G=\hh{1,2,3,4}$};
      \node[anchor=west] (k0k) at (2,-12) {$k=0$};
        \node[anchor=west] (k0ka1) at (4,-13) {$a=1$};
        \node[anchor=west] (k0ka2) at (4,-14) {$a=2$};
        \node[anchor=west] (k0ka3) at (4,-15) {$a=3$};
        \node[anchor=west] (k0ka4) at (4,-16) {$a=4$};
        \node[anchor=west] (k0ka5) at (4,-17) {$a>4$};
      \node[anchor=west] (k1k) at (2,-18) {$k>0$};
        \node[anchor=west] (k1kabcd0) at (4,-19) {$|\s{a,b,c,d}\cap\s{1,2,3,4}|=0$};
        \node[anchor=west] (k1kabcd1) at (4,-20) {$|\s{a,b,c,d}\cap\s{1,2,3,4}|=1$};
        \node[anchor=west] (k1kabcd2) at (4,-21) {$|\s{a,b,c,d}\cap\s{1,2,3,4}|=2$};
        \node[anchor=west] (k1kabcd3) at (4,-22) {$|\s{a,b,c,d}\cap\s{1,2,3,4}|=3$};
        \node[anchor=west] (k1kabcd4) at (4,-23) {$|\s{a,b,c,d}\cap\s{1,2,3,4}|=4$};                  
  \draw[->] (init) |- (x);
  \draw[->] (init) |- (m);
  \draw[->] (init) |- (k);  
  \draw[->] (x) |- (x0k);
  \draw[->] (x) |- (x1k);
  \draw[->] (x1k) |- (x1kabcd0);
  \draw[->] (x1k) |- (x1kabcd1);
  \draw[->] (x1k) |- (x1kabcd2);

  \draw[->] (m) |- (m0k);
  \draw[->] (m) |- (m1k);
  \draw[->] (m1k) |- (m1ka0);
  \draw[->] (m1k) |- (m1ka1);

  \draw[->] (k) |- (k0k);
  \draw[->] (k) |- (k1k);
  \draw[->] (k0k) |- (k0ka1);
  \draw[->] (k0k) |- (k0ka2);
  \draw[->] (k0k) |- (k0ka3);
  \draw[->] (k0k) |- (k0ka4);
  \draw[->] (k0k) |- (k0ka5);
  \draw[->] (k1k) |- (k1kabcd0);
  \draw[->] (k1k) |- (k1kabcd1);
  \draw[->] (k1k) |- (k1kabcd2);
  \draw[->] (k1k) |- (k1kabcd3);
  \draw[->] (k1k) |- (k1kabcd4);      
\end{tikzpicture}
  \caption{The case distinction.}
  \label{fig:casedist}
\end{figure}

\begin{proof}
We proceed by case distinction. Since $r$, $t$ and $N$ are uniquely
determined by $G$ and $s$, it suffices to distinguish cases based on
the pair $(G,s)$. Let $v_s$ and $v_r$ be the pivot columns of $s$ and
$r$, respectively. Let $\level(s) = (j,k,m)$, where $j$ is the index
of $v_s$ in $s$, $k = \lde(v_s)$, and $m$ is the number of odd entries
in $2^k v_s$. We consider the cases $G=\xx{x,x+1}$, $G=\mone{1}$, and
$G=\hh{1,2,3,4}$ in turn. For each choice of $G$, we distinguish
further subcases depending on whether $k=0$ or
$k>0$. \cref{fig:casedist} represents the first three levels of the
case distinction.
\begin{case}{$G=\xxdef$.}

  \begin{subcase}{$k=0$.}
    Then $v_s=(-1)^{\tau_a} e_a$, where $\tau_a \in \Z_2$ and $1 \leq
    a \leq j$. We now consider the cases $j\leq x$ and $j>x$ in
    turn. For each choice of $j$ we distinguish further subcases
    corresponding to different values of $a$.

      \begin{ssubcase}{$j\leq  x $.}
        Then $\xx{x,x+1}$ acts non-trivially on the previously fixed
        columns and this case is therefore retrograde.
      \end{ssubcase}

      \begin{ssubcase}{$j>x$.}
        
        \begin{sssubcase}{$a \notin \{ x ,  x +1\}$.}
          Then $v_r=v_s$. Hence, $\level(r)=\level(s)$ and, from both
          $s$ and $r$, the algorithm prescribes
          $\xx{a,j}\mone{a}^{\tau_a}$. We complete the resulting
          diagrams as follows, depending on whether $ x +1 = j$ (left)
          or $ x +1 < j$ (right).
          \[
          \begin{tikzcd}[column sep=large, row sep=large]
            s \arrow[d, Rightarrow, "\xx{a , x +1}\mone{a }^{\tau_a}" left] \arrow[r, "\xx{ x , x +1}"] & r \arrow[d, Rightarrow, "\xx{a , x +1}\mone{a }^{\tau_a}"] \\
            t \arrow[r, "\xx{a , x }" below] & q
          \end{tikzcd}
          \qquad
          \begin{tikzcd}[column sep=large, row sep=large]
            s \arrow[d, Rightarrow, "\xx{a,j}\mone{a}^{\tau_a}" left] \arrow[r, "\xx{ x , x +1}"] & r \arrow[d, Rightarrow, "\xx{a,j}\mone{a}^{\tau_a}"] \\
            t \arrow[r, "\xx{ x , x +1}" below] & q
          \end{tikzcd}
          \]
          The diagrams commute by \cref{rel:disjoint2,rel:rename1,rel:rename2} since
          \[
          \xx{a,x+1}\mone{a}\xx{x,x+1} \approx \xx{a,x+1}\xx{x,x+1}\mone{a}\approx \xx{x,x+1}\xx{a,x}\mone{a}\approx \xx{a,x}\xx{a,x+1}\mone{a}
          \]
          and 
          \[
          \xx{a,j}\mone{a}\xx{x,x+1} \approx \xx{a,j}\xx{x,x+1}\mone{a}\approx \xx{x,x+1}\xx{a,j}\mone{a}\approx \xx{a,x}\xx{a,x+1}\mone{a}.
          \]
          Moreover, the level property is satisfied since
          $\level(t),\level(q)<\level(s)$.
        \end{sssubcase}
        
        \begin{sssubcase}{$a \in \{ x , x +1\}$.}
          Then $\xx{ x , x +1}$ acts non-trivially on $v_s$ and
          so $v_r \neq v_s$. If $j= x +1$, then the diagram to
          complete is one of the diagrams below, depending on whether
          $a =  x $ (left) or $a =  x  + 1$ (right).
          \[
          \begin{tikzcd}[column sep=large, row sep=large]
          s \arrow[d, Rightarrow, "\xx{a,j}\mone{a}^{\tau_a}" left] \arrow[r, "\xx{ x , x +1}"] & r \arrow[d, Rightarrow, "\mone{j}^{\tau_a}"] \\
          t & q
          \end{tikzcd}
          \qquad
          \begin{tikzcd}[column sep=large, row sep=large]
          s \arrow[d, Rightarrow, "\mone{a}^{\tau_a}" left] \arrow[r, "\xx{ x , x +1}"] & r \arrow[d, Rightarrow, "\xx{ x ,a}\mone{ x }^{\tau_a}"] \\
          t  & q
          \end{tikzcd}
          \]                      
          We then complete the diagrams as follows.                 
          \[
          \begin{tikzcd}[column sep=large, row sep=large]
          s \arrow[d, Rightarrow, "\xx{ x , x +1}\mone{ x }^{\tau_a}" left] \arrow[r, "\xx{ x , x +1}"] & r \arrow[d, Rightarrow, "\mone{ x +1}^{\tau_a}"] \\
          t \arrow[r, "\epsilon" below] & q
          \end{tikzcd}
          \qquad
          \begin{tikzcd}[column sep=large, row sep=large]
          s \arrow[d, Rightarrow, "\mone{ x +1}^{\tau_a}" left] \arrow[r, "\xx{ x , x +1}"] & r \arrow[d, Rightarrow, "\xx{ x , x +1}\mone{ x }^{\tau_a}"] \\
          t \arrow[r, "\epsilon" below] & q
          \end{tikzcd}
          \]          
          The diagrams commute by \cref{rel:rename3}
          and \cref{rel:orderx,rel:rename3}, respectively. Moreover,
          the level property is satisfied since we have
          $\level(q)=\level(t)<\level(s)$ in both cases. Now if $j > x
          + 1$, then the diagram to complete is one of the diagrams
          below, depending on whether $a = x $ (left) or $a = x + 1$
          (right).
          \[
          \begin{tikzcd}[column sep=large, row sep=large]
          s \arrow[d, Rightarrow, "\xx{ x  ,j}\mone{ x  }^{\tau_a}" left] \arrow[r, "\xx{ x , x +1}"] & r \arrow[d, Rightarrow, "\xx{ x +1,j}\mone{ x +1}^{\tau_a}"] \\
          t & q
          \end{tikzcd}
          \qquad
          \begin{tikzcd}[column sep=large, row sep=large]
          s \arrow[d, Rightarrow, "\xx{ x +1,j}\mone{ x +1}^{\tau_a}" left] \arrow[r, "\xx{ x , x +1}"] & r \arrow[d, Rightarrow, "\xx{ x ,j}\mone{ x }^{\tau_a}"] \\
          t  & q
          \end{tikzcd}
          \]                      
          We then complete the diagrams as follows.
          \[
          \begin{tikzcd}[column sep=large, row sep=large]
          s \arrow[d, Rightarrow, "\xx{ x ,j}\mone{ x }^{\tau_a}" left] \arrow[r, "\xx{ x , x +1}"] & r \arrow[d, Rightarrow, "\xx{ x +1,j}\mone{ x +1}^{\tau_a}"] \\
          t \arrow[r, "\xx{ x , x +1}" below]& q
          \end{tikzcd}
          \qquad
          \begin{tikzcd}[column sep=large, row sep=large]
          s \arrow[d, Rightarrow, "\xx{ x +1,j}\mone{ x +1}^{\tau_a}" left] \arrow[r, "\xx{ x , x +1}"] & r \arrow[d, Rightarrow, "\xx{ x ,j}\mone{ x }^{\tau_a}"] \\
          t  \arrow[r, "\xx{ x , x +1}" below] & q
          \end{tikzcd}
          \]
          Both diagrams commute by \cref{rel:rename1,rel:rename3} and
          the level property is satisfied since we have
          $\level(t)<\level(s)$ and $\level(q)<\level(s)$ in both
          cases.
        \end{sssubcase}
      \end{ssubcase}
  \end{subcase}

  \begin{subcase}{$k>0$.}
    Let $u = 2^k v_s$ and let $a,b,c,d$ be the indices of the first
    four odd entries of $u$. In this case, $N$ is of the form
    \[
    \hh{a,b,c,d}\mone{a}^{\tau_a} \mone{b}^{\tau_b}
    \mone{c}^{\tau_c}\mone{d}^{\tau_d},
    \]
    where $\tau_a, \tau_b, \tau_c, \tau_d \in \Z_2$. We have
    $|\s{a,b,c,d} \cap \s{x,x+1}| \in \s{0,1,2}$. We consider each one
    of these cases in turn.

    \begin{ssubcase}{$|\s{a,b,c,d} \cap \s{x,x+1}| = 0$.}
      If $ x +1\leq j$, then $\xx{x,x+1}$ acts trivially
      on the previously fixed columns and doesn't affect the number of
      odd entries in $u$. Hence $\level(r)=\level(s)$ and the first
      four odd entries in the integral part of $v_r$ also have indices
      $a$, $b$, $c$, and $d$. Thus the diagram to complete is the one
      shown below.
      \[
      \begin{tikzcd}[column sep=large, row sep=large]
      s \arrow[d, Rightarrow, "\hh{a,b,c,d}\mone{a}^{\tau_a}\mone{b}^{\tau_b}\mone{c}^{\tau_c}\mone{d}^{\tau_d}" left] \arrow[r, "\xx{x,x+1}"] & r \arrow[d, Rightarrow, "\hh{a,b,c,d}\mone{a}^{\tau_a}\mone{b}^{\tau_b}\mone{c}^{\tau_c}\mone{d}^{\tau_d}"] \\
       t & q
      \end{tikzcd}
      \]
      We then complete the diagram as follows.
      \[
      \begin{tikzcd}[column sep=large, row sep=large]
      s \arrow[d, Rightarrow, "\hh{a,b,c,d}\mone{a}^{\tau_a}\mone{b}^{\tau_b}\mone{c}^{\tau_c}\mone{d}^{\tau_d}" left] \arrow[r, "\xx{x,x+1}"] & r \arrow[d, Rightarrow, "\hh{a,b,c,d}\mone{a}^{\tau_a}\mone{b}^{\tau_b}\mone{c}^{\tau_c}\mone{d}^{\tau_d}"] \\
       t \arrow[r, "\xx{x,x+1}" below] & q
      \end{tikzcd}
      \]
      The diagram commutes by
      \cref{rel:disjoint2,rel:disjoint3}. Moreover, the level property
      is satisfied because $\level(t)<\level(s)$ and
      $\level(q)<\level(r)=\level(s)$. Now if $ x +1>j$, then
      $\xx{x,x+1}$ acts non-trivially on the previously
      fixed columns and the case is retrograde.
    \end{ssubcase}

    \begin{ssubcase}{$|\s{a,b,c,d} \cap \s{x,x+1}| = 1$.}

      \begin{sssubcase}\label{xssc:abcx}{$|\s{a,b,c} \cap \s{x,x+1}|=1$.}
        Then $ x +1<d$ and there are six subcases to consider,
        depending on whether $ x \in\s{a,b,c}$ or
        $ x +1\in\s{a,b,c}$. These cases can be uniformly
        represented by the diagram
        \[
        \begin{tikzcd}[column sep=large, row sep=large]
        s \arrow[d, Rightarrow, "\hh{a,b,c,d}\mone{a}^{\tau_a}\mone{b}^{\tau_b}\mone{c}^{\tau_c}\mone{d}^{\tau_d}" left] \arrow[r, "\xx{x,x+1}"] & r \arrow[d, Rightarrow, "\hh{a',b',c',d}\mone{a'}^{\tau_{a'}}\mone{b'}^{\tau_{b'}}\mone{c'}^{\tau_{c'}}\mone{d}^{\tau_{d}}"] \\
        t & q
        \end{tikzcd}
        \]
        where, for an index $p\in\s{a,b,c}$, we have $p'= x $ if
        $p=x+1$ and $p'=x+1$ if $p=x$. We can then
        complete the diagram as follows.
        \[
        \begin{tikzcd}[column sep=large, row sep=large]
        s \arrow[d, Rightarrow, "\hh{a,b,c,d}\mone{a}^{\tau_a}\mone{b}^{\tau_b}\mone{c}^{\tau_c}\mone{d}^{\tau_d}" left] \arrow[r, "\xx{x,x+1}"] & r \arrow[d, Rightarrow, "\hh{a',b',c',d}\mone{a'}^{\tau_{a'}}\mone{b'}^{\tau_{b'}}\mone{c'}^{\tau_{c'}}\mone{d}^{\tau_{d}}"] \\
        t\arrow[r, "\xx{x,x+1}" below] & q
        \end{tikzcd}
        \]
        The diagram commutes by
        \cref{rel:disjoint2,rel:rename3,rel:rename4,rel:rename5,rel:rename6,rel:rename7}
        and the level property is satisfied since $ x +1<d$ implies
        that $\level(r)=\level(s)$ and therefore that
        $\level(t)<\level(s)$ and $\level(q)<\level(r)=\level(s)$.
      \end{sssubcase}

      \begin{sssubcase}\label{xssc:ldegzerod}{$d= x +1$.}
        This case is similar to the previous one and the completed
        diagram is given below.
        \[
        \begin{tikzcd}[column sep=large, row sep=large]
         s \arrow[d, Rightarrow, "\hh{a,b,c, x +1}\mone{a}^{\tau_a}\mone{b}^{\tau_b}\mone{c}^{\tau_c}\mone{d}^{\tau_{d}}" left] 
        \arrow[r, "\xx{ x ,d}"] & 
        r\arrow[d, Rightarrow, "\hh{a,b,c, x }\mone{a}^{\tau_a}\mone{b}^{\tau_b}\mone{c}^{\tau_c}\mone{ x }^{\tau_{d}}"] \\
        t\arrow[r, "\xx{ x ,d}" below] & q
        \end{tikzcd}\\
        \]
        To see that the diagram commutes and that the level property
        is satisfied, one can reason in the same way as in
        \cref{xssc:abcx}.
      \end{sssubcase}            

      \begin{sssubcase}{$d= x $.}
        Let $e=d+1$ and let $u_e$ be the $e$-th component of $u$. If
        $u_e$ is even, then the indices of the first four odd entries
        of the integral part of $v_r$ are $a$, $b$, $c$, and $e$. We
        can then reason as \cref{xssc:abcx}, using the completed
        diagram below.
        \[
        \begin{tikzcd}[column sep=large, row sep=large]
        s \arrow[d, Rightarrow, "\hh{a,b,c,d}\mone{a}^{\tau_a}\mone{b}^{\tau_b}\mone{c}^{\tau_c}\mone{d}^{\tau_d}" left] \arrow[r, "\xx{d,e}"] & r \arrow[d, Rightarrow, "\hh{a,b,c,e}\mone{a}^{\tau_a}\mone{b}^{\tau_b}\mone{c}^{\tau_c}\mone{e}^{\tau_d}"] \\
        t \arrow[r, "\xx{d,e}" below]& q
        \end{tikzcd}
        \]        
        If $u_e$ is odd, then the indices of the first four odd
        entries of the integral part of $v_r$ are $a$, $b$, $c$, and
        $d$, and the diagram to complete is the one given below.
        \[
        \begin{tikzcd}[column sep=large, row sep=large]
        s \arrow[d, Rightarrow, "\hh{a,b,c,d}\mone{a}^{\tau_a}\mone{b}^{\tau_b}\mone{c}^{\tau_c}\mone{d}^{\tau_d}" left] \arrow[r, "\xx{d,e}"] & r \arrow[d, Rightarrow, "\hh{a,b,c,d}\mone{a}^{\tau_a}\mone{b}^{\tau_b}\mone{c}^{\tau_c}\mone{d}^{\tau_e}"] \\
        t & q
        \end{tikzcd}
        \]        
        In this case, there are at least two quadruples of odd entries
        in $u$. Let $f$, $g$, and $h$ be the indices of the first
        three odd entries of $u$ after $e$ and write $\overline{u}$
        for the vector composed of the first eight odd entries of
        $u$. Then, by \cref{applem:norm}, we have $\overline{u}^\intercal
        \overline{u}\equiv 0 \pmod{16}$ or $\overline{u}^\intercal
        \overline{u} \equiv 8 \pmod{16}$. We consider both of these
        cases in turn.

        \begin{ssssubcase}{$\overline{u}^\intercal \overline{u} \equiv 0 \pmod{16}$.}           
          Then we consider the diagram below.
          
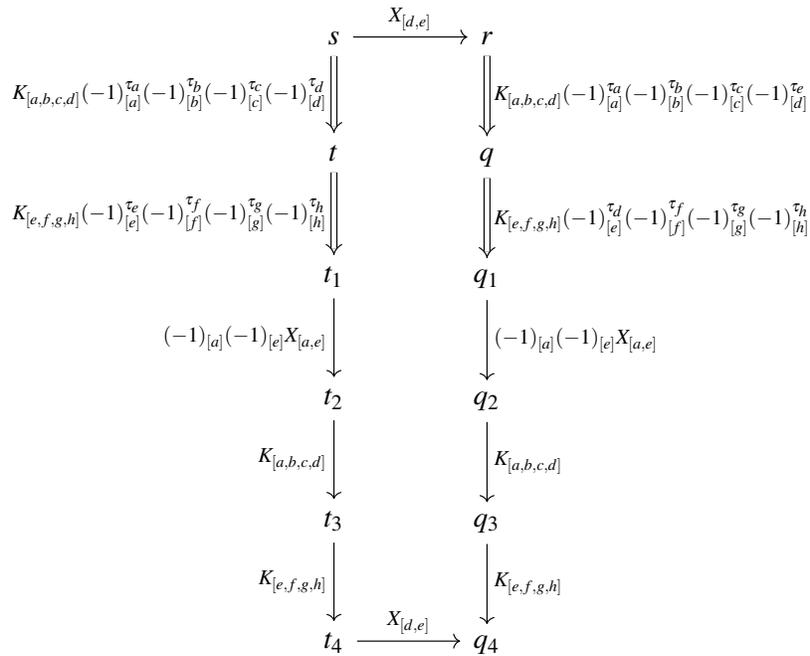
\[
          \begin{tikzcd}[column sep=large, row sep=large]
          s \arrow[d, Rightarrow, "\hh{a,b,c,d}\mone{a}^{\tau_a}\mone{b}^{\tau_b}\mone{c}^{\tau_c}\mone{d}^{\tau_d}" left] \arrow[r, "\xx{d,e}"] & r \arrow[d, Rightarrow, "\hh{a,b,c,d}\mone{a}^{\tau_a}\mone{b}^{\tau_b}\mone{c}^{\tau_c}\mone{d}^{\tau_e}"] \\
          t \arrow[d, Rightarrow, "\hh{e,f,g,h}\mone{e}^{\tau_e}\mone{f}^{\tau_f}\mone{g}^{\tau_g}\mone{h}^{\tau_h}" left] & q \arrow[d, Rightarrow, "\hh{e,f,g,h}\mone{e}^{\tau_d}\mone{f}^{\tau_f}\mone{g}^{\tau_g}\mone{h}^{\tau_h}"] \\
          t_1 \arrow[d, "\mone{a}\mone{e}\xx{a,e}" left] & q_1\arrow[d, "\mone{a}\mone{e}\xx{a,e}"] \\
          t_2 \arrow[d, "\hh{a,b,c,d}" left]& q_2 \arrow[d, "\hh{a,b,c,d}"] \\
          t_3 \arrow[d, "\hh{e,f,g,h}" left]& q_3 \arrow[d, "\hh{e,f,g,h}"] \\
          t_4 \arrow[r, "\xx{d,e}" above]& q_4
          \end{tikzcd}
          \]          
          To see that the diagram commutes, note that the occurrences
          of $(-1)$ in the top part of the diagram can be commuted
          past $X$ and $K$ and cancelled,
          using \cref{rel:ordermone,rel:disjoint2,rel:disjoint4,rel:disjoint5}. The
          fact that the diagram commutes is then a consequence
          of \cref{rel:x}. We now verify that the diagram satisfies
          the level property. The first two edges descending from $s$
          are prescribed by the algorithm. Thus
          \[
          \hh{a,b,c,d}          
          \mone{a}^{\tau_a}\mone{b}^{\tau_b}\mone{c}^{\tau_c}\mone{d}^{\tau_d}
          \overline{u} = \overline{w}
          \]
          with $\overline{w}\equiv 00001111 \pmod{2}$, so that
          $\level(t_1)\leq (j,k,m-1)<\level(s)$. Similarly,
          \[
          \hh{e,f,g,h}
          \mone{e}^{\tau_e}\mone{f}^{\tau_f}\mone{g}^{\tau_g}\mone{h}^{\tau_h}          
          \overline{w} = \overline{w}_1
          \]
          with $\overline{w}_1 \equiv 00000000 \pmod{2}$ so that
          $\level(t_1)\leq (j,k,m-2)<\level(s)$. Moreover, by
          \cref{applem:normresidue1}, we know that $\overline{w}_1 =
          2\overline{w}_1'$ with $\overline{w}_1'\equiv 10000111
          \pmod{2}$ or $\overline{w}_1'\equiv 01111000 \pmod{2}$.
          Hence,
          \[
          \mone{a}\mone{e}\xx{a,e} \overline{w}_1 = \overline{w}_2
          \]
          with $\overline{w}_2 = 2\overline{w}_2'$ and
          $\overline{w}_2'\equiv 00001111 \pmod{2}$ or
          $\overline{w}_2'\equiv 11110000 \pmod{2}$. Thus
          $\level(t_2)=\level(t_1)\leq (j,k,m-2) < \level(s)$. By
          \cref{applem:honemod4}, we have
          \[
          \hh{a,b,c,d} \overline{w}_2 = \overline{w}_3 \quad \mbox{
            and } \quad \hh{e,f,g,h} \overline{w}_3 = \overline{w}_4
          \]
          with $\overline{w}_3 = 2\overline{w}_3'$ and $\overline{w}_4
          = 2\overline{w}_4'$ for some
          $\overline{w}_3',\overline{w}_4'\in\Z^8$. Hence, we get
          \[
          \level(t_4),\level(t_3) \leq \level(t_2)<\level(s).
          \]
          We can reason analogously with the right hand side of the
          diagram to show that
          \[
          \level(q_4), \level(q_3),\level(q_2), \level(q_1),\level(q),
          \level(r)<\level(s).
          \]
          This proves that the diagram satisfies the level property.
        \end{ssssubcase}      
        
        \begin{ssssubcase}{$\overline{u}^\intercal \overline{u} \equiv 8 \pmod{16}$.}                        
          Then we consider the diagram below.
          \[
          \begin{tikzcd}[column sep=large, row sep=large]
          s \arrow[d, Rightarrow, "\hh{a,b,c,d}\mone{a}^{\tau_a}\mone{b}^{\tau_b}\mone{c}^{\tau_c}\mone{d}^{\tau_d}" left] \arrow[r, "\xx{d,e}"] & r \arrow[d, Rightarrow, "\hh{a,b,c,d}\mone{a}^{\tau_a}\mone{b}^{\tau_b}\mone{c}^{\tau_c}\mone{d}^{\tau_e}"] \\
          t \arrow[d, Rightarrow, "\hh{e,f,g,h}\mone{e}^{\tau_e}\mone{f}^{\tau_f}\mone{g}^{\tau_g}\mone{h}^{\tau_h}" left] & q \arrow[d, Rightarrow, "\hh{e,f,g,h}\mone{e}^{\tau_d}\mone{f}^{\tau_f}\mone{g}^{\tau_g}\mone{h}^{\tau_h}"] \\
          t_1 \arrow[d,"\mone{a}\mone{h}\xx{a,h}" left] & q_1\arrow[d,"\mone{a}\mone{h}\xx{a,h}"] \\
          t_2 \arrow[d, "\hh{a,b,c,d}" left]& q_2 \arrow[d, "\hh{a,b,c,d}"] \\
          t_3 \arrow[d, "\hh{e,f,g,h}" left]& q_3 \arrow[d, "\hh{e,f,g,h}"] \\
          t_4 \arrow[r, "\xx{d,e}" above]& q_4
          \end{tikzcd}
          \]          
          To see that the diagram commutes, note that the occurrences
          of $(-1)$ in the top part of the diagram can be commuted
          past $X$ and $K$ and cancelled,
          using \cref{rel:ordermone,rel:disjoint2,rel:disjoint4,rel:disjoint5}. The
          fact that the diagram commutes is then a consequence
          of \cref{prop:derivedrels1alt}. We now verify that the diagram
          satisfies the level property. The first two edges descending
          from $s$ are prescribed by the algorithm. Thus
          \[
          \hh{a,b,c,d}          
          \mone{a}^{\tau_a}\mone{b}^{\tau_b}\mone{c}^{\tau_c}\mone{d}^{\tau_d}
          \overline{u} = \overline{w}
          \]
          with $\overline{w}\equiv 00001111 \pmod{2}$, so that
          $\level(t_1)\leq (j,k,m-1)<\level(s)$. Similarly,
          \[
          \hh{e,f,g,h}
          \mone{e}^{\tau_e}\mone{f}^{\tau_f}\mone{g}^{\tau_g}\mone{h}^{\tau_h}          
          \overline{w} = \overline{w}_1
          \]
          with $\overline{w}_1 \equiv 00000000 \pmod{2}$ so that
          $\level(t_1)\leq (j,k,m-2)<\level(s)$. Moreover, by
          \cref{applem:normresidue2}, we know that $\overline{w}_1 =
          2\overline{w}_1'$ with $\overline{w}_1'\equiv 10001000
          \pmod{2}$ or $\overline{w}_1'\equiv 01110111 \pmod{2}$.
          Hence,
          \[
          \mone{a}\mone{h}\xx{a,h} \overline{w}_1 = \overline{w}_2
          \]
          with $\overline{w}_2 = 2\overline{w}_2'$ and
          $\overline{w}_2'\equiv 00001001 \pmod{2}$ or
          $\overline{w}_2'\equiv 11110110 \pmod{2}$. Thus
          $\level(t_2)=\level(t_1)\leq (j,k,m-2) < \level(s)$. By
          \cref{applem:honemod4}, we have
          \[
          \hh{a,b,c,d} \overline{w}_2 = \overline{w}_3 \quad \mbox{
            and } \quad \hh{e,f,g,h} \overline{w}_3 = \overline{w}_4
          \]
          with $\overline{w}_3 = 2\overline{w}_3'$ and $\overline{w}_4
          = 2\overline{w}_4'$ for some
          $\overline{w}_3',\overline{w}_4'\in\Z^8$. Hence, we get
          \[
          \level(t_4),\level(t_3) \leq \level(t_2)<\level(s).
          \]
          We can reason analogously with the right hand side of the
          diagram to show that
          \[
          \level(q_4), \level(q_3),\level(q_2), \level(q_1),\level(q),
          \level(r)<\level(s).
          \]
          This proves that the diagram satisfies the level property.
        \end{ssssubcase}            
      \end{sssubcase}        
    \end{ssubcase}

    \begin{ssubcase}{$|\s{a,b,c,d} \cap \s{x,x+1}| = 2$.}
      Then the diagram to complete is one of the diagrams below,
      depending on whether $\s{x,x+1}=\s{a,b}$ (top),
      $\s{x,x+1}=\s{b,c}$ (center), or
      $\s{x,x+1}=\s{c,d}$ (bottom).
      \[
      \begin{tikzcd}[column sep=large, row sep=large]
       s \arrow[d, Rightarrow, "\hh{a,b,c,d}\mone{a}^{\tau_a}\mone{b}^{\tau_{b}}\mone{c}^{\tau_c}\mone{d}^{\tau_d}" left] 
      \arrow[r, "\xx{a,b}"] & 
      r\arrow[d, Rightarrow, "\hh{a,b,c,d}\mone{a}^{\tau_b}\mone{b}^{\tau_a}\mone{c}^{\tau_c}\mone{d}^{\tau_d}"] \\
      t                                   & q
      \end{tikzcd}\\
      \]
      \[
      \begin{tikzcd}[column sep=large, row sep=large]
       s \arrow[d, Rightarrow, "\hh{a,b,c,d}\mone{a}^{\tau_a}\mone{b}^{\tau_b}\mone{c}^{\tau_c}\mone{d}^{\tau_d}" left] 
      \arrow[r, "\xx{b,c}"] & 
      r\arrow[d, Rightarrow, "\hh{a,b,c,d}\mone{a}^{\tau_a}\mone{b}^{\tau_c}\mone{c}^{\tau_b}\mone{d}^{\tau_d}"] \\
      t                                   & q
      \end{tikzcd}\\
      \]
      \[
      \begin{tikzcd}[column sep=large, row sep=large]
       s \arrow[d, Rightarrow, "\hh{a,b,c,d}\mone{a}^{\tau_a}\mone{b}^{\tau_b}\mone{c}^{\tau_c}\mone{d}^{\tau_d}" left] 
      \arrow[r, "\xx{c,d}"] & 
      r\arrow[d, Rightarrow, "\hh{a,b,c,d}\mone{a}^{\tau_a}\mone{b}^{\tau_b}\mone{c}^{\tau_d}\mone{d}^{\tau_c}"] \\
      t                                   & q
      \end{tikzcd}\\
      \]
      We then complete the diagrams as follows.
      \[
      \begin{tikzcd}[column sep=7em, row sep=large]
       s \arrow[d, Rightarrow, "\hh{a,b,c,d}\mone{a}^{\tau_a}\mone{b}^{\tau_{b}}\mone{c}^{\tau_c}\mone{d}^{\tau_d}" left] 
      \arrow[r, "\xx{a,b}"] & 
      r\arrow[d, Rightarrow, "\hh{a,b,c,d}\mone{a}^{\tau_b}\mone{b}^{\tau_a}\mone{c}^{\tau_c}\mone{d}^{\tau_d}"] \\
      t\arrow[r, "\mone{d}\mone{b}\xx{b,d}" below] & q
      \end{tikzcd}\\
      \]
      \[
      \begin{tikzcd}[column sep=large, row sep=large]
       s \arrow[d, Rightarrow, "\hh{a,b,c,d}\mone{a}^{\tau_a}\mone{b}^{\tau_b}\mone{c}^{\tau_c}\mone{d}^{\tau_d}" left] 
      \arrow[r, "\xx{b,c}"] & 
      r\arrow[d, Rightarrow, "\hh{a,b,c,d}\mone{a}^{\tau_a}\mone{b}^{\tau_c}\mone{c}^{\tau_b}\mone{d}^{\tau_d}"] \\
      t\arrow[r, "\xx{b,c}" below] & q
      \end{tikzcd}\\
      \]
      \[
      \begin{tikzcd}[column sep=large, row sep=large]
       s \arrow[d, Rightarrow, "\hh{a,b,c,d}\mone{a}^{\tau_a}\mone{b}^{\tau_b}\mone{c}^{\tau_c}\mone{d}^{\tau_d}" left] 
      \arrow[r, "\xx{c,d}"] & 
      r\arrow[d, Rightarrow, "\hh{a,b,c,d}\mone{a}^{\tau_a}\mone{b}^{\tau_b}\mone{c}^{\tau_d}\mone{d}^{\tau_c}"] \\
      t\arrow[r, "\xx{b,d}" below] & q
      \end{tikzcd}\\
      \]
      The diagrams commute by
      \cref{rel:disjoint2,rel:disjoint4,rel:rename3,rel:ksym1,rel:ksym2,rel:ksym3}. Moreover,
      the level property is satisfied in the three diagrams since the
      level of $s$ is unaffected by $\xx{x,x+1}$ so that
      $\level(t)<\level(s)$ and $\level(q)<\level(r)=\level(s)$.
    \end{ssubcase}
  \end{subcase}
\end{case}

\begin{case}{$G=\monedef$.}

  \begin{subcase}{$k=0$.}
    Then $v_s=(-1)^{\tau_a} e_a$, where $\tau_a \in \Z_2$ and $1 \leq
    a \leq j$. We now consider the cases $a=1$ and $a>1$ in turn. For
    each choice of $a$ we distinguish further subcases corresponding
    to different values of $j$.

    \begin{ssubcase}{$a=1$.}

      \begin{sssubcase}{$j=1$.}
        Then $\tau_a=1$, $r=I$, and the completed diagram is given
        below.
        \[
        \begin{tikzcd}[column sep=large, row sep=large]
        s \arrow[d, Rightarrow, "\monedef" left] \arrow[r, "\monedef"]
        & r \arrow[d, Rightarrow, "\epsilon"] \\ t \arrow[r,
          "\epsilon" below]& q
        \end{tikzcd}
        \]
        The diagram commutes since $\approx$ is reflexive and the
        level property is satisfied since
        $\level(t)=\level(q)<\level(s)$.
      \end{sssubcase}

      \begin{sssubcase}{$j>1$.}
        Then $v_r=(-1)^{\tau_a+1}e_1$ and the completed diagram is given
        below.
        \[
        \begin{tikzcd}[column sep=large, row sep=large]
        s \arrow[d, Rightarrow, "\xx{1,j}\mone{1}^{\tau_a}" left]
        \arrow[r, "\monedef"] & r \arrow[d, Rightarrow,
          "\xx{1,j}\monedef^{\tau_a+1}"] \\ t \arrow[r, "\epsilon"
          below]& q
        \end{tikzcd}
        \]
        The diagram commutes by \cref{rel:ordermone} and the level
        property is satisfied since $\level(t)=\level(q)<\level(s)$.
      \end{sssubcase}
    \end{ssubcase}

    \begin{ssubcase}{$a>1$.}

      \begin{sssubcase}{$j=a$.}
        Then $\mone{1}$ acts trivially on $v_s$ and so
        $v_r=v_s$. Hence, the completed diagram is given below.
        \[
        \begin{tikzcd}[column sep=large, row sep=large]
        s \arrow[d, Rightarrow, "\mone{a}^{\tau_a}" left] \arrow[r,
          "\monedef"] & r \arrow[d, Rightarrow, "\mone{a}^{\tau_a}"]
        \\ t \arrow[r, "\monedef" below]& q
        \end{tikzcd}
        \]
        The diagram commutes by \cref{rel:disjoint4} and the level
        property is satisfied since $\level(t)< \level(s)$ and
        $\level(q)<\level(r) = \level(s)$.
      \end{sssubcase}

      \begin{sssubcase}{$j > a$.}
        Then $\mone{1}$ acts trivially on $v_s$ and so
        $v_r=v_s$. Hence, the completed diagram is given below.
        \[
        \begin{tikzcd}[column sep=large, row sep=large]
        s \arrow[d, Rightarrow, "\xx{a,j}\mone{a}^{\tau_a}" left]
        \arrow[r, "\monedef"] & r \arrow[d, Rightarrow,
          "\xx{a,j}\mone{a}^{\tau_a}"] \\ t \arrow[r, "\monedef"
          below]& q
        \end{tikzcd}
        \]
        The diagram commutes by \cref{rel:disjoint2,rel:disjoint4}
        and the level property is satisfied since
        $\level(t)< \level(s)$ and $\level(q)<\level(r) = \level(s)$.
      \end{sssubcase}
    \end{ssubcase}
  \end{subcase}

  \begin{subcase}{$k>0$.}
    Let $u = 2^k v_s$ and let $a,b,c,d$ be the indices of the first
    four odd entries of $u$. In this case, $N$ is of the form
    \[
    \hh{a,b,c,d}\mone{a}^{\tau_a} \mone{b}^{\tau_b}
    \mone{c}^{\tau_c}\mone{d}^{\tau_d},
    \]
    where $\tau_a, \tau_b, \tau_c, \tau_d \in \Z_2$. We have $a=1$ or
    $a>1$. We consider each one of these cases in turn.

    \begin{ssubcase}{$a = 1$.}
      Then $\mone{1}$ acts non-trivially on $v_s$ and so $v_r\neq
      v_s$. Hence, the completed diagram is given below.
      \[
      \begin{tikzcd}[column sep=large, row sep=large]
      s \arrow[d, Rightarrow, "\hh{1,b,c,d} \mone{1}^{\tau_1}
        \mone{b}^{\tau_b}\mone{c}^{\tau_c}\mone{d}^{\tau_d}" left]
      \arrow[r, "\monedef"] & r \arrow[d, Rightarrow,"\hh{1,b,c,d}
        \mone{1}^{\tau_1+1} \mone{b}^{\tau_b}
        \mone{c}^{\tau_c}\mone{d}^{\tau_d}"] \\ t \arrow[r, "\epsilon"
        below]& q
      \end{tikzcd}
      \]
      The diagram commutes by \cref{rel:disjoint4,rel:ordermone}
      and the level property is satisfied since
      $\level(t)=\level(q)<\level(r)=\level(s)$.
    \end{ssubcase}  

    \begin{ssubcase}{$a > 1$.}
      Then $\mone{1}$ does not affect the odd entries of $v_s$. Hence,
      the completed diagram is given below.
      \[
      \begin{tikzcd}[column sep=large, row sep=large]
      s \arrow[d, Rightarrow, "\hh{a,b,c,d} \mone{a}^{\tau_a}
        \mone{b}^{\tau_b}\mone{c}^{\tau_c}\mone{d}^{\tau_d}" left]
      \arrow[r, "\monedef"] & r \arrow[d, Rightarrow, "\hh{a,b,c,d}
        \mone{a}^{\tau_a}
        \mone{b}^{\tau_b}\mone{c}^{\tau_c}\mone{d}^{\tau_d}~~"] \\ t
      \arrow[r, "\mone{1}" below]& q
      \end{tikzcd}
      \]
      The diagram commutes by
      \cref{rel:disjoint4,rel:disjoint5,rel:ordermone}
      and the level property is satisfied since
      $\level(t)=\level(q)<\level(r)=\level(s)$.
    \end{ssubcase}
  \end{subcase}
\end{case}

\begin{case}{$G=\hhdef$.}

  \begin{subcase}{$k=0$.}
    Then $v_s = (-1)^{\tau_a} e_a$, where $\tau_a \in \Z_2$ and $1\leq
    a \leq j$. We now consider the cases $a=1$, $a=2$, $a=3$, $a=4$,
    and $a>4$ in turn. For each choice of $a$ we distinguish further
    subcases corresponding to different values of $j$.

    \begin{ssubcase}{$a = 1$.}

      \begin{sssubcase}{$j=1$.}
        Then $\tau_1=1$. Hence, from $s$, the algorithm prescribes
        $\mone{1}$. The level of $r$ is $(4,1,4)$ and, from $r$, the
        algorithm prescribes
        \[
        \hh{1,2,3,4}\mone{2}\mone{3}, \quad \xx{1,4}, \quad \xx{2,3},
        \quad \mbox{ and } \quad \mone{1}.
        \]
        We complete the resulting diagram as follows.
        \[
        \begin{tikzcd}[column sep=large, row sep=large]
        s \arrow[d, Rightarrow,"\mone{1}" left] \arrow[r, "\hhdef"] & r \arrow[d, Rightarrow, "\hh{1,2,3,4}\mone{2} \mone{3}"] \\
        t  \arrow[dddr,"\varepsilon" swap]& q_1 \arrow[d, Rightarrow, "\xx{1,4}"]&\\      
            & q_2 \arrow[d, Rightarrow, "\xx{2,3}"]\\
            & q_3 \arrow[d, Rightarrow, "\mone{1}"] \\
            & q_4
        \end{tikzcd}      
        \]
        The diagram commutes by \cref{rel:k31,rel:orderx,rel:orderk} since
        \begin{align*}
        \mone{1}\xx{2,3}\xx{1,4}\hh{1,2,3,4}\mone{2}\mone{3}\hh{1,2,3,4} &\approx
          \mone{1}\xx{2,3}\xx{1,4}\xx{1,4}\xx{2,3} \\
          &\approx \mone{1}.
        \end{align*}
        Moreover, the level property is satisfied since
        $\level(t),\level(q_4)<(1,0,0)=\level(s)$.
      \end{sssubcase}

      \begin{sssubcase}{$j=2$.}
        Then, from $s$, the algorithm prescribes
        $\xx{1,2}\mone{2}^{\tau_1}$.  The level of $r$ is $(4,1,4)$
        and, from $r$, the algorithm prescribes
        \[
        \hh{1,2,3,4}\mone{2}\mone{3}, \quad \xx{1,4}, \quad \xx{2,3},
        \quad \mbox{ and } \quad \xx{1,2}\mone{1}^{\tau_1}.
        \]
        We complete the resulting diagram as follows.
        
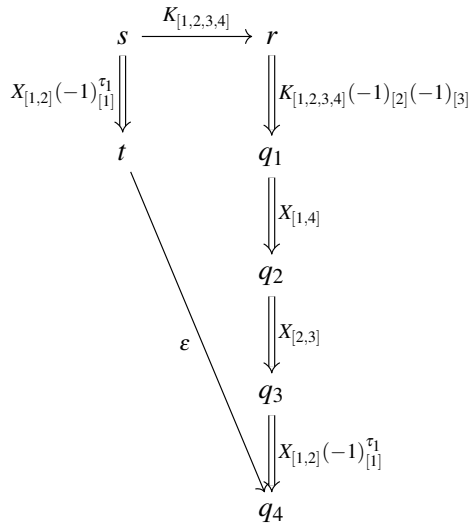
\[
        \begin{tikzcd}[column sep=large, row sep=large]
        s \arrow[d, Rightarrow,"\xx{1,2}\mone{1}^{\tau_1}" left] \arrow[r, "\hhdef"] & r \arrow[d, Rightarrow, "\hh{1,2,3,4}\mone{2} \mone{3}"] \\
        t  \arrow[dddr,"\varepsilon" swap]& q_1 \arrow[d, Rightarrow, "\xx{1,4}"]&\\      
            & q_2 \arrow[d, Rightarrow, "\xx{2,3}"]\\
            & q_3 \arrow[d, Rightarrow, "\xx{1,2}\mone{1}^{\tau_1}"] \\
            & q_4
        \end{tikzcd}      
        \]
        The diagram commutes by reasoning as in the previous case.
        Moreover, the level property is satisfied since
        $\level(t),\level(q_4) < (2,0,0)=\level(s)$.
      \end{sssubcase}

      \begin{sssubcase}{$j=3$.}
        Then, from $s$, the algorithm prescribes
        $\xx{1,3}\mone{1}^{\tau_1}$.  The level of $r$ is $(4,1,4)$
        and, from $r$, the algorithm prescribes
        \[
        \hh{1,2,3,4}\mone{2}\mone{3}, \quad \xx{1,4}, \quad \mbox{ and
        } \quad \xx{1,3}\mone{1}^{\tau_1} .
        \]
        We complete the resulting diagram as follows.
        \[
        \begin{tikzcd}[column sep=large, row sep=large]
        s \arrow[d, Rightarrow,"\xx{1,3}\mone{1}^{\tau_1}" left] \arrow[r, "\hhdef"] & r \arrow[d, Rightarrow, "\hh{1,2,3,4}\mone{2} \mone{3}"] \\
        t  \arrow[ddr,"\xx{1,2}" swap]& q_1 \arrow[d, Rightarrow, "\xx{1,4}"]&\\      
            & q_2 \arrow[d, Rightarrow, "\xx{1,3}\mone{1}^{\tau_1}"]\\
            & q_3 
        \end{tikzcd}      
        \]
        The diagram commutes by
        \cref{rel:k31,rel:orderx,rel:orderk,rel:rename1,rel:rename2}
        since
        \begin{align*}
        \xx{1,3}\mone{1}^{\tau_1}\xx{1,4}\hh{1,2,3,4}\mone{2}\mone{3}\hh{1,2,3,4} &\approx
          \xx{1,3}\mone{1}^{\tau_1}\xx{1,4}\xx{1,4}\xx{2,3}\\
          &\approx \xx{1,3}\xx{2,3}\mone{1}^{\tau_1}\\
          &\approx \xx{1,2}\xx{1,3}\mone{1}^{\tau_1}.          
        \end{align*}
        Moreover, the level property is satisfied since
        $\level(t),\level(q_3) < (3,0,0)=\level(s)$.
      \end{sssubcase}

      \begin{sssubcase}{$j=4$.}
        Then, from $s$, the algorithm prescribes
        $\xx{1,4}\mone{1}^{\tau_1}$.  The level of $r$ is $(4,1,4)$
        and, from $r$, the algorithm prescribes
        \[
        \hh{1,2,3,4}\mone{1}^{\tau_1}\mone{2}^{\tau_1}\mone{3}^{\tau_1}\mone{4}^{\tau_1},
        \quad \mbox{ and } \quad \xx{1,4}.
        \]
        We complete the resulting diagram as follows.
        \[
        \begin{tikzcd}[column sep=large, row sep=large]
        s \arrow[d, Rightarrow,"\xx{1,4}\mone{1}^{\tau_1}" left] \arrow[r, "\hhdef"] & r \arrow[d, Rightarrow, "\hh{1,2,3,4}\mone{1}^{\tau_1}\mone{2}^{\tau_1}\mone{3}^{\tau_1}\mone{4}^{\tau_1}"] \\
        t  \arrow[dr,"\mone{1}^{\tau_1}\mone{2}^{\tau_1}\mone{3}^{\tau_1}" swap]& q_1 \arrow[d, Rightarrow, "\quad \xx{1,4}"]&\\      
            & q_2 
        \end{tikzcd}      
        \]
        The diagram commutes by
        \cref{itm:relk4,rel:orderk,rel:rename3,rel:disjoint2,rel:disjoint4}
        since
        \begin{align*}
        \xx{1,4} \hh{1,2,3,4}\mone{1}^{\tau_1}\mone{2}^{\tau_1}\mone{3}^{\tau_1}\mone{4}^{\tau_1} \hh{1,2,3,4}  &\approx \xx{1,4} \mone{1}^{\tau_1}\mone{2}^{\tau_1}\mone{3}^{\tau_1}\mone{4}^{\tau_1} \\
        &\approx \mone{1}^{\tau_1}\mone{2}^{\tau_1}\mone{3}^{\tau_1} \xx{1,4}\mone{1}^{\tau_1}.         
        \end{align*}
        Moreover, the level property is satisfied since
        $\level(t),\level(q_2) < (4,0,0)=\level(s)$ and the extent of
        $\mone{1}^{\tau_1}\mone{2}^{\tau_1}\mone{3}^{\tau_1}$ is
        strictly less than 4.
      \end{sssubcase}

      \begin{sssubcase}{$j>4$.}
        Then, from $s$, the algorithm prescribes
        $\xx{1,j}\mone{1}^{\tau_1}$.  The level of $r$ is $(j,1,4)$
        and, from $r$, the algorithm prescribes
        \[
        \hh{1,2,3,4}\mone{1}^{\tau_1}\mone{2}^{\tau_1}\mone{3}^{\tau_1}\mone{4}^{\tau_1},
        \quad \mbox{ and } \quad \xx{1,j}.
        \]
        We complete the resulting diagram as follows.
        
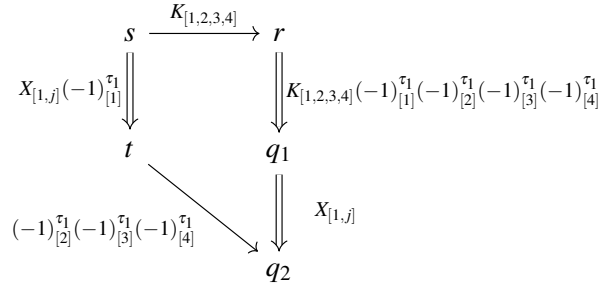
\[
        \begin{tikzcd}[column sep=large, row sep=large]
        s \arrow[d, Rightarrow,"\xx{1,j}\mone{1}^{\tau_1}" left] \arrow[r, "\hhdef"] & r \arrow[d, Rightarrow, "\hh{1,2,3,4}\mone{1}^{\tau_1}\mone{2}^{\tau_1}\mone{3}^{\tau_1}\mone{4}^{\tau_1}"] \\
        t  \arrow[dr,"\mone{2}^{\tau_1}\mone{3}^{\tau_1}\mone{4}^{\tau_1}" swap]& q_1 \arrow[d, Rightarrow, "\quad \xx{1,j}"]&\\      
            & q_2 
        \end{tikzcd}      
        \]
        The diagram commutes by reasoning as in the previous case.
        Moreover, the level property is satisfied since
        $\level(t),\level(q_2) < (j,0,0)=\level(s)$ and the
        extent of
        $\mone{2}^{\tau_1}\mone{3}^{\tau_1}\mone{4}^{\tau_1}$ is
        strictly less than $j$.
      \end{sssubcase}      
    \end{ssubcase}

    \begin{ssubcase}{$a = 2$.}

      \begin{sssubcase}{$j=2$.}
        Then $\tau_2=1$. Hence, from $s$, the algorithm prescribes
        $\mone{2}$. The level of $r$ is $(4,1,4)$ and, from $r$, the algorithm
        prescribes
        \[
        \hh{1,2,3,4}\mone{2}\mone{3}, \quad \xx{1,4}, \quad \xx{2,3} \quad \mbox{ and
        } \quad \mone{2}.
        \]
        We complete the resulting diagram as follows.
        \[
        \begin{tikzcd}[column sep=large, row sep=large]
        s \arrow[d, Rightarrow,"\mone{2}" left] \arrow[r, "\hhdef"] & r \arrow[d, Rightarrow, "\hh{1,2,3,4}\mone{2} \mone{3}"] \\
        t  \arrow[dddr,"\varepsilon" swap]& q_1 \arrow[d, Rightarrow, "\xx{1,4}"]&\\      
            & q_2 \arrow[d, Rightarrow, "\xx{2,3}"]\\
            & q_3 \arrow[d, Rightarrow, "\mone{2}"]\\
            & q_4 
        \end{tikzcd}      
        \]
        The diagram commutes by \cref{rel:k31,rel:orderx}
        \begin{align*}
        \mone{2}\xx{2,3}\xx{1,4}\hh{1,2,3,4}\mone{2}\mone{3}\hh{1,2,3,4} &\approx
          \mone{2}\xx{2,3}\xx{1,4}\xx{1,4}\xx{2,3}\\
          &\approx \mone{2}.
        \end{align*}
        Moreover, the level property is satisfied since
        $\level(t),\level(q_4) < (2,0,0)=\level(s)$.        
      \end{sssubcase}

      \begin{sssubcase}{$j=3$.}
        Then, from $s$, the algorithm prescribes
        $\xx{2,3}\mone{2}^{\tau_2}$. The level of $r$ is $(4,1,4)$
        and, from $r$, the algorithm prescribes
        \[
        \hh{1,2,3,4}\mone{2}\mone{3}, \quad \xx{1,4}, \quad \mbox{ and
        } \quad \mone{3}^{\tau_2}.
        \]
        We complete the resulting diagram as follows.
        \[
        \begin{tikzcd}[column sep=large, row sep=large]
        s \arrow[d, Rightarrow,"\xx{2,3}\mone{2}^{\tau_2}" left] \arrow[r, "\hhdef"] & r \arrow[d, Rightarrow, "\hh{1,2,3,4}\mone{2} \mone{3}"] \\
        t  \arrow[ddr,"\varepsilon" swap]& q_1 \arrow[d, Rightarrow, "\xx{1,4}"]&\\      
            & q_2 \arrow[d, Rightarrow, "\mone{3}^{\tau_2}"]\\
            & q_3 
        \end{tikzcd}      
        \]
        The diagram commutes by \cref{rel:k31,rel:orderx,rel:rename3}
        \begin{align*}
        \mone{3}^{\tau_2}\xx{1,4}\hh{1,2,3,4}\mone{2}\mone{3}\hh{1,2,3,4} &\approx
          \mone{3}^{\tau_2}\xx{1,4}\xx{1,4}\xx{2,3}\\
          &\approx \mone{3}^{\tau_2}\xx{2,3}\\          
          &\approx \xx{2,3}\mone{2}^{\tau_2}.          
        \end{align*}
        Moreover, the level property is satisfied since
        $\level(t),\level(q_3) < (3,0,0)=\level(s)$.                
      \end{sssubcase}

      \begin{sssubcase}{$j=4$.}
        Then, from $s$, the algorithm prescribes
        $\xx{2,4}\mone{2}^{\tau_2}$. The level of $r$ is $(4,1,4)$
        and, from $r$, the algorithm prescribes
        \[
        \hh{1,2,3,4}\mone{1}^{\tau_2}\mone{2}^{\tau_2+1}\mone{3}^{\tau_2}\mone{4}^{\tau_2+1},
        \quad \mbox{ and } \quad \xx{1,4}.
        \]
        We complete the resulting diagram as follows.
        \[
        \begin{tikzcd}[column sep=large, row sep=large]
        s \arrow[d, Rightarrow,"\xx{2,4}\mone{2}^{\tau_2}" left] \arrow[r, "\hhdef"] & r \arrow[d, Rightarrow, "\hh{1,2,3,4}\mone{1}^{\tau_2}\mone{2}^{\tau_2+1}\mone{3}^{\tau_2}\mone{4}^{\tau_2+1}"] \\
        t  \arrow[dr,"\mone{1}^{\tau_2}\mone{2}^{\tau_2}\mone{3}^{\tau_2}\xx{1,2}\xx{2,3}" swap]& q_1 \arrow[d, Rightarrow, "\xx{1,4}"]&\\      
            & q_2 \\
        \end{tikzcd}      
        \]
        The diagram commutes by
        \cref{rel:k21,rel:orderx,rel:rename1,rel:rename2,rel:rename3}.
        Indeed, when $\tau_2=0$, 
        \begin{align*}
        \xx{1,4}\hh{1,2,3,4}\mone{2}\mone{4}\hh{1,2,3,4} &\approx \xx{1,4}\xx{1,2}\xx{3,4} \\
        &\approx \xx{1,2}\xx{2,3}\xx{2,4}
        \end{align*}
        and when $\tau_2=1$
        \begin{align*}
        \xx{1,4}\hh{1,2,3,4}\mone{1}\mone{3}\hh{1,2,3,4} &\approx \xx{1,4}\xx{1,2}\xx{3,4}\mone{1}\mone{2}\mone{3}\mone{4} \\
        &\approx \mone{1}\mone{2}\mone{3}\xx{1,2}\xx{2,3}\xx{2,4}\mone{2}.
        \end{align*}
        Moreover, the level property is satisfied since
        $\level(t),\level(q_3) < (4,0,0)=\level(s)$ and the extent of
        $\mone{1}^{\tau_2}\mone{2}^{\tau_2}\mone{3}^{\tau_2}\xx{1,2}\xx{2,3}$
        is strictly less than $4$.
      \end{sssubcase}

      \begin{sssubcase}{$j>4$.}
        Then, from $s$, the algorithm prescribes
        $\xx{2,j}\mone{2}^{\tau_2}$. The level of $r$ is $(4,1,4)$
        and, from $r$, the algorithm prescribes
        \[
        \hh{1,2,3,4}\mone{1}^{\tau_2}\mone{2}^{\tau_2+1}\mone{3}^{\tau_2}\mone{4}^{\tau_2+1},
        \quad \mbox{ and } \quad \xx{1,j}.
        \]
        We complete the resulting diagram as follows.
        \[
        \begin{tikzcd}[column sep=large, row sep=large]
        s \arrow[d, Rightarrow,"\xx{2,j}\mone{2}^{\tau_2}" left] \arrow[r, "\hhdef"] & r \arrow[d, Rightarrow, "\hh{1,2,3,4}\mone{1}^{\tau_2}\mone{2}^{\tau_2+1}\mone{3}^{\tau_2}\mone{4}^{\tau_2+1}"] \\
        t  \arrow[dr,"\mone{2}^{\tau_2}\mone{3}^{\tau_2}\mone{4}^{\tau_2}\xx{1,2}\xx{3,4}" swap]& q_1 \arrow[d, Rightarrow, "\xx{1,j}"]&\\      
            & q_2 \\
        \end{tikzcd}      
        \]
        The diagram commutes by \cref{rel:k21,rel:orderx,rel:rename1,rel:rename2,rel:rename3}.
        Indeed, when $\tau_2=0$, 
        \begin{align*}
        \xx{1,j}\hh{1,2,3,4}\mone{2}\mone{4}\hh{1,2,3,4} &\approx \xx{1,j}\xx{1,2}\xx{3,4} \\
        &\approx \xx{2,j}\xx{1,2}\xx{3,4}
        \end{align*}
        and when $\tau_2=1$
        \begin{align*}
        \xx{1,j}\hh{1,2,3,4}\mone{1}\mone{3}\hh{1,2,3,4} &\approx \xx{1,j}\xx{1,2}\xx{3,4}\mone{1}\mone{2}\mone{3}\mone{4} \\
        &\approx \mone{2}\mone{3}\mone{4}\xx{1,2}\xx{3,4}\xx{2,j}\mone{2}.
        \end{align*}
        Moreover, the level property is satisfied since
        $\level(t),\level(q_3) < (j,0,0)=\level(s)$ and the extent of
        $\mone{2}^{\tau_2}\mone{3}^{\tau_2}\mone{4}^{\tau_2}\xx{1,2}\xx{3,4}$
        is strictly less than $j$.
      \end{sssubcase}      
    \end{ssubcase}

    \begin{ssubcase}{$a = 3$.}

      \begin{sssubcase}{$j=3$.}
        Then $\tau_a=1$. Hence, from $s$, the algorithm prescribes
        $\mone{3}$. The level of $r$ is $(4,1,4)$ and, from $r$, the
        algorithm prescribes
        \[
        \hh{1,2,3,4}\mone{2}\mone{3} \quad \xx{1,4} \quad
        \mbox{ and } \quad \xx{2,3}\mone{2}.
        \]
        We complete the resulting diagram as follows.
        \[
        \begin{tikzcd}[column sep=large, row sep=large]
        s \arrow[d, Rightarrow,"\mone{3}" left] \arrow[r, "\hhdef"] & r \arrow[d, Rightarrow, "\hh{1,2,3,4}\mone{2} \mone{3}"] \\
        t  \arrow[ddr,"\varepsilon" swap]& q_1 \arrow[d, Rightarrow, "\xx{1,4}"]&\\      
            & q_2 \arrow[d, Rightarrow, "\xx{2,3}\mone{2}"]\\
            & q_3 \\        
        \end{tikzcd}      
        \]
        The diagram commutes by \cref{rel:k31,rel:orderx,rel:rename3}
        \begin{align*}
        \xx{2,3}\mone{2}\xx{1,4}\hh{1,2,3,4}\mone{2}\mone{3}\hh{1,2,3,4} &\approx \xx{2,3}\mone{2}\xx{1,4}\xx{1,4}\xx{2,3}\\
        &\approx \mone{3}.
        \end{align*}
        Moreover, the level property is satisfied since
        $\level(t),\level(q_3)<(3,0,0)=\level(s)$.
      \end{sssubcase}

      \begin{sssubcase}{$j= 4$.}
        Then, from $s$, the algorithm prescribes
        $\xx{3,4}\mone{3}^{\tau_3}$. The level of $r$ is $(4,1,4)$
        and, from $r$, the algorithm prescribes
        \[
        \hh{1,2,3,4}\mone{1}^{\tau_3} \mone{2}^{\tau_3}
        \mone{3}^{\tau_3+1} \mone{4}^{\tau_3+1} \quad \mbox{ and } \xx{1,4}.
        \]
        We complete the resulting diagrams as follows.
        \[
        \begin{tikzcd}[column sep=large, row sep=large]
        s \arrow[d, Rightarrow,"\xx{3,4}\mone{3}^{\tau_3}" left] \arrow[r, "\hhdef"] & r \arrow[d, Rightarrow, "\hh{1,2,3,4}\mone{1}^{\tau_3} \mone{2}^{\tau_3}
        \mone{3}^{\tau_3+1} \mone{4}^{\tau_3+1}"] \\
        t_1 \arrow[d,"\mone{1}^{\tau_3}\mone{2}^{\tau_3}\mone{3}^{\tau_3}" left]& q_1 \arrow[d, Rightarrow, "\xx{1,4}"]&\\      
        t_2 \arrow[r,"\xx{1,3}\xx{2,3}" swap]  & q_2
        \end{tikzcd}      
        \]
        The diagram commutes by
        \cref{rel:rename1,rel:rename2,rel:orderk,rel:k11,rel:k12,rel:disjoint2,rel:disjoint4}.
        Indeed, when $\tau_3=0$, 
        \begin{align*}
        \xx{1,4}\hh{1,2,3,4}\mone{3}\mone{4}\hh{1,2,3,4} &\approx \xx{1,4}\xx{1,3}\xx{2,4} \\
        &\approx \xx{1,3}\xx{2,3}\xx{3,4}
        \end{align*}
        and when $\tau_3=1$
        \begin{align*}
        \xx{1,4}\hh{1,2,3,4}\mone{1}\mone{2}\hh{1,2,3,4} &\approx \xx{1,4}\xx{1,3}\xx{2,4}\mone{1}\mone{2}\mone{3}\mone{4} \\
        &\approx \xx{1,3}\xx{2,3}\mone{1}\mone{2}\mone{3}\xx{3,4}\mone{3}.
        \end{align*}
        Moreover, the level property is satisfied since
        $\level(t_1),\level(q_2)<(4,0,0)=\level(s)$ and the extent of
        $\xx{1,3}\xx{2,3}\mone{1}^{\tau_3}\mone{2}^{\tau_3}\mone{3}^{\tau_3}$
        is strictly less than $4$.
      \end{sssubcase}    
    
      \begin{sssubcase}{$j>4$.}
        Then, from $s$, the algorithm prescribes
        $\xx{3,j}\mone{3}^{\tau_3}$. The level of $r$ is $(4,1,4)$
        and, from $r$, the algorithm prescribes
        \[
        \hh{1,2,3,4}\mone{1}^{\tau_3} \mone{2}^{\tau_3}
        \mone{3}^{\tau_3+1} \mone{4}^{\tau_3+1} \quad \mbox{ and }
        \quad \xx{1,j}.
        \]
        We complete the resulting diagrams as follows.
        \[
        \begin{tikzcd}[column sep=large, row sep=large]
        s \arrow[d, Rightarrow,"\xx{3,j}\mone{3}^{\tau_3}" left] \arrow[r, "\hhdef"] & r \arrow[d, Rightarrow, "\hh{1,2,3,4}\mone{1}^{\tau_3} \mone{2}^{\tau_3}
        \mone{3}^{\tau_3+1} \mone{4}^{\tau_3+1}"] \\
        t_1 \arrow[d,"\mone{1}^{\tau_3}\mone{2}^{\tau_3}\mone{4}^{\tau_3}" left]& q_1 \arrow[d, Rightarrow, "\xx{1,j}"]&\\      
        t_2 \arrow[r,"\xx{1,3}\xx{2,4}" swap]  & q_2
        \end{tikzcd}      
        \]
        The diagram commutes by
        \cref{rel:rename1,rel:rename2,rel:orderk,rel:k11,rel:k12,rel:disjoint2,rel:disjoint4}.
        Indeed, when $\tau_3=0$, 
        \begin{align*}
        \xx{1,j}\hh{1,2,3,4}\mone{3}\mone{4}\hh{1,2,3,4} &\approx \xx{1,j}\xx{1,3}\xx{2,4} \\
        &\approx \xx{1,3}\xx{2,4}\xx{3,j}
        \end{align*}
        and when $\tau_3=1$
        \begin{align*}
        \xx{1,j}\hh{1,2,3,4}\mone{1}\mone{2}\hh{1,2,3,4} &\approx \xx{1,j}\xx{1,3}\xx{2,4}\mone{1}\mone{2}\mone{3}\mone{4} \\
        &\approx \xx{1,3}\xx{2,4}\mone{1}\mone{2}\mone{4}\xx{3,j}\mone{3}.
        \end{align*}
        Moreover, the level property is satisfied since
        $\level(t_1),\level(q_2)<(j,0,0)=\level(s)$ and the extent of
        $\xx{1,3}\xx{2,4}\mone{1}^{\tau_3}\mone{2}^{\tau_3}\mone{4}^{\tau_3}$
        is strictly less than $j$.
      \end{sssubcase}      
    \end{ssubcase}             

    \begin{ssubcase}{$a = 4$.}

      \begin{sssubcase}{$j=4$.}      
        Then, $\tau_a=1$. Hence, from $s$, the algorithm prescribes
        $\mone{4}$. The level of $r$ is $(4,1,4)$ and, from $r$, the
        algorithm prescribes
        \[
        \hh{1,2,3,4}\mone{1} \mone{4} \quad \mbox{ and } \quad
        \xx{1,j}.
        \]
        We complete the resulting diagram as follows.
        \[
        \begin{tikzcd}[column sep=large, row sep=large]
        s \arrow[d, Rightarrow,"\mone{4}" left] \arrow[r, "\hhdef"] & r \arrow[d, Rightarrow, "\hh{1,2,3,4}\mone{1} \mone{4}"] \\
        t_1 \arrow[d, "\mone{1}\mone{2}\mone{3}" left]& q_1 \arrow[d, Rightarrow, "\xx{1,4}"]&\\
        t_2 \arrow[r,"\xx{2,3}" swap] & q_2
        \end{tikzcd}
        \]
        The diagram commutes by \cref{rel:k32,rel:orderx}
        \begin{align*}
        \xx{1,4}\hh{1,2,3,4}\mone{1}\mone{4}\hh{1,2,3,4} &\approx \xx{1,4}\xx{1,4}\xx{2,3}\mone{1}\mone{2}\mone{3}\mone{4} \\
        &\approx \xx{2,3}\mone{1}\mone{2}\mone{3}\mone{4}.
        \end{align*}
        Moreover, the level property is satisfied since
        $\level(t_1),\level(q_2)<(4,0,0)=\level(s)$ and the extent of
        $\xx{2,3}\mone{1}\mone{2}\mone{3}$ is strictly less than $4$.
      \end{sssubcase}

      \begin{sssubcase}{$j>4$.}              
        Then, from $s$, the algorithm prescribes
        $\xx{4,j}\mone{4}^{\tau_4}$. The level of $r$ is $(4,1,4)$
        and, from $r$, the algorithm prescribes
        \[
        \hh{1,2,3,4}\mone{1}^{\tau_4} \mone{2}^{\tau_4+1}
        \mone{3}^{\tau_4+1} \mone{4}^{\tau_4} \quad \mbox{ and } \quad
        \xx{1,j}.
        \]
        We complete the resulting diagrams as follows.
        \[
        \begin{tikzcd}[column sep=large, row sep=large]
        s \arrow[d, Rightarrow,"\xx{4,j}\mone{4}^{\tau_4}" left] \arrow[r, "\hhdef"] & r \arrow[d, Rightarrow, "\hh{1,2,3,4}\mone{1}^{\tau_4} \mone{2}^{\tau_4+1}
        \mone{3}^{\tau_4+1} \mone{4}^{\tau_4}"] \\
        t_1 \arrow[d,"\mone{1}\mone{2}\mone{3}" left]& q_1 \arrow[d, Rightarrow, "\xx{1,j}"]&\\      
        t_2 \arrow[r,"\xx{2,3}\xx{1,4}" swap]  & q_2
        \end{tikzcd}      
        \]
        The diagram commutes by
        \cref{rel:rename1,rel:rename2,rel:orderk,rel:k31,rel:k32,rel:disjoint2,rel:disjoint4}.
        Indeed, when $\tau_4=0$, 
        \begin{align*}
        \xx{1,j}\hh{1,2,3,4}\mone{2}\mone{3}\hh{1,2,3,4} &\approx \xx{1,j}\xx{1,4}\xx{2,3} \\
        &\approx \xx{2,3}\xx{1,4}\xx{4,j}
        \end{align*}
        and when $\tau_4=1$
        \begin{align*}
        \xx{1,j}\hh{1,2,3,4}\mone{1}\mone{4}\hh{1,2,3,4} &\approx \xx{1,j}\xx{1,4}\xx{2,3}\mone{1}\mone{2}\mone{3}\mone{4} \\
        &\approx \xx{2,3}\xx{1,4}\mone{1}\mone{2}\mone{3}\xx{4,j}\mone{4}.
        \end{align*}
        Moreover, the level property is satisfied since
        $\level(t_1),\level(q_2)<(j,0,0)=\level(s)$ and the extent of
        $\xx{2,3}\xx{1,4}\mone{1}\mone{2}\mone{3}$ is strictly less
        than $j$.
      \end{sssubcase}        
    \end{ssubcase}        

    \begin{ssubcase}{$a>4$.}

      \begin{sssubcase}{$j=a$.}
        Then, $\tau_a=1$ and $v_r=v_s$. Hence, $\level(r)=\level(s)$
        and, from both $s$ and $r$, the algorithm prescribes
        $\mone{a}$. We complete the resulting diagram as follows.
        \[
        \begin{tikzcd}[column sep=large, row sep=large]
          s \arrow[d, Rightarrow,"\mone{a}" left] \arrow[r, "\hhdef"] & r \arrow[d, Rightarrow, "\mone{a}"] \\
          t \arrow[r, "\hhdef" swap] & q
        \end{tikzcd}
        \]
        The diagram commutes by \cref{rel:disjoint5}. And the level property is satisfied since $\level(t),
        \level(q)<\level(s)$.
      \end{sssubcase}

      \begin{sssubcase}{$j>a$.}
        Then, $v_r=v_s$. Hence $\level(r)=\level(s)$ and, from both
        $s$ and $r$, the algorithm prescribes
        $\xx{a,j}\mone{a}^{\tau_a}$. We complete the resulting
        diagrams as follows.
        \[
        \begin{tikzcd}[column sep=large, row sep=large]
          s \arrow[d, Rightarrow,"\xx{a,j}\mone{a}^\tau" left] \arrow[r, "\hhdef"] & r \arrow[d, Rightarrow, "\xx{a,j}\mone{a}^\tau~~"] \\
          t \arrow[r, "\hhdef" swap]  & q
        \end{tikzcd}
        \]
        The diagram commutes by
        \cref{rel:disjoint3,rel:disjoint5}.
        Moreover, the level property is satisfied because
        $\level(t),\level(q)<\level(s)$.
      \end{sssubcase}
    \end{ssubcase}
  \end{subcase}

  \begin{subcase}{$k>0$.}
    Let $u = 2^k v_s$ and let $a,b,c,d$ be the indices of the first
    four odd entries of $u$. In this case, $N$ is of the form
    \[
    \hh{a,b,c,d}\mone{a}^{\tau_a} \mone{b}^{\tau_b}
    \mone{c}^{\tau_c}\mone{d}^{\tau_d},
    \]
    where $\tau_a, \tau_b, \tau_c, \tau_d \in \Z_2$. We have
    $|\s{a,b,c,d} \cap \s{1,2,3,4}| \in \s{0,1,2,3,4}$. We consider
    each one of these cases in turn.

    \begin{ssubcase}{$|\s{a,b,c,d} \cap \s{1,2,3,4}| = 0$.}
      Then $5 \leq a < b < c < d$ so that $u_1 \equiv u_2 \equiv u_3
      \equiv u_4 \equiv 0 \pmod{2}$. Write $\overline{u}$ for the
      vector composed of the first four entries of $u$. Then since all
      of the entries of $\overline{u}$ are even and since the square
      of even number is either $0$ or $4$ modulo 8, we have
      $\overline{u}^\intercal \overline{u}\equiv 0 \pmod{8}$ or
      $\overline{u}^\intercal \overline{u}\equiv 4 \pmod{8}$. We
      consider both of these cases in turn.
              
        \begin{sssubcase}{$\overline{u}^\intercal \overline{u}\equiv 0 \pmod{8}$.}
          Then, by \cref{prop:kevens1}, the first four entries of the
          integral part of the pivot column of $r$ are even. Hence
          $\level(r)=\level(s)$ and, from $r$, the algorithm
          prescribes
          \[
          \hh{a,b,c,d}\mone{a}^{\tau_a} \mone{b}^{\tau_b}
          \mone{c}^{\tau_c}\mone{d}^{\tau_d}.
          \]
          We complete the resulting diagram as follows.
          \[
          \begin{tikzcd}[column sep=large] 
            s \arrow[d, Rightarrow, "\hh{a,b,c,d}\mone{a}^{\tau_a}\mone{b}^{\tau_b}\mone{c}^{\tau_c}\mone{d}^{\tau_d}" left]
            \arrow[r, "\hh{1,2,3,4}"]& 
            r \arrow[d, Rightarrow, "\hh{a,b,c,d}\mone{a}^{\tau_a}\mone{b}^{\tau_b}\mone{c}^{\tau_c}\mone{d}^{\tau_d}" right]\\
            t\arrow[r, "\hh{1,2,3,4}" swap] & q
          \end{tikzcd}
          \]
          The diagram commutes by \cref{rel:disjoint5,rel:disjoint6}
          and the level property is satisfied since
          $\level(t)<\level(s)$ and $\level(q)<\level(r)=\level(s)$.
        \end{sssubcase}
              
        \begin{sssubcase}{$\overline{u}^\intercal \overline{u}\equiv 4 \pmod{8}$.}
          Then, by \cref{prop:kevens2}, the first four entries of the
          pivot column of $r$ are odd. Moreover, by \cref{prop:kevens2}
          evenly many of these entries are congruent to 1 modulo
          4. Hence $\level(r)=(j,k,m+1)$ and, from $r$, the algorithm
          prescribes
          \[
          \hh{1,2,3,4}\mone{1}^{\tau_1'} \mone{2}^{\tau_2'}
          \mone{3}^{\tau_3'}\mone{4}^{\tau_4'} \quad \mbox{and} \quad
          \hh{a,b,c,d}\mone{a}^{\tau_a}\mone{b}^{\tau_b}
          \mone{c}^{\tau_c}\mone{d}^{\tau_d}
          \]
          for some $\tau_1', \tau_2', \tau_3', \tau_4' \in \Z_2$ such
          that evenly many of $\tau_1', \tau_2', \tau_3', \tau_4'$ are
          even. As result, by \cref{prop:kevenk}, there is a word $W$
          over $\s{X, (-1)}$ such that $\extent(W)\leq 4$ and
          \[
          \hh{1,2,3,4}\mone{1}^{\tau_1'} \mone{2}^{\tau_2'}
          \mone{3}^{\tau_3'}\mone{4}^{\tau_4'} \hh{1,2,3,4} \approx W.
          \]
          We complete the diagram as follows.
          \[
          \begin{tikzcd}[column sep=large]
          s \arrow[d, Rightarrow, "\hh{a,b,c,d}\mone{a}^{\tau_a}\mone{b}^{\tau_b}\mone{c}^{\tau_c}\mone{d}^{\tau_d}",swap] \arrow[r, "\hh{1,2,3,4}"] & r \arrow[d, Rightarrow,
            "\hh{1,2,3,4}\mone{1}^{\tau'_1}\mone{2}^{\tau'_2}\mone{3}^{\tau'_3}\mone{4}^{\tau'_4}"
            right]\\
          t \arrow[dr, "W", swap] & q_1 \arrow[d, Rightarrow, "\hh{a,b,c,d}\mone{a}^{\tau_a}\mone{b}^{\tau_b}
          \mone{c}^{\tau_c}\mone{d}^{\tau_d}" right] \\
          & q_2
          \end{tikzcd}
          \]
          The diagram commutes by \cref{rel:disjoint3,rel:disjoint5},
          since
          \[
          \hh{1,2,3,4}\mone{1}^{\tau_1'} \mone{2}^{\tau_2'}
          \mone{3}^{\tau_3'}\mone{4}^{\tau_4'} \hh{1,2,3,4} \approx W.
          \]
          Moreover, the level property is satisfied since
          $\level(t),\level(q_2)<\level(q_1)=\level(s)$ and the level
          of $t$ is invariant under the action of $W$, because $W$ is a
          word over $\s{X,(-1)}$ and $\extent(W)<a$.
        \end{sssubcase}
    \end{ssubcase}

    \begin{ssubcase}{$|\s{a,b,c,d} \cap \s{1,2,3,4}| = 1$.}
      Then $1 \leq a \leq 4$ and $5 \leq b < c < d$. We now consider
      the cases $a=1$, $a=2$, $a=3$, and $a=4$ in turn.

        \begin{sssubcase}{$a=1$.}
          Then, from $s$, the algorithm prescribes
          \[
          \hh{1,b,c,d}\mone{1}^{\tau_1}\mone{b}^{\tau_b}\mone{c}^{\tau_c}\mone{d}^{\tau_d}.
          \]
          Moreover, by \cref{prop:koddodds}, $\level(r)=(j,k+1,1)$
          and, writing $\overline{r}$ for the first four entries of
          the integral part of $r$, we have $\overline{r}\equiv 1111
          \pmod{4}$ or $\overline{r}\equiv 3333 \pmod{4}$. Hence, from
          $r$ the algorithm prescribes
          \[
          \hh{1,2,3,4}\mone{1}^{\tau}\mone{2}^{\tau}\mone{3}^{\tau}\mone{4}^{\tau}
          \]
          where the value of $\tau$ depends on whether
          $\overline{r}\equiv 1111 \pmod{4}$ or $\overline{r}\equiv
          3333 \pmod{4}$. Now, since
          \[
          \hh{1,2,3,4}\mone{1}^{\tau}\mone{2}^{\tau} \mone{3}^{\tau}
          \mone{4}^{\tau}\hh{1,2,3,4}\approx
          \mone{1}^{\tau}\mone{2}^{\tau}\mone{3}^{\tau}\mone{4}^{\tau},
          \]
          by \cref{itm:relk4}, we know that from $q_1=
          (\hh{1,2,3,4}\mone{1}^{\tau}\mone{2}^{\tau}\mone{3}^{\tau}\mone{4}^{\tau})
          r$ the algorithm prescribes
          \[
          \hh{1,b,c,d}\mone{1}^{\tau_1+\tau}\mone{b}^{\tau_b}\mone{c}^{\tau_c}\mone{d}^{\tau_d}.
          \]
          We therefore complete the resulting diagram as follows.
          \[
          \begin{tikzcd}[column sep=large]
          s \arrow[d, Rightarrow, "\hh{1,b,c,d}\mone{1}^{\tau_1}\mone{b}^{\tau_b}\mone{c}^{\tau_c}\mone{d}^{\tau_d}" left] \arrow[r, "\hh{1,2,3,4}"] & r \arrow[d, Rightarrow, "\hh{1,2,3,4}\mone{1}^{\tau}\mone{2}^{\tau}\mone{3}^{\tau}\mone{4}^{\tau}" right]\\
          t\arrow[dr, "\mone{2}^\tau\mone{3}^\tau\mone{4}^\tau", swap] & q_1 \arrow[d, Rightarrow, "\hh{1,b,c,d}\mone{1}^{\tau_1+\tau}\mone{b}^{\tau_b}\mone{c}^{\tau_c}\mone{d}^{\tau_d}" right] \\
          & q_2
          \end{tikzcd}
          \]
          The diagram commutes by
          \cref{itm:relk4,rel:orderk,rel:ordermone,rel:disjoint4,rel:disjoint5}
          \begin{align*}
            \begin{split}
\hh{1,b,c,d}&\mone{1}^{\tau_1+\tau}\mone{b}^{\tau_b}\mone{c}^{\tau_c}\mone{d}^{\tau_d}
          \hh{1,2,3,4}\mone{1}^{\tau}\mone{2}^{\tau} \mone{3}^{\tau} \mone{4}^{\tau}\hh{1,2,3,4} \\             
          &\approx \hh{1,b,c,d}\mone{1}^{\tau_1+\tau}\mone{b}^{\tau_b}\mone{c}^{\tau_c}\mone{d}^{\tau_d}
          \mone{1}^{\tau}\mone{2}^{\tau} \mone{3}^{\tau}
          \mone{4}^{\tau} \\
          &\approx \mone{2}^{\tau} \mone{3}^{\tau}
          \mone{4}^{\tau}\hh{1,b,c,d}\mone{1}^{\tau_1}\mone{b}^{\tau_b}\mone{c}^{\tau_c}\mone{d}^{\tau_d}. 
            \end{split}
          \end{align*}
          Moreover, the level property is satisfied
          since $\level(t)<\level(s)$, $\level(q_2) < \level (q_1) =
          \level(s)$ and $\mone{2}^\tau\mone{3}^\tau\mone{4}^\tau$
          cannot increase the number of odd entries.
        \end{sssubcase}

        \begin{sssubcase}{$a=2$.}
          Then, from $s$, the algorithm prescribes
          \[
          \hh{2,b,c,d}\mone{2}^{\tau_2}\mone{b}^{\tau_b}\mone{c}^{\tau_c}\mone{d}^{\tau_d}.
          \]
          Moreover, by \cref{prop:koddodds}, $\level(r)=(j,k+1,1)$
          and, writing $\overline{r}$ for the first four entries of
          the integral part of $r$, we have $\overline{r}\equiv 1313
          \pmod{4}$ or $\overline{r}\equiv 3131 \pmod{4}$. Hence, from
          $r$ the algorithm prescribes
          \[
          \hh{1,2,3,4}\mone{1}^{\tau}\mone{2}^{\tau+1}\mone{3}^{\tau}\mone{4}^{\tau+1}
          \]
          where the value of $\tau$ depends on whether
          $\overline{r}\equiv 1313 \pmod{4}$ or $\overline{r}\equiv
          3131 \pmod{4}$.  Now, since
          \[
          \hh{1,2,3,4}\mone{1}^{\tau}\mone{2}^{\tau+1} \mone{3}^{\tau}
          \mone{4}^{\tau+1}\hh{1,2,3,4}\approx
          \xx{1,2}\xx{3,4}\mone{1}^{\tau}\mone{2}^{\tau}\mone{3}^{\tau}\mone{4}^{\tau},
          \]
          by \cref{rel:k22,rel:k21}, we know that from $q_1=
          (\hh{1,2,3,4}\mone{1}^{\tau}\mone{2}^{\tau+1}\mone{3}^{\tau}\mone{4}^{\tau+1})
          r$ the algorithm prescribes
          \[
          \hh{1,b,c,d}\mone{1}^{\tau_2+\tau}\mone{b}^{\tau_b}\mone{c}^{\tau_c}\mone{d}^{\tau_d}.
          \]
          We therefore complete the resulting diagram as follows.
          \[
          \begin{tikzcd}[column sep=large]
          s \arrow[d, Rightarrow, "\hh{2,b,c,d}\mone{2}^{\tau_2}\mone{b}^{\tau_b}\mone{c}^{\tau_c}\mone{d}^{\tau_d}" left] \arrow[r, "\hh{1,2,3,4}"] & r \arrow[d, Rightarrow, "\hh{1,2,3,4}\mone{1}^{\tau}\mone{2}^{\tau+1}\mone{3}^{\tau}\mone{4}^{\tau+1}" right]\\
          t\arrow[dr, "\xx{1,2}\xx{3,4}\mone{1}^{\tau}\mone{3}^\tau\mone{4}^{\tau}", swap] & q_1 \arrow[d, Rightarrow, "\hh{1,b,c,d}\mone{1}^{\tau_2+\tau}\mone{b}^{\tau_b}\mone{c}^{\tau_c}\mone{d}^{\tau_d}" right] \\
          & q_2
          \end{tikzcd}
          \]
          The diagram commutes by
          \cref{rel:k22,rel:k21,rel:orderk,rel:ordermone,rel:disjoint2,rel:disjoint3,rel:disjoint4,rel:disjoint5,rel:rename3,rel:rename4}
          \begin{align*}
            \begin{split}
\hh{1,b,c,d}&\mone{1}^{\tau_2+\tau}\mone{b}^{\tau_b}\mone{c}^{\tau_c}\mone{d}^{\tau_d}
          \hh{1,2,3,4}\mone{1}^{\tau}\mone{2}^{\tau+1} \mone{3}^{\tau} \mone{4}^{\tau+1}\hh{1,2,3,4} \\             
          &\approx \hh{1,b,c,d}\mone{1}^{\tau_2+\tau}\mone{b}^{\tau_b}\mone{c}^{\tau_c}\mone{d}^{\tau_d}
           \xx{1,2}\xx{3,4}\mone{1}^{\tau}\mone{2}^{\tau}\mone{3}^{\tau}\mone{4}^{\tau}\\
          &\approx \xx{1,2}\xx{3,4}\mone{1}^{\tau} \mone{3}^{\tau}
          \mone{4}^{\tau}\hh{2,b,c,d}\mone{2}^{\tau_2}\mone{b}^{\tau_b}\mone{c}^{\tau_c}\mone{d}^{\tau_d}. 
            \end{split}
          \end{align*}
          Moreover, the level property is satisfied since
          $\level(t)<\level(s)$, $\level(q_2) < \level (q_1) =
          \level(s)$ and
          $\xx{1,2}\xx{3,4}\mone{1}^\tau\mone{2}^\tau\mone{3}^\tau\mone{4}^\tau$
          cannot increase the number of odd entries .
        \end{sssubcase}
        
        \begin{sssubcase}{$a=3$.}
          Then, from $s$, the algorithm prescribes
          \[
          \hh{3,b,c,d}\mone{3}^{\tau_3}\mone{b}^{\tau_b}\mone{c}^{\tau_c}\mone{d}^{\tau_d}.
          \]
          Moreover, by \cref{prop:koddodds}, $\level(r)=(j,k+1,1)$
          and, writing $\overline{r}$ for the first four entries of
          the integral part of $r$, we have $\overline{r}\equiv 1133
          \pmod{4}$ or $\overline{r}\equiv 3311 \pmod{4}$. Hence, from
          $r$ the algorithm prescribes
          \[
          \hh{1,2,3,4}\mone{1}^{\tau}\mone{2}^{\tau}\mone{3}^{\tau+1}\mone{4}^{\tau+1}
          \]
          where the value of $\tau$ depends on whether
          $\overline{r}\equiv 1133 \pmod{4}$ or $\overline{r}\equiv
          3311 \pmod{4}$.  Now, since
          \[
          \hh{1,2,3,4}\mone{1}^{\tau}\mone{2}^{\tau} \mone{3}^{\tau+1}
          \mone{4}^{\tau+1}\hh{1,2,3,4}\approx
          \xx{1,3}\xx{2,4}\mone{1}^{\tau}\mone{2}^{\tau}\mone{3}^{\tau}\mone{4}^{\tau},
          \]
          by \cref{rel:k12,rel:k11}, we know that from $q_1=
          (\hh{1,2,3,4}\mone{1}^{\tau}\mone{2}^{\tau}\mone{3}^{\tau+1}\mone{4}^{\tau+1})
          r$ the algorithm prescribes
          \[
          \hh{1,b,c,d}\mone{1}^{\tau_3+\tau}\mone{b}^{\tau_b}\mone{c}^{\tau_c}\mone{d}^{\tau_d}.
          \]
          We therefore complete the resulting diagram as follows.
          \[
          \begin{tikzcd}[column sep=large]
          s \arrow[d, Rightarrow, "\hh{3,b,c,d}\mone{3}^{\tau_3}\mone{b}^{\tau_b}\mone{c}^{\tau_c}\mone{d}^{\tau_d}" left] \arrow[r, "\hh{1,2,3,4}"] & r \arrow[d, Rightarrow, "\hh{1,2,3,4}\mone{1}^{\tau}\mone{2}^{\tau}\mone{3}^{\tau+1}\mone{4}^{\tau+1}" right]\\
          t\arrow[dr, "\xx{1,3}\xx{2,4}\mone{1}^{\tau}\mone{2}^\tau\mone{4}^{\tau}", swap] & q_1 \arrow[d, Rightarrow, "\hh{1,b,c,d}\mone{1}^{\tau_3+\tau}\mone{b}^{\tau_b}\mone{c}^{\tau_c}\mone{d}^{\tau_d}" right] \\
          & q_2
          \end{tikzcd}
          \]
          The diagram commutes by
          \cref{rel:k12,rel:k11,rel:orderk,rel:ordermone,rel:disjoint2,rel:disjoint3,rel:disjoint4,rel:disjoint5,rel:rename3,rel:rename4}
          \begin{align*}
            \begin{split}
            \hh{1,b,c,d}&\mone{1}^{\tau_3+\tau}\mone{b}^{\tau_b}\mone{c}^{\tau_c}\mone{d}^{\tau_d}
            \hh{1,2,3,4}\mone{1}^{\tau}\mone{2}^{\tau} \mone{3}^{\tau+1} \mone{4}^{\tau+1}\hh{1,2,3,4} \\             
            &\approx \hh{1,b,c,d}\mone{1}^{\tau_3+\tau}\mone{b}^{\tau_b}\mone{c}^{\tau_c}\mone{d}^{\tau_d}
            \xx{1,3}\xx{2,4}\mone{1}^{\tau}\mone{2}^{\tau}\mone{3}^{\tau}\mone{4}^{\tau}\\
            &\approx \xx{1,3}\xx{2,4}\mone{1}^{\tau} \mone{2}^{\tau}
            \mone{4}^{\tau}\hh{3,b,c,d}\mone{3}^{\tau_3}\mone{b}^{\tau_b}\mone{c}^{\tau_c}\mone{d}^{\tau_d}. 
            \end{split}
          \end{align*}
          Moreover, the level property is satisfied since
          $\level(t)<\level(s)$, $\level(q_2) < \level (q_1) =
          \level(s)$ and $\xx{1,3}\xx{2,4}\mone{1}^{\tau}
          \mone{2}^{\tau} \mone{4}^{\tau}$ cannot increase the number
          of odd entries.
        \end{sssubcase}

        \begin{sssubcase}{$a=4$.}
          Then, from $s$, the algorithm prescribes
          \[
          \hh{4,b,c,d}\mone{4}^{\tau_4}\mone{b}^{\tau_b}\mone{c}^{\tau_c}\mone{d}^{\tau_d}.
          \]
          Moreover, by \cref{prop:koddodds}, $\level(r)=(j,k+1,1)$
          and, writing $\overline{r}$ for the first four entries of
          the integral part of $r$, we have $\overline{r}\equiv 1331
          \pmod{4}$ or $\overline{r}\equiv 3113 \pmod{4}$. Hence, from
          $r$ the algorithm prescribes
          \[
          \hh{1,2,3,4}\mone{1}^{\tau}\mone{2}^{\tau+1}\mone{3}^{\tau+1}\mone{4}^{\tau}
          \]
          where the value of $\tau$ depends on whether
          $\overline{r}\equiv 1331 \pmod{4}$ or $\overline{r}\equiv
          3113 \pmod{4}$.  Now, since
          \[
          \hh{1,2,3,4}\mone{1}^{\tau}\mone{2}^{\tau+1} \mone{3}^{\tau+1}
          \mone{4}^{\tau}\hh{1,2,3,4}\approx
          \xx{1,4}\xx{2,3}\mone{1}^{\tau}\mone{2}^{\tau}\mone{3}^{\tau}\mone{4}^{\tau},
          \]
          by \cref{rel:k32,rel:k31}, we know that from $q_1=
          (\hh{1,2,3,4}\mone{1}^{\tau}\mone{2}^{\tau+1}\mone{3}^{\tau+1}\mone{4}^{\tau})
          r$ the algorithm prescribes
          \[
          \hh{1,b,c,d}\mone{1}^{\tau_4+\tau}\mone{b}^{\tau_b}\mone{c}^{\tau_c}\mone{d}^{\tau_d}.
          \]
          We therefore complete the resulting diagram as follows.
          \[
          \begin{tikzcd}[column sep=large]
          s \arrow[d, Rightarrow, "\hh{4,b,c,d}\mone{4}^{\tau_4}\mone{b}^{\tau_b}\mone{c}^{\tau_c}\mone{d}^{\tau_d}" left] \arrow[r, "\hh{1,2,3,4}"] & r \arrow[d, Rightarrow, "\hh{1,2,3,4}\mone{1}^{\tau}\mone{2}^{\tau+1}\mone{3}^{\tau+1}\mone{4}^{\tau}" right]\\
          t\arrow[dr, "\xx{1,4}\xx{2,3}\mone{1}^{\tau}\mone{2}^\tau\mone{3}^{\tau}", swap] & q_1 \arrow[d, Rightarrow, "\hh{1,b,c,d}\mone{1}^{\tau_4+\tau}\mone{b}^{\tau_b}\mone{c}^{\tau_c}\mone{d}^{\tau_d}" right] \\
          & q_2
          \end{tikzcd}
          \]
          The diagram commutes by
          \cref{rel:k32,rel:k31,rel:orderk,rel:ordermone,rel:disjoint2,rel:disjoint3,rel:disjoint4,rel:disjoint5,rel:rename3,rel:rename4}
          \begin{align*}
          \begin{split}
          \hh{1,b,c,d}&\mone{1}^{\tau_4+\tau}\mone{b}^{\tau_b}\mone{c}^{\tau_c}\mone{d}^{\tau_d}
          \hh{1,2,3,4}\mone{1}^{\tau}\mone{2}^{\tau+1} \mone{3}^{\tau+1} \mone{4}^{\tau}\hh{1,2,3,4} \\             
          &\approx \hh{1,b,c,d}\mone{1}^{\tau_4+\tau}\mone{b}^{\tau_b}\mone{c}^{\tau_c}\mone{d}^{\tau_d}
           \xx{1,4}\xx{2,3}\mone{1}^{\tau}\mone{2}^{\tau}\mone{3}^{\tau}\mone{4}^{\tau}\\
          &\approx \xx{1,4}\xx{2,3}\mone{1}^{\tau} \mone{2}^{\tau}
          \mone{3}^{\tau}\hh{4,b,c,d}\mone{4}^{\tau_4}\mone{b}^{\tau_b}\mone{c}^{\tau_c}\mone{d}^{\tau_d}. 
          \end{split}
          \end{align*}
          Moreover, the level property is satisfied since
          $\level(t)<\level(s)$, $\level(q_2) < \level (q_1) =
          \level(s)$ and
          $\xx{1,2}\xx{3,4}\mone{1}^\tau\mone{2}^\tau\mone{3}^\tau\mone{4}^\tau$
          cannot increase the number of odd entries .
        \end{sssubcase}
      \end{ssubcase}
      
    \begin{ssubcase}{$|\s{a,b,c,d} \cap \s{1,2,3,4}| = 2$.}
      Then $a,b\in \s{1,2,3,4}$ and $5\leq c<d$. We now consider the
      cases $\s{a,b}=\s{1,2}$, $\s{a,b}=\s{1,3}$, $\s{a,b}=\s{1,4}$,
      $\s{a,b}=\s{2,3}$, $\s{a,b}=\s{2,4}$, and $\s{a,b}=\s{3,4}$ in
      turn.

        \begin{sssubcase}{$\s{a,b} = \s{1,2}$.}
          Then, from $s$, the algorithm prescribes
          \[
          \hh{1,2,c,d}\mone{1}^{\tau_1}\mone{2}^{\tau_2}\mone{c}^{\tau_c}\mone{d}^{\tau_d}.
          \]
          Moreover, by \cref{prop:kevenodds}, $\level(r)=\level(s)$
          and, writing $\overline{r}$ for the first four entries of
          the integral part of $r$, we have $\overline{r}\equiv 1010
          \pmod{2}$ or $\overline{r}\equiv 0101 \pmod{2}$. We consider
          both cases in turn.

          \begin{ssssubcase}{$\overline{r}\equiv 0101\pmod{2}$.}
            In this case, from $r$, the algorithm prescribes
            \[
            \hh{1,3,c,d}\mone{1}^{\tau_1'}\mone{3}^{\tau_3'}\mone{c}^{\tau_c}\mone{d}^{\tau_d}.
            \]
            By \cref{prop:rewriterules}, there exist words
            $\word{V}$ and $\word{W}$ over $\s{\mone{x},\xx{x,y}}$,
            with $x,y\in\s{1,2,3,4,c,d}$, such that
            \[
            \hh{1,3,c,d}\mone{1}^{\tau_1'}\mone{3}^{\tau_3'}\mone{c}^{\tau_c}\mone{d}^{\tau_d}
            \hh{1,2,3,4} \mone{d}^{\tau_d}
            \mone{c}^{\tau_c}\mone{2}^{\tau_2}\mone{1}^{\tau_1}\hh{1,2,c,d}
            \approx \word{V}\hh{1,2,3,4}\word{W}.
            \]
            Hence, we can complete the diagram as follows.
            \[
              \begin{tikzcd}[column sep=2.5cm, row sep=large]
                s \arrow[d, Rightarrow, "\hh{1,2,c,d}\mone{1}^{\tau_1}\mone{2}^{\tau_2}\mone{c}^{\tau_c}\mone{d}^{\tau_d}" left] 
                  \arrow[r, "\hh{1,2,3,4}"] & 
                  r\arrow[d, Rightarrow, "\hh{1,3,c,d}\mone{1}^{\tau_1'}\mone{3}^{\tau_3'}\mone{c}^{\tau_c}\mone{d}^{\tau_d}"] \\
                t\arrow[r, "\word{V}\hh{1,2,3,4}\word{W}", swap] & q
              \end{tikzcd}
            \]
            The diagram commutes by construction. To see that the
            level property is satisfied, first note that
            \[
            \level(t),\level(q)\leq (j,k,m-1)<\level(s). 
            \]
            Now, since
            $\word{V}$ and $\word{W}$ are words over
            $\s{\mone{x},\xx{x,y}}$, they can neither increase nor
            decrease the number of odd entries. As a result, because
            $\word{V}\hh{1,2,3,4}\word{W}$ contains a single
            occurrence of $K$, it cannot raise the level of the state
            to $\level(s)$ and also lower it back to
            $\level(q)$. Thus,
            \[
            \level(\word{V}\hh{1,2,3,4}\word{W})<\level(s)
            \]
            as desired.
          \end{ssssubcase}

          \begin{ssssubcase}{$\overline{r}\equiv 0101\pmod{2}$.}
            In this case, from $r$, the algorithm prescribes
            \[
            \hh{2,4,c,d}\mone{2}^{\tau_2'}\mone{4}^{\tau_4'}\mone{c}^{\tau_c}\mone{d}^{\tau_d}.
            \]
            By \cref{prop:rewriterules}, there exist words
            $\word{V}$ and $\word{W}$ over $\s{\mone{x},\xx{x,y}}$,
            with $x,y\in\s{1,2,3,4,c,d}$, such that
            \[
            \hh{2,4,c,d}\mone{2}^{\tau_2'}\mone{4}^{\tau_4'}\mone{c}^{\tau_c}\mone{d}^{\tau_d}
            \hh{1,2,3,4} \mone{d}^{\tau_d}
            \mone{c}^{\tau_c}\mone{2}^{\tau_2}\mone{1}^{\tau_1}\hh{1,2,c,d}
            \approx \word{V}\hh{1,2,3,4}\word{W}.
            \]
            Hence, we can complete the diagram as follows.
            \[
              \begin{tikzcd}[column sep=2.5cm, row sep=large]
                s \arrow[d, Rightarrow, "\hh{1,2,c,d}\mone{1}^{\tau_1}\mone{2}^{\tau_2}\mone{c}^{\tau_c}\mone{d}^{\tau_d}" left] 
                  \arrow[r, "\hh{1,2,3,4}"] & 
                  r\arrow[d, Rightarrow, "\hh{2,4,c,d}\mone{2}^{\tau_2'}\mone{4}^{\tau_4'}\mone{c}^{\tau_c}\mone{d}^{\tau_d}"] \\
                t\arrow[r, "\word{V}\hh{1,2,3,4}\word{W}", swap] & q
              \end{tikzcd}
            \]
            The diagram commutes by construction. To see that the
            level property is satisfied, first note that
            \[
            \level(t),\level(q)\leq (j,k,m-1)<\level(s).
            \]
            Now, since
            $\word{V}$ and $\word{W}$ are words over
            $\s{\mone{x},\xx{x,y}}$, they can neither increase nor
            decrease the number of odd entries. As a result, because
            $\word{V}\hh{1,2,3,4}\word{W}$ contains a single
            occurrence of $K$, it cannot raise the level of the state
            to $\level(s)$ and also lower it back to
            $\level(q)$. Thus,
            \[
            \level(\word{V}\hh{1,2,3,4}\word{W})<\level(s)
            \]
            as desired.
          \end{ssssubcase}          
        \end{sssubcase}      

        \begin{sssubcase}{$\s{a,b} = \s{1,3}$.}
          Then, from $s$, the algorithm prescribes
          \[
          \hh{1,3,c,d}\mone{1}^{\tau_1}\mone{3}^{\tau_3}\mone{c}^{\tau_c}\mone{d}^{\tau_d}.
          \]
          Moreover, by \cref{prop:kevenodds}, $\level(r)=\level(s)$
          and, writing $\overline{r}$ for the first four entries of
          the integral part of $r$, we have $\overline{r}\equiv 1100
          \pmod{2}$ or $\overline{r}\equiv 0011 \pmod{2}$. We consider
          both cases in turn.

          \begin{ssssubcase}{$\overline{r}\equiv 1100\pmod{2}$.}
            In this case, from $r$, the algorithm prescribes
            \[
            \hh{1,2,c,d}\mone{1}^{\tau_1'}\mone{2}^{\tau_2'}\mone{c}^{\tau_c}\mone{d}^{\tau_d}.
            \]
            By \cref{prop:rewriterules}, there exist words
            $\word{V}$ and $\word{W}$ over $\s{\mone{x},\xx{x,y}}$,
            with $x,y\in\s{1,2,3,4,c,d}$, such that
            \[
            \hh{1,2,c,d}\mone{1}^{\tau_1'}\mone{2}^{\tau_2'}\mone{c}^{\tau_c}\mone{d}^{\tau_d}
            \hh{1,2,3,4} \mone{d}^{\tau_d}
            \mone{c}^{\tau_c}\mone{3}^{\tau_3}\mone{1}^{\tau_1}\hh{1,3,c,d}
            \approx \word{V}\hh{1,2,3,4}\word{W}.
            \]
            Hence, we can complete the diagram as follows.
            \[
              \begin{tikzcd}[column sep=2.5cm, row sep=large]
                s \arrow[d, Rightarrow, "\hh{1,3,c,d}\mone{1}^{\tau_1}\mone{3}^{\tau_3}\mone{c}^{\tau_c}\mone{d}^{\tau_d}" left] 
                  \arrow[r, "\hh{1,2,3,4}"] & 
                  r\arrow[d, Rightarrow, "\hh{1,2,c,d}\mone{1}^{\tau_1'}\mone{2}^{\tau_2'}\mone{c}^{\tau_c}\mone{d}^{\tau_d}"] \\
                t\arrow[r, "\word{V}\hh{1,2,3,4}\word{W}", swap] & q
              \end{tikzcd}
            \]
            The diagram commutes by construction. To see that the
            level property is satisfied, first note that
            \[
            \level(t),\level(q)\leq (j,k,m-1)<\level(s).
            \]
            Now, since
            $\word{V}$ and $\word{W}$ are words over
            $\s{\mone{x},\xx{x,y}}$, they can neither increase nor
            decrease the number of odd entries. As a result, because
            $\word{V}\hh{1,2,3,4}\word{W}$ contains a single
            occurrence of $K$, it cannot raise the level of the state
            to $\level(s)$ and also lower it back to
            $\level(q)$. Thus,
            \[
            \level(\word{V}\hh{1,2,3,4}\word{W})<\level(s)
            \]
            as desired.
          \end{ssssubcase}

          \begin{ssssubcase}{$\overline{r}\equiv 0011\pmod{2}$.}
            In this case, from $r$, the algorithm prescribes
            \[
            \hh{3,4,c,d}\mone{3}^{\tau_3'}\mone{4}^{\tau_4'}\mone{c}^{\tau_c}\mone{d}^{\tau_d}.
            \]
            By \cref{prop:rewriterules}, there exist words
            $\word{V}$ and $\word{W}$ over $\s{\mone{x},\xx{x,y}}$,
            with $x,y\in\s{1,2,3,4,c,d}$, such that
            \[
            \hh{3,4,c,d}\mone{3}^{\tau_3'}\mone{4}^{\tau_4'}\mone{c}^{\tau_c}\mone{d}^{\tau_d}
            \hh{1,2,3,4} \mone{d}^{\tau_d}
            \mone{c}^{\tau_c}\mone{3}^{\tau_3}\mone{1}^{\tau_1}\hh{1,3,c,d}
            \approx \word{V}\hh{1,2,3,4}\word{W}.
            \]
            Hence, we can complete the diagram as follows.
            \[
              \begin{tikzcd}[column sep=2.5cm, row sep=large]
                s \arrow[d, Rightarrow, "\hh{1,3,c,d}\mone{1}^{\tau_1}\mone{3}^{\tau_3}\mone{c}^{\tau_c}\mone{d}^{\tau_d}" left] 
                  \arrow[r, "\hh{1,2,3,4}"] & 
                  r\arrow[d, Rightarrow, "\hh{3,4,c,d}\mone{3}^{\tau_3'}\mone{4}^{\tau_4'}\mone{c}^{\tau_c}\mone{d}^{\tau_d}"] \\
                t\arrow[r, "\word{V}\hh{1,2,3,4}\word{W}", swap] & q
              \end{tikzcd}
            \]
            The diagram commutes by construction. To see that the
            level property is satisfied, first note that
            \[
            \level(t),\level(q)\leq (j,k,m-1)<\level(s).
            \]
            Now, since
            $\word{V}$ and $\word{W}$ are words over
            $\s{\mone{x},\xx{x,y}}$, they can neither increase nor
            decrease the number of odd entries. As a result, because
            $\word{V}\hh{1,2,3,4}\word{W}$ contains a single
            occurrence of $K$, it cannot raise the level of the state
            to $\level(s)$ and also lower it back to
            $\level(q)$. Thus,
            \[
            \level(\word{V}\hh{1,2,3,4}\word{W})<\level(s)
            \]
            as desired.
          \end{ssssubcase}          
        \end{sssubcase}      

        \begin{sssubcase}{$\s{a,b} = \s{1,4}$.}
          Then, from $s$, the algorithm prescribes
          \[
          \hh{1,4,c,d}\mone{1}^{\tau_1}\mone{4}^{\tau_4}\mone{c}^{\tau_c}\mone{d}^{\tau_d}.
          \]
          Moreover, by \cref{prop:kevenodds}, $\level(r)=\level(s)$
          and, writing $\overline{r}$ for the first four entries of
          the integral part of $r$, we have $\overline{r}\equiv 1001
          \pmod{2}$ or $\overline{r}\equiv 0110 \pmod{2}$. We consider
          both cases in turn.

          \begin{ssssubcase}{$\overline{r}\equiv 1001\pmod{2}$.}
            In this case, from $r$, the algorithm prescribes
            \[
            \hh{1,4,c,d}\mone{1}^{\tau_1'}\mone{4}^{\tau_4'}\mone{c}^{\tau_c}\mone{d}^{\tau_d}.
            \]
            By \cref{prop:rewriterules}, there exist words
            $\word{V}$ and $\word{W}$ over $\s{\mone{x},\xx{x,y}}$,
            with $x,y\in\s{1,2,3,4,c,d}$, such that
            \[
            \hh{1,4,c,d}\mone{1}^{\tau_1'}\mone{4}^{\tau_4'}\mone{c}^{\tau_c}\mone{d}^{\tau_d}
            \hh{1,2,3,4} \mone{d}^{\tau_d}
            \mone{c}^{\tau_c}\mone{4}^{\tau_4}\mone{1}^{\tau_1}\hh{1,4,c,d}
            \approx \word{V}\hh{1,2,3,4}\word{W}.
            \]
            Hence, we can complete the diagram as follows.
            \[
              \begin{tikzcd}[column sep=2.5cm, row sep=large]
                s \arrow[d, Rightarrow, "\hh{1,4,c,d}\mone{1}^{\tau_1}\mone{4}^{\tau_4}\mone{c}^{\tau_c}\mone{d}^{\tau_d}" left] 
                  \arrow[r, "\hh{1,2,3,4}"] & 
                  r\arrow[d, Rightarrow, "\hh{1,4,c,d}\mone{1}^{\tau_1'}\mone{4}^{\tau_4'}\mone{c}^{\tau_c}\mone{d}^{\tau_d}"] \\
                t\arrow[r, "\word{V}\hh{1,2,3,4}\word{W}", swap] & q
              \end{tikzcd}
            \]
            The diagram commutes by construction. To see that the
            level property is satisfied, first note that
            \[
            \level(t),\level(q)\leq (j,k,m-1)<\level(s).
            \] 
            Now, since
            $\word{V}$ and $\word{W}$ are words over
            $\s{\mone{x},\xx{x,y}}$, they can neither increase nor
            decrease the number of odd entries. As a result, because
            $\word{V}\hh{1,2,3,4}\word{W}$ contains a single
            occurrence of $K$, it cannot raise the level of the state
            to $\level(s)$ and also lower it back to
            $\level(q)$. Thus,
            \[
            \level(\word{V}\hh{1,2,3,4}\word{W})<\level(s)
            \]
            as desired.
          \end{ssssubcase}

          \begin{ssssubcase}{$\overline{r}\equiv 0110\pmod{2}$.}
            In this case, from $r$, the algorithm prescribes
            \[
            \hh{2,3,c,d}\mone{2}^{\tau_2'}\mone{3}^{\tau_3'}\mone{c}^{\tau_c}\mone{d}^{\tau_d}.
            \]
            By \cref{prop:rewriterules}, there exist words
            $\word{V}$ and $\word{W}$ over $\s{\mone{x},\xx{x,y}}$,
            with $x,y\in\s{1,2,3,4,c,d}$, such that
            \[
            \hh{2,3,c,d}\mone{2}^{\tau_2'}\mone{3}^{\tau_3'}\mone{c}^{\tau_c}\mone{d}^{\tau_d}
            \hh{1,2,3,4} \mone{d}^{\tau_d}
            \mone{c}^{\tau_c}\mone{4}^{\tau_4}\mone{1}^{\tau_1}\hh{1,4,c,d}
            \approx \word{V}\hh{1,2,3,4}\word{W}.
            \]
            Hence, we can complete the diagram as follows.
            \[
              \begin{tikzcd}[column sep=2.5cm, row sep=large]
                s \arrow[d, Rightarrow, "\hh{1,4,c,d}\mone{1}^{\tau_1}\mone{4}^{\tau_4}\mone{c}^{\tau_c}\mone{d}^{\tau_d}" left] 
                  \arrow[r, "\hh{1,2,3,4}"] & 
                  r\arrow[d, Rightarrow, "\hh{2,3,c,d}\mone{2}^{\tau_2'}\mone{3}^{\tau_3'}\mone{c}^{\tau_c}\mone{d}^{\tau_d}"] \\
                t\arrow[r, "\word{V}\hh{1,2,3,4}\word{W}", swap] & q
              \end{tikzcd}
            \]
            The diagram commutes by construction. To see that the
            level property is satisfied, first note that
            \[
            \level(t),\level(q)\leq (j,k,m-1)<\level(s). 
            \]
            Now, since
            $\word{V}$ and $\word{W}$ are words over
            $\s{\mone{x},\xx{x,y}}$, they can neither increase nor
            decrease the number of odd entries. As a result, because
            $\word{V}\hh{1,2,3,4}\word{W}$ contains a single
            occurrence of $K$, it cannot raise the level of the state
            to $\level(s)$ and also lower it back to
            $\level(q)$. Thus,
            \[
            \level(\word{V}\hh{1,2,3,4}\word{W})<\level(s)
            \]
            as desired.
          \end{ssssubcase}
        \end{sssubcase}      

        \begin{sssubcase}{$\s{a,b} = \s{2,3}$.}
          Then, from $s$, the algorithm prescribes
          \[
          \hh{2,3,c,d}\mone{2}^{\tau_2}\mone{3}^{\tau_3}\mone{c}^{\tau_c}\mone{d}^{\tau_d}.
          \]
          Moreover, by \cref{prop:kevenodds}, $\level(r)=\level(s)$
          and, writing $\overline{r}$ for the first four entries of
          the integral part of $r$, we have $\overline{r}\equiv 1001
          \pmod{2}$ or $\overline{r}\equiv 0110 \pmod{2}$. We consider
          both cases in turn.

          \begin{ssssubcase}{$\overline{r}\equiv 1001\pmod{2}$.}
            In this case, from $r$, the algorithm prescribes
            \[
            \hh{1,4,c,d}\mone{1}^{\tau_1'}\mone{4}^{\tau_4'}\mone{c}^{\tau_c}\mone{d}^{\tau_d}.
            \]
            By \cref{prop:rewriterules}, there exist words
            $\word{V}$ and $\word{W}$ over $\s{\mone{x},\xx{x,y}}$,
            with $x,y\in\s{1,2,3,4,c,d}$, such that
            \[
            \hh{1,4,c,d}\mone{1}^{\tau_1'}\mone{4}^{\tau_4'}\mone{c}^{\tau_c}\mone{d}^{\tau_d}
            \hh{1,2,3,4} \mone{d}^{\tau_d}
            \mone{c}^{\tau_c}\mone{3}^{\tau_3}\mone{2}^{\tau_2}\hh{2,3,c,d}
            \approx \word{V}\hh{1,2,3,4}\word{W}.
            \]
            Hence, we can complete the diagram as follows.
            \[
              \begin{tikzcd}[column sep=2.5cm, row sep=large]
                s \arrow[d, Rightarrow, "\hh{2,3,c,d}\mone{2}^{\tau_2}\mone{3}^{\tau_3}\mone{c}^{\tau_c}\mone{d}^{\tau_d}" left] 
                  \arrow[r, "\hh{1,2,3,4}"] & 
                  r\arrow[d, Rightarrow, "\hh{1,4,c,d}\mone{1}^{\tau_1'}\mone{4}^{\tau_4'}\mone{c}^{\tau_c}\mone{d}^{\tau_d}"] \\
                t\arrow[r, "\word{V}\hh{1,2,3,4}\word{W}", swap] & q
              \end{tikzcd}
            \]
            The diagram commutes by construction. To see that the
            level property is satisfied, first note that
            \[
            \level(t),\level(q)\leq (j,k,m-1)<\level(s). 
            \]
            Now, since
            $\word{V}$ and $\word{W}$ are words over
            $\s{\mone{x},\xx{x,y}}$, they can neither increase nor
            decrease the number of odd entries. As a result, because
            $\word{V}\hh{1,2,3,4}\word{W}$ contains a single
            occurrence of $K$, it cannot raise the level of the state
            to $\level(s)$ and also lower it back to
            $\level(q)$. Thus,
            \[
            \level(\word{V}\hh{1,2,3,4}\word{W})<\level(s)
            \]
            as desired.
          \end{ssssubcase}

          \begin{ssssubcase}{$\overline{r}\equiv 0110\pmod{2}$.}
            In this case, from $r$, the algorithm prescribes
            \[
            \hh{2,3,c,d}\mone{2}^{\tau_2'}\mone{3}^{\tau_3'}\mone{c}^{\tau_c}\mone{d}^{\tau_d}.
            \]
            By \cref{prop:rewriterules}, there exist words
            $\word{V}$ and $\word{W}$ over $\s{\mone{x},\xx{x,y}}$,
            with $x,y\in\s{1,2,3,4,c,d}$, such that
            \[
            \hh{2,3,c,d}\mone{2}^{\tau_2'}\mone{3}^{\tau_3'}\mone{c}^{\tau_c}\mone{d}^{\tau_d}
            \hh{1,2,3,4} \mone{d}^{\tau_d}
            \mone{c}^{\tau_c}\mone{3}^{\tau_3}\mone{2}^{\tau_2}\hh{2,3,c,d}
            \approx \word{V}\hh{1,2,3,4}\word{W}.
            \]
            Hence, we can complete the diagram as follows.
            \[
              \begin{tikzcd}[column sep=2.5cm, row sep=large]
                s \arrow[d, Rightarrow, "\hh{2,3,c,d}\mone{2}^{\tau_2}\mone{3}^{\tau_3}\mone{c}^{\tau_c}\mone{d}^{\tau_d}" left] 
                  \arrow[r, "\hh{1,2,3,4}"] & 
                  r\arrow[d, Rightarrow, "\hh{2,3,c,d}\mone{2}^{\tau_2'}\mone{3}^{\tau_3'}\mone{c}^{\tau_c}\mone{d}^{\tau_d}"] \\
                t\arrow[r, "\word{V}\hh{1,2,3,4}\word{W}", swap] & q
              \end{tikzcd}
            \]
            The diagram commutes by construction. To see that the
            level property is satisfied, first note that
            \[
            \level(t),\level(q)\leq (j,k,m-1)<\level(s). 
            \]
            Now, since
            $\word{V}$ and $\word{W}$ are words over
            $\s{\mone{x},\xx{x,y}}$, they can neither increase nor
            decrease the number of odd entries. As a result, because
            $\word{V}\hh{1,2,3,4}\word{W}$ contains a single
            occurrence of $K$, it cannot raise the level of the state
            to $\level(s)$ and also lower it back to
            $\level(q)$. Thus,
            \[
            \level(\word{V}\hh{1,2,3,4}\word{W})<\level(s)
            \]
            as desired.
          \end{ssssubcase}
        \end{sssubcase}      

        \begin{sssubcase}{$\s{a,b} = \s{2,4}$.}
          Then, from $s$, the algorithm prescribes
          \[
          \hh{2,4,c,d}\mone{2}^{\tau_2}\mone{4}^{\tau_4}\mone{c}^{\tau_c}\mone{d}^{\tau_d}.
          \]
          Moreover, by \cref{prop:kevenodds}, $\level(r)=\level(s)$
          and, writing $\overline{r}$ for the first four entries of
          the integral part of $r$, we have $\overline{r}\equiv 1100
          \pmod{2}$ or $\overline{r}\equiv 0011 \pmod{2}$. We consider
          both cases in turn.

          \begin{ssssubcase}{$\overline{r}\equiv 1100\pmod{2}$.}
            In this case, from $r$, the algorithm prescribes
            \[
            \hh{1,2,c,d}\mone{1}^{\tau_1'}\mone{2}^{\tau_2'}\mone{c}^{\tau_c}\mone{d}^{\tau_d}.
            \]
            By \cref{prop:rewriterules}, there exist words
            $\word{V}$ and $\word{W}$ over $\s{\mone{x},\xx{x,y}}$,
            with $x,y\in\s{1,2,3,4,c,d}$, such that
            \[
            \hh{1,2,c,d}\mone{1}^{\tau_1'}\mone{2}^{\tau_2'}\mone{c}^{\tau_c}\mone{d}^{\tau_d}
            \hh{1,2,3,4} \mone{d}^{\tau_d}
            \mone{c}^{\tau_c}\mone{4}^{\tau_4}\mone{2}^{\tau_2}\hh{2,4,c,d}
            \approx \word{V}\hh{1,2,3,4}\word{W}.
            \]
            Hence, we can complete the diagram as follows.
            \[
              \begin{tikzcd}[column sep=2.5cm, row sep=large]
                s \arrow[d, Rightarrow, "\hh{2,4,c,d}\mone{2}^{\tau_2}\mone{4}^{\tau_4}\mone{c}^{\tau_c}\mone{d}^{\tau_d}" left] 
                  \arrow[r, "\hh{1,2,3,4}"] & 
                  r\arrow[d, Rightarrow, "\hh{1,2,c,d}\mone{1}^{\tau_1'}\mone{2}^{\tau_2'}\mone{c}^{\tau_c}\mone{d}^{\tau_d}"] \\
                t\arrow[r, "\word{V}\hh{1,2,3,4}\word{W}", swap] & q
              \end{tikzcd}
            \]
            The diagram commutes by construction. To see that the
            level property is satisfied, first note that
            \[
            \level(t),\level(q)\leq (j,k,m-1)<\level(s).
            \]
            Now, since
            $\word{V}$ and $\word{W}$ are words over
            $\s{\mone{x},\xx{x,y}}$, they can neither increase nor
            decrease the number of odd entries. As a result, because
            $\word{V}\hh{1,2,3,4}\word{W}$ contains a single
            occurrence of $K$, it cannot raise the level of the state
            to $\level(s)$ and also lower it back to
            $\level(q)$. Thus,
            \[
            \level(\word{V}\hh{1,2,3,4}\word{W})<\level(s)
            \]
            as desired.
          \end{ssssubcase}

          \begin{ssssubcase}{$\overline{r}\equiv 0011\pmod{2}$.}
            In this case, from $r$, the algorithm prescribes
            \[
            \hh{3,4,c,d}\mone{3}^{\tau_3'}\mone{4}^{\tau_4'}\mone{c}^{\tau_c}\mone{d}^{\tau_d}.
            \]
            By \cref{prop:rewriterules}, there exist words
            $\word{V}$ and $\word{W}$ over $\s{\mone{x},\xx{x,y}}$,
            with $x,y\in\s{1,2,3,4,c,d}$, such that
            \[
            \hh{3,4,c,d}\mone{3}^{\tau_3'}\mone{4}^{\tau_4'}\mone{c}^{\tau_c}\mone{d}^{\tau_d}
            \hh{1,2,3,4} \mone{d}^{\tau_d}
            \mone{c}^{\tau_c}\mone{4}^{\tau_4}\mone{2}^{\tau_2}\hh{2,4,c,d}
            \approx \word{V}\hh{1,2,3,4}\word{W}.
            \]
            Hence, we can complete the diagram as follows.
            \[
              \begin{tikzcd}[column sep=2.5cm, row sep=large]
                s \arrow[d, Rightarrow, "\hh{2,4,c,d}\mone{2}^{\tau_2}\mone{4}^{\tau_4}\mone{c}^{\tau_c}\mone{d}^{\tau_d}" left] 
                  \arrow[r, "\hh{1,2,3,4}"] & 
                  r\arrow[d, Rightarrow, "\hh{3,4,c,d}\mone{3}^{\tau_3'}\mone{4}^{\tau_4'}\mone{c}^{\tau_c}\mone{d}^{\tau_d}"] \\
                t\arrow[r, "\word{V}\hh{1,2,3,4}\word{W}", swap] & q
              \end{tikzcd}
            \]
            The diagram commutes by construction. To see that the
            level property is satisfied, first note that
            \[
            \level(t),\level(q)\leq (j,k,m-1)<\level(s). 
            \]
            Now, since
            $\word{V}$ and $\word{W}$ are words over
            $\s{\mone{x},\xx{x,y}}$, they can neither increase nor
            decrease the number of odd entries. As a result, because
            $\word{V}\hh{1,2,3,4}\word{W}$ contains a single
            occurrence of $K$, it cannot raise the level of the state
            to $\level(s)$ and also lower it back to
            $\level(q)$. Thus,
            \[
            \level(\word{V}\hh{1,2,3,4}\word{W})<\level(s)
            \]
            as desired.
          \end{ssssubcase}          
        \end{sssubcase}      

        \begin{sssubcase}{$\s{a,b} = \s{3,4}$.}
          Then, from $s$, the algorithm prescribes
          \[
          \hh{3,4,c,d}\mone{3}^{\tau_3}\mone{4}^{\tau_4}\mone{c}^{\tau_c}\mone{d}^{\tau_d}.
          \]
          Moreover, by \cref{prop:kevenodds}, $\level(r)=\level(s)$
          and, writing $\overline{r}$ for the first four entries of
          the integral part of $r$, we have $\overline{r}\equiv 1010
          \pmod{2}$ or $\overline{r}\equiv 0101 \pmod{2}$. We consider
          both cases in turn.

          \begin{ssssubcase}{$\overline{r}\equiv 1010\pmod{2}$.}
            In this case, from $r$, the algorithm prescribes
            \[
            \hh{1,3,c,d}\mone{1}^{\tau_1'}\mone{3}^{\tau_3'}\mone{c}^{\tau_c}\mone{d}^{\tau_d}.
            \]
            By \cref{prop:rewriterules}, there exist words
            $\word{V}$ and $\word{W}$ over $\s{\mone{x},\xx{x,y}}$,
            with $x,y\in\s{1,2,3,4,c,d}$, such that
            \[
            \hh{1,3,c,d}\mone{1}^{\tau_1'}\mone{3}^{\tau_3'}\mone{c}^{\tau_c}\mone{d}^{\tau_d}
            \hh{1,2,3,4} \mone{d}^{\tau_d}
            \mone{c}^{\tau_c}\mone{4}^{\tau_4}\mone{3}^{\tau_3}\hh{3,4,c,d}
            \approx \word{V}\hh{1,2,3,4}\word{W}.
            \]
            Hence, we can complete the diagram as follows.
            \[
              \begin{tikzcd}[column sep=2.5cm, row sep=large]
                s \arrow[d, Rightarrow, "\hh{3,4,c,d}\mone{3}^{\tau_3}\mone{4}^{\tau_4}\mone{c}^{\tau_c}\mone{d}^{\tau_d}" left] 
                  \arrow[r, "\hh{1,2,3,4}"] & 
                  r\arrow[d, Rightarrow, "\hh{1,3,c,d}\mone{1}^{\tau_1'}\mone{3}^{\tau_3'}\mone{c}^{\tau_c}\mone{d}^{\tau_d}"] \\
                t\arrow[r, "\word{V}\hh{1,2,3,4}\word{W}", swap] & q
              \end{tikzcd}
            \]
            The diagram commutes by construction. To see that the
            level property is satisfied, first note that
            \[
            \level(t),\level(q)\leq (j,k,m-1)<\level(s).
            \] 
            Now, since
            $\word{V}$ and $\word{W}$ are words over
            $\s{\mone{x},\xx{x,y}}$, they can neither increase nor
            decrease the number of odd entries. As a result, because
            $\word{V}\hh{1,2,3,4}\word{W}$ contains a single
            occurrence of $K$, it cannot raise the level of the state
            to $\level(s)$ and also lower it back to
            $\level(q)$. Thus,
            \[
            \level(\word{V}\hh{1,2,3,4}\word{W})<\level(s)
            \]
            as desired.
          \end{ssssubcase}

          \begin{ssssubcase}{$\overline{r}\equiv 0101\pmod{2}$.}
            In this case, from $r$, the algorithm prescribes
            \[
            \hh{2,4,c,d}\mone{2}^{\tau_2'}\mone{4}^{\tau_4'}\mone{c}^{\tau_c}\mone{d}^{\tau_d}.
            \]
            By \cref{prop:rewriterules}, there exist words
            $\word{V}$ and $\word{W}$ over $\s{\mone{x},\xx{x,y}}$,
            with $x,y\in\s{1,2,3,4,c,d}$, such that
            \[
            \hh{2,4,c,d}\mone{2}^{\tau_2'}\mone{4}^{\tau_4'}\mone{c}^{\tau_c}\mone{d}^{\tau_d}
            \hh{1,2,3,4} \mone{d}^{\tau_d}
            \mone{c}^{\tau_c}\mone{4}^{\tau_4}\mone{3}^{\tau_3}\hh{3,4,c,d}
            \approx \word{V}\hh{1,2,3,4}\word{W}.
            \]
            Hence, we can complete the diagram as follows.
            \[
              \begin{tikzcd}[column sep=2.5cm, row sep=large]
                s \arrow[d, Rightarrow, "\hh{3,4,c,d}\mone{3}^{\tau_3}\mone{4}^{\tau_4}\mone{c}^{\tau_c}\mone{d}^{\tau_d}" left] 
                  \arrow[r, "\hh{1,2,3,4}"] & 
                  r\arrow[d, Rightarrow, "\hh{2,4,c,d}\mone{2}^{\tau_2'}\mone{4}^{\tau_4'}\mone{c}^{\tau_c}\mone{d}^{\tau_d}"] \\
                t\arrow[r, "\word{V}\hh{1,2,3,4}\word{W}", swap] & q
              \end{tikzcd}
            \]
            The diagram commutes by construction. To see that the
            level property is satisfied, first note that
            \[
            \level(t),\level(q)\leq (j,k,m-1)<\level(s). 
            \]
            Now, since
            $\word{V}$ and $\word{W}$ are words over
            $\s{\mone{x},\xx{x,y}}$, they can neither increase nor
            decrease the number of odd entries. As a result, because
            $\word{V}\hh{1,2,3,4}\word{W}$ contains a single
            occurrence of $K$, it cannot raise the level of the state
            to $\level(s)$ and also lower it back to
            $\level(q)$. Thus,
            \[
            \level(\word{V}\hh{1,2,3,4}\word{W})<\level(s)
            \]
            as desired.
          \end{ssssubcase}          
        \end{sssubcase}
    \end{ssubcase}
      
    \begin{ssubcase}{$|\s{a,b,c,d} \cap \s{1,2,3,4}| = 3$.}      
      Then $a,b,c \in \s{1,2,3,4}$ and $5 \leq d$. We now consider the
      cases $\s{a,b,c}=\s{1,2,3}$, $\s{a,b,c}=\s{1,2,4}$,
      $\s{a,b,c}=\s{1,3,4}$, and $\s{a,b,c}=\s{2,3,4}$ in turn.

        \begin{sssubcase}{$\s{a,b,c}=\s{1,2,3}$.}
          Then, from $s$, the algorithm prescribes
          \[
          \hh{1,2,3,d}\mone{1}^{\tau_1}\mone{2}^{\tau_2}\mone{3}^{\tau_3}\mone{d}^{\tau_d}.
          \]
          Moreover, by \cref{prop:koddodds}, $\level(r)=(j,k+1,1)$
          and, writing $\overline{r}$ for the first four entries of
          the integral part of $r$, we have $\overline{r}\equiv 1331
          \pmod{4}$ or $\overline{r}\equiv 3113 \pmod{4}$. Hence, from
          $r$ the algorithm prescribes
          \[
          \hh{1,2,3,4}\mone{1}^{\tau}\mone{2}^{\tau+1}\mone{3}^{\tau+1}\mone{4}^{\tau}
          \]
          where the value of $\tau$ depends on whether
          $\overline{r}\equiv 1331 \pmod{4}$ or $\overline{r}\equiv
          3113 \pmod{4}$. Now, since
          \[
          \hh{1,2,3,4}\mone{1}^{\tau}\mone{2}^{\tau+1} \mone{3}^{\tau+1}
          \mone{4}^{\tau}\hh{1,2,3,4}\approx
          \xx{1,4}\xx{2,3}\mone{1}^{\tau}\mone{2}^{\tau}\mone{3}^{\tau}\mone{4}^{\tau},
          \]
          by \cref{rel:k32,rel:k31}, we know that from $q_1=
          (\hh{1,2,3,4}\mone{1}^{\tau}\mone{2}^{\tau+1}\mone{3}^{\tau+1}\mone{4}^{\tau})
          r$ the algorithm prescribes
          \[
          \hh{2,3,4,d}\mone{2}^{\tau_3+\tau}\mone{3}^{\tau_2+\tau}\mone{4}^{\tau_1+\tau}\mone{d}^{\tau_d}.
          \]
          We therefore complete the resulting diagram as follows.
          \[
          \begin{tikzcd}[column sep=large]
          s \arrow[d, Rightarrow, "\hh{1,2,3,d}\mone{1}^{\tau_1}\mone{2}^{\tau_2}\mone{3}^{\tau_3}\mone{d}^{\tau_d}" left] \arrow[r, "\hh{1,2,3,4}"] & r \arrow[d, Rightarrow, "\hh{1,2,3,4}\mone{1}^{\tau}\mone{2}^{\tau+1}\mone{3}^{\tau+1}\mone{4}^{\tau}" right]\\
          t\arrow[dr, "\mone{1}^{\tau}\mone{4}\mone{d}\xx{4,d}\xx{1,2}\xx{2,3}\xx{3,4}", swap] & q_1 \arrow[d, Rightarrow, "\hh{2,3,4,d}\mone{2}^{\tau_3+\tau}\mone{3}^{\tau_2+\tau}\mone{4}^{\tau_1+\tau}\mone{d}^{\tau_d}" right] \\
          & q_2
          \end{tikzcd}
          \]
          The diagram commutes by \cref{rel:k32,rel:k31,rel:orderx,rel:ordermone,rel:disjoint2,rel:disjoint3,rel:disjoint4,rel:disjoint5,rel:rename1,rel:rename2,rel:rename3,rel:rename4,rel:rename5,rel:rename6,rel:rename7,rel:ksym5,rel:k11}
          \begin{align*}
          \begin{split}
          \hh{2,3,4,d}&\mone{2}^{\tau_3+\tau}\mone{3}^{\tau_2+\tau}\mone{4}^{\tau_1+\tau}\mone{d}^{\tau_d}
          \hh{1,2,3,4}\mone{1}^{\tau}\mone{2}^{\tau+1} \mone{3}^{\tau+1} \mone{4}^{\tau}\hh{1,2,3,4} \\             
          &\approx\hh{2,3,4,d}\mone{2}^{\tau_3+\tau}\mone{3}^{\tau_2+\tau}\mone{4}^{\tau_1+\tau}\mone{d}^{\tau_d}
          \xx{1,4}\xx{2,3}\mone{1}^{\tau}\mone{2}^{\tau}\mone{3}^{\tau}\mone{4}^{\tau}\\
          &\approx\hh{2,3,4,d}\mone{1}^{\tau}\mone{2}^{\tau_3}\mone{3}^{\tau_2}\mone{4}^{\tau_1}\mone{d}^{\tau_d}
          \xx{1,4}\xx{2,3}\\
          &\approx\hh{2,3,4,d}\mone{1}^{\tau}\xx{1,4}\xx{2,3}\mone{1}^{\tau_1}\mone{2}^{\tau_2}\mone{3}^{\tau_3}\mone{d}^{\tau_d}\\                    
          &\approx\mone{1}^{\tau}\hh{2,3,4,d}\xx{1,4}\xx{2,4}\xx{2,4}\xx{2,3}\mone{1}^{\tau_1}\mone{2}^{\tau_2}\mone{3}^{\tau_3}\mone{d}^{\tau_d}\\
          &\approx\mone{1}^{\tau}\hh{2,3,4,d}\xx{2,4}\xx{1,2}\xx{2,3}\xx{3,4}\mone{1}^{\tau_1}\mone{2}^{\tau_2}\mone{3}^{\tau_3}\mone{d}^{\tau_d}\\
          &\approx\mone{1}^{\tau}\hh{2,3,4,d}\xx{2,4}\xx{3,d}\xx{3,d}\xx{1,2}\xx{2,3}\xx{3,4}\mone{1}^{\tau_1}\mone{2}^{\tau_2}\mone{3}^{\tau_3}\mone{d}^{\tau_d}\\
          &\approx\mone{1}^{\tau}\mone{4}\mone{d}\hh{2,3,4,d}\xx{3,d}\xx{1,2}\xx{2,3}\xx{3,4}\mone{1}^{\tau_1}\mone{2}^{\tau_2}\mone{3}^{\tau_3}\mone{d}^{\tau_d}\\
          &\approx\mone{1}^{\tau}\mone{4}\mone{d}\xx{4,d}\xx{1,2}\xx{2,3}\xx{3,4}\hh{1,2,3,d}\mone{1}^{\tau_1}\mone{2}^{\tau_2}\mone{3}^{\tau_3}\mone{d}^{\tau_d}.                                        
          \end{split}
          \end{align*}
          Moreover, the level property is satisfied since
          $\level(t)<\level(s)$, $\level(q_2) < \level (q_1) =
          \level(s)$ and
          $\mone{1}^{\tau}\mone{4}\mone{d}\xx{4,d}\xx{1,2}\xx{2,3}\xx{3,4}$
          cannot increase the number of odd entries .
        \end{sssubcase}

        \begin{sssubcase}{$\s{a,b,c} = \s{1,2,4}$.}
          Then, from $s$, the algorithm prescribes
          \[
          \hh{1,2,4,d}\mone{1}^{\tau_1}\mone{2}^{\tau_2}\mone{4}^{\tau_4}\mone{d}^{\tau_d}.
          \]
          Moreover, by \cref{prop:koddodds}, $\level(r)=(j,k+1,1)$
          and, writing $\overline{r}$ for the first four entries of
          the integral part of $r$, we have $\overline{r}\equiv 1133
          \pmod{4}$ or $\overline{r}\equiv 3311 \pmod{4}$. Hence, from
          $r$ the algorithm prescribes
          \[
          \hh{1,2,3,4}\mone{1}^{\tau}\mone{2}^{\tau}\mone{3}^{\tau+1}\mone{4}^{\tau+1}
          \]
          where the value of $\tau$ depends on whether
          $\overline{r}\equiv 1133 \pmod{4}$ or $\overline{r}\equiv
          3311 \pmod{4}$.  Now, since
          \[
          \hh{1,2,3,4}\mone{1}^{\tau}\mone{2}^{\tau} \mone{3}^{\tau+1}
          \mone{4}^{\tau+1}\hh{1,2,3,4}\approx
          \xx{1,3}\xx{2,4}\mone{1}^{\tau}\mone{2}^{\tau}\mone{3}^{\tau}\mone{4}^{\tau},
          \]
          by \cref{rel:k12,rel:k11}, we know that from $q_1=
          (\hh{1,2,3,4}\mone{1}^{\tau}\mone{2}^{\tau}\mone{3}^{\tau+1}\mone{4}^{\tau+1})
          r$ the algorithm prescribes
          \[
          \hh{2,3,4,d}\mone{2}^{\tau_4+\tau}\mone{3}^{\tau_1+\tau}\mone{4}^{\tau_2+\tau}\mone{d}^{\tau_d}.
          \]
          We therefore complete the resulting diagram as follows.
          \[
          \begin{tikzcd}[column sep=large]
          s \arrow[d, Rightarrow, "\hh{1,2,4,d}\mone{1}^{\tau_1}\mone{2}^{\tau_2}\mone{4}^{\tau_4}\mone{d}^{\tau_d}" left] \arrow[r, "\hh{1,2,3,4}"] & r \arrow[d, Rightarrow, "\hh{1,2,3,4}\mone{1}^{\tau}\mone{2}^{\tau}\mone{3}^{\tau+1}\mone{4}^{\tau+1}" right]\\
          t\arrow[dr, "\mone{1}^\tau\mone{3}\mone{d}\xx{3,d}\xx{3,4}\xx{1,3}\xx{1,2}", swap] & q_1 \arrow[d, Rightarrow, "\hh{2,3,4,d}\mone{2}^{\tau_4+\tau}\mone{3}^{\tau_1+\tau}\mone{4}^{\tau_2+\tau}\mone{d}^{\tau_d}" right] \\
          & q_2
          \end{tikzcd}
          \]
          The diagram commutes by
          \cref{rel:k12,rel:k11,rel:orderk,rel:ordermone,rel:orderx,rel:disjoint2,rel:disjoint4,rel:disjoint5,rel:rename3,rel:rename4,rel:rename1,rel:rename2,rel:rename5,rel:ksym1,rel:ksym2}
          \begin{align*}
          \begin{split}
          \hh{2,3,4,d}&\mone{2}^{\tau_4+\tau}\mone{3}^{\tau_1+\tau}\mone{4}^{\tau_2+\tau}\mone{d}^{\tau_d}
          \hh{1,2,3,4}\mone{1}^{\tau}\mone{2}^{\tau} \mone{3}^{\tau+1} \mone{4}^{\tau+1}\hh{1,2,3,4} \\             
          &\approx\hh{2,3,4,d}\mone{2}^{\tau_4+\tau}\mone{3}^{\tau_1+\tau}\mone{4}^{\tau_2+\tau}\mone{d}^{\tau_d}
          \xx{1,3}\xx{2,4}\mone{1}^{\tau}\mone{2}^{\tau}\mone{3}^{\tau}\mone{4}^{\tau}\\
          &\approx\hh{2,3,4,d}\xx{1,3}\xx{2,4}\mone{1}^{\tau_1}\mone{2}^{\tau_2}\mone{3}^{\tau}\mone{4}^{\tau_4}\mone{d}^{\tau_d}\\
          &\approx\hh{2,3,4,d}\xx{1,2}\xx{1,2}\xx{1,3}\xx{2,4}\mone{1}^{\tau_1}\mone{2}^{\tau_2}\mone{3}^{\tau}\mone{4}^{\tau_4}\mone{d}^{\tau_d}\\
          &\approx\xx{1,2}\hh{1,3,4,d}\xx{2,3}\xx{1,2}\xx{2,4}\mone{1}^{\tau_1}\mone{2}^{\tau_2}\mone{3}^{\tau}\mone{4}^{\tau_4}\mone{d}^{\tau_d}\\
          &\approx\xx{1,2}\xx{2,3}\hh{1,2,4,d}\xx{1,2}\xx{2,4}\mone{3}^\tau\mone{1}^{\tau_1}\mone{2}^{\tau_2}\mone{4}^{\tau_4}\mone{d}^{\tau_d}\\
          &\approx\xx{1,2}\xx{2,3}\mone{2}\mone{d}\xx{2,d}\xx{2,4}\mone{3}^\tau\hh{1,2,4,d}\mone{1}^{\tau_1}\mone{2}^{\tau_2}\mone{4}^{\tau_4}\mone{d}^{\tau_d}\\
          &\approx\mone{1}^\tau\mone{3}\mone{d}\xx{1,2}\xx{2,3}\xx{2,d}\xx{2,4}\hh{1,2,4,d}\mone{1}^{\tau_1}\mone{2}^{\tau_2}\mone{4}^{\tau_4}\mone{d}^{\tau_d}\\
          &\approx\mone{1}^\tau\mone{3}\mone{d}\xx{3,d}\xx{3,4}\xx{1,3}\xx{1,2}\hh{1,2,4,d}\mone{1}^{\tau_1}\mone{2}^{\tau_2}\mone{4}^{\tau_4}\mone{d}^{\tau_d}.          
          \end{split}
          \end{align*}
          Moreover, the level property is satisfied since
          $\level(t)<\level(s)$, $\level(q_2) < \level (q_1) =
          \level(s)$ and
          $\mone{1}^\tau\mone{3}\mone{d}\xx{3,d}\xx{3,4}\xx{1,3}\xx{1,2}$
          cannot increase the number of odd entries.
        \end{sssubcase}
      
        \begin{sssubcase}{$\s{a,b,c} = \s{1,3,4}$.}
          Then, from $s$, the algorithm prescribes
          \[
          \hh{1,3,4,d}\mone{1}^{\tau_1}\mone{3}^{\tau_3}\mone{4}^{\tau_4}\mone{d}^{\tau_d}.
          \]
          Moreover, by \cref{prop:koddodds}, $\level(r)=(j,k+1,1)$
          and, writing $\overline{r}$ for the first four entries of
          the integral part of $r$, we have $\overline{r}\equiv 1313
          \pmod{4}$ or $\overline{r}\equiv 3131 \pmod{4}$. Hence, from
          $r$ the algorithm prescribes
          \[
          \hh{1,2,3,4}\mone{1}^{\tau}\mone{2}^{\tau+1}\mone{3}^{\tau}\mone{4}^{\tau+1}
          \]
          where the value of $\tau$ depends on whether
          $\overline{r}\equiv 1313 \pmod{4}$ or $\overline{r}\equiv
          3131 \pmod{4}$.  Now, since
          \[
          \hh{1,2,3,4}\mone{1}^{\tau}\mone{2}^{\tau+1} \mone{3}^{\tau}
          \mone{4}^{\tau+1}\hh{1,2,3,4}\approx
          \xx{1,2}\xx{3,4}\mone{1}^{\tau}\mone{2}^{\tau}\mone{3}^{\tau}\mone{4}^{\tau},
          \]
          by \cref{rel:k22,rel:k21}, we know that from $q_1=
          (\hh{1,2,3,4}\mone{1}^{\tau}\mone{2}^{\tau+1}\mone{3}^{\tau}\mone{4}^{\tau+1})
          r$ the algorithm prescribes
          \[
          \hh{2,3,4,d}\mone{2}^{\tau_1+\tau}\mone{3}^{\tau_4+\tau}\mone{4}^{\tau_3+\tau}\mone{d}^{\tau_d}.
          \]
          We therefore complete the resulting diagram as follows.
          \[
          \begin{tikzcd}[column sep=large]
          s \arrow[d, Rightarrow, "\hh{1,3,4,d}\mone{1}^{\tau_1}\mone{3}^{\tau_3}\mone{4}^{\tau_4}\mone{d}^{\tau_d}" left] \arrow[r, "\hh{1,2,3,4}"] & r \arrow[d, Rightarrow, "\hh{1,2,3,4}\mone{1}^{\tau}\mone{2}^{\tau+1}\mone{3}^{\tau}\mone{4}^{\tau+1}" right]\\
          t\arrow[dr, "\mone{1}^\tau\xx{1,2}\xx{3,4}", swap] & q_1 \arrow[d, Rightarrow, "\hh{2,3,4,d}\mone{2}^{\tau_1+\tau}\mone{3}^{\tau_4+\tau}\mone{4}^{\tau_3+\tau}\mone{d}^{\tau_d}" right] \\
          & q_2
          \end{tikzcd}
          \]
          The diagram commutes by
          \cref{rel:k22,rel:k21,rel:ordermone,rel:disjoint2,rel:disjoint4,rel:disjoint5,rel:rename3,rel:rename4,rel:ksym2}
          \begin{align*}
          \begin{split}
          \hh{2,3,4,d}&\mone{2}^{\tau_1+\tau}\mone{3}^{\tau_4+\tau}\mone{4}^{\tau_3+\tau}\mone{d}^{\tau_d}
          \hh{1,2,3,4}\mone{1}^{\tau}\mone{2}^{\tau+1} \mone{3}^{\tau} \mone{4}^{\tau+1}\hh{1,2,3,4} \\
          &\approx \hh{2,3,4,d}\mone{2}^{\tau_1+\tau}\mone{3}^{\tau_4+\tau}\mone{4}^{\tau_3+\tau}\mone{d}^{\tau_d}\xx{1,2}\xx{3,4}
          \mone{1}^{\tau}\mone{2}^{\tau} \mone{3}^{\tau} \mone{4}^{\tau} \\
          &\approx \hh{2,3,4,d}\xx{1,2}\xx{3,4}\mone{1}^{\tau_1}\mone{2}^{\tau}\mone{3}^{\tau_3}\mone{4}^{\tau_4}\mone{d}^{\tau_d}\\
          &\approx \xx{1,2}\xx{3,4}\hh{1,3,4,d}\mone{1}^{\tau_1}\mone{2}^{\tau}\mone{3}^{\tau_3}\mone{4}^{\tau_4}\mone{d}^{\tau_d}\\
          &\approx \mone{1}^\tau\xx{1,2}\xx{3,4}\hh{1,3,4,d}\mone{1}^{\tau_1}\mone{3}^{\tau_3}\mone{4}^{\tau_4}\mone{d}^{\tau_d}.
          \end{split}
          \end{align*}
          Moreover, the level property is satisfied since
          $\level(t)<\level(s)$, $\level(q_2) < \level (q_1) =
          \level(s)$ and $\mone{1}^\tau\xx{1,2}\xx{3,4}$ cannot
          increase the number of odd entries.
        \end{sssubcase}
      
        \begin{sssubcase}{$\s{a,b,c} = \s{2,3,4}$.}
          Then, from $s$, the algorithm prescribes
          \[
          \hh{2,3,4,d}\mone{2}^{\tau_2}\mone{3}^{\tau_3}\mone{4}^{\tau_4}\mone{d}^{\tau_d}.
          \]
          Moreover, by \cref{prop:koddodds}, $\level(r)=(j,k+1,1)$ and,
          writing $\overline{r}$ for the first four entries of the
          integral part of $r$, we have $\overline{r}\equiv 1111
          \pmod{4}$ or $\overline{r}\equiv 3333 \pmod{4}$. Hence, from
          $r$ the algorithm prescribes
          \[
          \hh{1,2,3,4}\mone{1}^{\tau}\mone{2}^{\tau}\mone{3}^{\tau}\mone{4}^{\tau}
          \]
          where the value of $\tau$ depends on whether
          $\overline{r}\equiv 1111 \pmod{4}$ or $\overline{r}\equiv
          3333 \pmod{4}$. Now, since
          \[
          \hh{1,2,3,4}\mone{1}^{\tau}\mone{2}^{\tau} \mone{3}^{\tau}
          \mone{4}^{\tau}\hh{1,2,3,4}\approx
          \mone{1}^{\tau}\mone{2}^{\tau}\mone{3}^{\tau}\mone{4}^{\tau},
          \]
          by \cref{itm:relk4}, we know that from $q_1=
          (\hh{1,2,3,4}\mone{1}^{\tau}\mone{2}^{\tau}\mone{3}^{\tau}\mone{4}^{\tau})
          r$ the algorithm prescribes
          \[
          \hh{2,3,4,d}\mone{2}^{\tau_2+\tau}\mone{3}^{\tau_3+\tau}\mone{4}^{\tau_4+\tau}\mone{d}^{\tau_d}.
          \]
          We therefore complete the resulting diagram as follows.
          \[
          \begin{tikzcd}[column sep=large]
          s \arrow[d, Rightarrow, "\hh{2,3,4,d}\mone{2}^{\tau_2}\mone{3}^{\tau_3}\mone{4}^{\tau_4}\mone{d}^{\tau_d}" left] \arrow[r, "\hh{1,2,3,4}"] & r \arrow[d, Rightarrow, "\hh{1,2,3,4}\mone{1}^{\tau}\mone{2}^{\tau}\mone{3}^{\tau}\mone{4}^{\tau}" right]\\
          t\arrow[dr, "\mone{1}^\tau", swap] & q_1 \arrow[d, Rightarrow, "\hh{2,3,4,d}\mone{2}^{\tau_2+\tau}\mone{3}^{\tau_3+\tau}\mone{4}^{\tau_4+\tau}\mone{d}^{\tau_d}" right] \\
          & q_2
          \end{tikzcd}
          \]
          The diagram commutes by
          \cref{itm:relk4,rel:ordermone,rel:disjoint4}
          \begin{align*}
          \begin{split}
          \hh{2,3,4,d}&\mone{2}^{\tau_2+\tau}\mone{3}^{\tau_3+\tau}\mone{4}^{\tau_4+\tau}\mone{d}^{\tau_d}
          \hh{1,2,3,4}\mone{1}^{\tau}\mone{2}^{\tau} \mone{3}^{\tau} \mone{4}^{\tau}\hh{1,2,3,4} \\             
          &\approx \hh{2,3,4,d}\mone{2}^{\tau_2+\tau}\mone{3}^{\tau_3+\tau}\mone{4}^{\tau_4+\tau}\mone{d}^{\tau_d}
          \mone{1}^{\tau}\mone{2}^{\tau} \mone{3}^{\tau} \mone{4}^{\tau} \\
          &\approx\hh{2,3,4,d}\mone{1}^\tau\mone{2}^{\tau_2}\mone{3}^{\tau_3}\mone{4}^{\tau_4}\mone{d}^{\tau_d}\\
          &\approx\hh{2,3,4,d}\mone{2}^{\tau_2}\mone{3}^{\tau_3}\mone{4}^{\tau_4}\mone{d}^{\tau_d}\mone{1}^\tau.
          \end{split}
          \end{align*}
          Moreover, the level property is satisfied since
          $\level(t)<\level(s)$, $\level(q_2) < \level (q_1) =
          \level(s)$ and $\mone{1}^\tau$ cannot increase the number of
          odd entries.
        \end{sssubcase}      
      \end{ssubcase}

      \begin{ssubcase}{$|\s{a,b,c,d} \cap \s{1,2,3,4}| = 4$.}
        Then the first odd entries of $u$ are odd and for $1\leq i\leq
        4$, there is $\tau_i\in\Z_2$ such that $u_i\equiv
        (-1)^{\tau_i}\pmod{4}$. We now consider the cases
        \[
        \tau_1+\tau_2+\tau_3 +\tau_4 \equiv 0 \pmod{2} \quad \mbox{and} \quad \tau_1+\tau_2+\tau_3 +\tau_4 \equiv 1 \pmod{2}
        \]
        in turn.

        \begin{sssubcase}{$\tau_1+\tau_2+\tau_3 +\tau_4 \equiv 0 \pmod{2}$.}
          Then, by \cref{applem:evenoddsOddodds}, we have $\level(r)
          \leq (j,k,m-1) < \level(s)$. From $s$, the algorithm
          prescribes
          \[
          \hh{1,2,3,4}\mone{1}^{\tau_1} \mone{2}^{\tau_2}
          \mone{3}^{\tau_3}\mone{4}^{\tau_4}
          \]
          where evenly many of the $\tau_i$ are odd. By
          \cref{prop:kevenk}, there exists a word $V$ over
          $\s{\mone{x},\xx{x,y}}$, with $x,y\in\s{1,2,3,4}$, such that
          $\hh{1,2,3,4}\mone{1}^{\tau_1}
          \mone{2}^{\tau_2}\mone{3}^{\tau_3}\mone{4}^{\tau_4}\hh{1,2,3,4}
          \approx V$.  Hence, we can complete the diagram as follows.
          \[
          \begin{tikzcd}[column sep=large]
            s \arrow[d, Rightarrow, "\hh{1,2,3,4}\mone{1}^{\tau_1}\mone{2}^{\tau_2}\mone{3}^{\tau_3}\mone{4}^{\tau_4}" left] 
            \arrow[r, "\hh{1,2,3,4}"] & 
            r\\
            t\arrow[ur, "V" swap]
          \end{tikzcd}\\
          \] 
          The diagram commutes by \cref{rel:orderk,rel:ordermone},
          since
          \[
          \hh{1,2,3,4}\mone{1}^{\tau_1}
          \mone{2}^{\tau_2}\mone{3}^{\tau_3}
          \mone{4}^{\tau_4}\hh{1,2,3,4} \approx V.
          \]
          Moreover, the level property is satisfied since
          $\level(t),\level(r) \leq (j,k,m-1) < \level(s)$ and a word
          over $\s{\mone{x},\xx{x,y}}$ cannot increase the number of
          odd entries.
        \end{sssubcase}

        \begin{sssubcase}{$\tau_1+\tau_2+\tau_3 +\tau_4 \equiv 1 \pmod{2}$.}
          Then, by \cref{applem:evenoddsOddodds}, we have
          $\level(r)=\level(s)$. From $s$ and $r$, the algorithm
          prescribes
          \[
          \hh{1,2,3,4}\mone{1}^{\tau_1} \mone{2}^{\tau_2}
          \mone{3}^{\tau_3}\mone{4}^{\tau_4} \quad \mbox{and} \quad
          \hh{1,2,3,4}\mone{1}^{\tau_1'} \mone{2}^{\tau_2'}
          \mone{3}^{\tau_3'}\mone{4}^{\tau_4'},
          \]
          respectively, where oddly many of the $\tau_i$ are odd and
          oddly many of the $\tau_i'$ are odd. By
          \cref{prop:koddkoddk}, there exists a word $V$ over
          $\s{\mone{x},\xx{x,y}}$, with $x,y\in\s{1,2,3,4}$, such that
          \[
          \hh{1,2,3,4}\mone{1}^{\tau_1}
          \mone{2}^{\tau_2}\mone{3}^{\tau_3}\mone{4}^{\tau_4}\hh{1,2,3,4}
          \mone{1}^{\tau_1'}
          \mone{2}^{\tau_2'}\mone{3}^{\tau_3'}\mone{4}^{\tau_4'}\hh{1,2,3,4}
          \approx V.
          \]
          Hence, we can complete the diagram as follows.
          \[
          \begin{tikzcd}[column sep=large, row sep=large]
              s \arrow[d, Rightarrow, "\hh{1,2,3,4}\mone{1}^{\tau_1}
          \mone{2}^{\tau_2}\mone{3}^{\tau_3}\mone{4}^{\tau_4}" left] 
              \arrow[r, "\hh{1,2,3,4}"] & 
              r\arrow[d, Rightarrow, "\hh{1,2,3,4}\mone{1}^{\tau_1'}
          \mone{2}^{\tau_2'}\mone{3}^{\tau_3'}\mone{4}^{\tau_4'}"] \\
              t\arrow[r, "V", swap] & q
          \end{tikzcd}
          \]
          The diagram commutes by \cref{rel:orderk,rel:ordermone},
          since
          \[
          \hh{1,2,3,4}\mone{1}^{\tau_1}
          \mone{2}^{\tau_2}\mone{3}^{\tau_3}\mone{4}^{\tau_4}\hh{1,2,3,4}
          \mone{1}^{\tau_1'}
          \mone{2}^{\tau_2'}\mone{3}^{\tau_3'}\mone{4}^{\tau_4'}\hh{1,2,3,4}
          \approx V.
          \]
          Moreover, the level property is satisfied since
          $\level(t)<\level(s)$ and $\level(q) \leq \level(r) =
          \level(s)$ and a word over $\s{\mone{x},\xx{x,y}}$ cannot
          increase the number of odd entries.
        \end{sssubcase}
      \end{ssubcase}
  \end{subcase}
\end{case}
\end{proof}  

\cref{lem:mainbasic} provides a restricted version of the Main
Lemma. We now show that it implies the full version.

\begin{lemma}
  \label{lem:wedge}
  Suppose $\word{N}^*:s\too t$ and $\word{M}^*:s\too r$ are (possibly
  empty) sequences of normal edges with a common source. Then there
  exists a sequence of simple edges $\word{G}^*:t\to r$ such that the
  diagram
  \[
  \begin{tikzcd}[column sep=small, row sep=large]
    & s \arrow[dl, Rightarrow, "\word{N}^*" swap]
    \arrow[dr, Rightarrow, "\word{M}^*"]
    \\
    t \arrow[rr, "\word{G}^*" swap] && r
  \end{tikzcd}
  \]
  commutes equationally and $\level(\word{G}^*)\leq
  \max(\level(t),\level(r))$.
\end{lemma}

\begin{proof}
  Since there is at most one normal edge from any given state, either
  $\word{N}^*$ must be a prefix of $\word{M}^*$ or vice
  versa. Therefore, there exists a sequence of normal
  edges either $\word{P}^*:t\too r$ or $\word{Q}^*:r\too t$. In the former
  case we take $\word{G}^*=\word{P}^*$, and in the latter case we take
  $\word{G}^*={\word{Q}^*}^{-1}$.
\end{proof}

\begin{lemma*}[Main Lemma]
  Let $s$, $t$, and $r$ be states, $N:s\too t$ be a normal edge, and
  $G:s\to r$ be a simple edge. Then there exist a state $q$, a
  sequence of normal edges $\word{N^*}:r\too q$, and a sequence of
  simple edges $\word{G^*}:t\to q$ such that the diagram
  \[
    \begin{tikzcd}[column sep=large, row sep=large]
      s \arrow[d, Rightarrow, "N" left] 
        \arrow[r, "G"] & 
        r\arrow[d, Rightarrow, "\word{N^*}"] \\
      t\arrow[r, "\word{G^*}" below] & q
    \end{tikzcd}
  \]
  commutes equationally and $\level(\word{G^*})<\level(s)$.
\end{lemma*}

\begin{proof}
By \cref{lem:simp}, there exists a sequence of basic edges $\word{H}^*
= H_1\ldots H_k$ such that $\word{H}^*\approx G$ and
$\level(\word{H}^*) = \level(G)$. For $1\leq j \leq k$, assume that
$H_j:s_j \to s_{j+1}$, with $s_1=s$ and $s_{k+1}=r$.

For each $1\leq j \leq k$, let $N_j:s_j \to t_j$ be the normal edge
originating at $s_j$. Note that $N_1 = N:s\to t$. By
\cref{lem:mainbasic}, there exist a state $q_j$, a sequence of normal
edges $\word{N^*}_j:s_{j+1}\too q_j$, and a sequence of simple edges
$\word{H^*}_j:t_j\to q_j$ such that the diagram
  \[
    \begin{tikzcd}[column sep=large, row sep=large]
      s_j \arrow[d, Rightarrow, "N_j" left] 
        \arrow[r, "H_j"] & 
        s_{j+1}\arrow[d, Rightarrow, "\word{N^*}_j"] \\
      t_j\arrow[r, "\word{H^*}_j" below] & q_j
    \end{tikzcd}
  \]
commutes equationally and $\level(\word{H}^*_j)<\level(s_j)$.

Moreover, for every $1\leq j \leq k-1$, $\word{N}^*_j : s_{j+1} \too
q_j$ and $N_{j+1}:s_{j+1} \too t_{j+1}$ are two sequences of normal
edges with a common source. Hence, by \cref{lem:wedge}, there exists a
sequence of simple edges $\word{F}_j^*:q_j \to t_{j+1}$ such that the
diagram
  \[
  \begin{tikzcd}[column sep=small, row sep=large]
    & s_{j+1} \arrow[dl, Rightarrow, "\word{N}_j^*" swap]
    \arrow[dr, Rightarrow, "N_{j+1}"]
    \\
    q_j \arrow[rr, "\word{F}_j^*" swap] && t_{j+1}
  \end{tikzcd}
  \]
commutes equationally and $\level(\word{F}_j^*)< \level(s_{j+1})$.

Now let $q=q_k$ and define $\word{G}^*:t \to q$ by $\word{G}^* =\word{H}^*_k.  \word{F}^*_{k-1} \word{H}^*_{k-1} \ldots \word{F}^*_2
\word{H}_2^* \word{F}^*_1 \word{H}^*_1$. Then the diagram
  \[
    \begin{tikzcd}[column sep=large, row sep=large]
      s \arrow[d, Rightarrow, "N" left] 
        \arrow[r, "G"] & 
        r\arrow[d, Rightarrow, "\word{N^*}"] \\
      t\arrow[r, "\word{G^*}" below] & q
    \end{tikzcd}
  \]
commutes equationally and $\level(\word{G^*})<\level(s)$, as desired.
\end{proof}


\bibliographystyle{eptcs}
\bibliography{intrels}

\begin{thebibliography}{10}
\providecommand{\bibitemdeclare}[2]{}
\providecommand{\surnamestart}{}
\providecommand{\surnameend}{}
\providecommand{\urlprefix}{Available at }
\providecommand{\url}[1]{\texttt{#1}}
\providecommand{\href}[2]{\texttt{#2}}
\providecommand{\urlalt}[2]{\href{#1}{#2}}
\providecommand{\doi}[1]{doi:\urlalt{http://dx.doi.org/#1}{#1}}
\providecommand{\bibinfo}[2]{#2}

\bibitemdeclare{unpublished}{Aharonov03asimple}
\bibitem{Aharonov03asimple}
\bibinfo{author}{Dorit \surnamestart Aharonov\surnameend}
  (\bibinfo{year}{2003}): \emph{\bibinfo{title}{A simple proof that {Toffoli}
  and {Hadamard} are quantum universal}}.
\newblock \bibinfo{note}{Available at \arxiv{quant-ph/0301040}}.

\bibitemdeclare{article}{amy2019}
\bibitem{amy2019}
\bibinfo{author}{Matthew \surnamestart Amy\surnameend} (\bibinfo{year}{2019}):
  \emph{\bibinfo{title}{Towards large-scale functional verification of
  universal quantum circuits}}.
\newblock {\sl \bibinfo{journal}{Electronic Proceedings in Theoretical Computer
  Science}} \bibinfo{volume}{287}, pp. \bibinfo{pages}{1--21},
  \doi{10.4204/EPTCS.287.1}.

\bibitemdeclare{article}{AGR2019}
\bibitem{AGR2019}
\bibinfo{author}{Matthew \surnamestart Amy\surnameend},
  \bibinfo{author}{Andrew~N. \surnamestart Glaudell\surnameend} \&
  \bibinfo{author}{Neil~J. \surnamestart Ross\surnameend}
  (\bibinfo{year}{2020}): \emph{\bibinfo{title}{Number-theoretic
  characterizations of some restricted {Clifford}+{$T$} circuits}}.
\newblock {\sl \bibinfo{journal}{{Quantum}}} \bibinfo{volume}{4}, p.
  \bibinfo{pages}{252}, \doi{10.22331/q-2020-04-06-252}.
\newblock \bibinfo{note}{Also available at \arxiv{1908.06076}}.

\bibitemdeclare{article}{bjs2010}
\bibitem{bjs2010}
\bibinfo{author}{Michael~J. \surnamestart Bremner\surnameend},
  \bibinfo{author}{Richard \surnamestart Jozsa\surnameend} \&
  \bibinfo{author}{Dan~J. \surnamestart Shepherd\surnameend}
  (\bibinfo{year}{2011}): \emph{\bibinfo{title}{Classical simulation of
  commuting quantum computations implies collapse of the polynomial
  hierarchy}}.
\newblock {\sl \bibinfo{journal}{Proceedings of The Royal Society A}}
  \bibinfo{volume}{467}(\bibinfo{number}{2126}), \doi{10.1098/rspa.2010.0301}.
\newblock \bibinfo{note}{Also available at \arxiv{1005.1407}}.

\bibitemdeclare{article}{fgkm15}
\bibitem{fgkm15}
\bibinfo{author}{Simon \surnamestart Forest\surnameend}, \bibinfo{author}{David
  \surnamestart Gosset\surnameend}, \bibinfo{author}{Vadym \surnamestart
  Kliuchnikov\surnameend} \& \bibinfo{author}{David \surnamestart
  McKinnon\surnameend} (\bibinfo{year}{2015}): \emph{\bibinfo{title}{Exact
  synthesis of single-qubit unitaries over {Clifford}-Cyclotomic gate sets}}.
\newblock {\sl \bibinfo{journal}{Journal of Mathematical Physics}}
  \bibinfo{volume}{56}(\bibinfo{number}{8}), p. \bibinfo{pages}{082201},
  \doi{10.1063/1.4927100}.
\newblock \bibinfo{note}{Also available at \arxiv{1501.04944}}.

\bibitemdeclare{article}{GS13}
\bibitem{GS13}
\bibinfo{author}{Brett \surnamestart Giles\surnameend} \&
  \bibinfo{author}{Peter \surnamestart Selinger\surnameend}
  (\bibinfo{year}{2013}): \emph{\bibinfo{title}{Exact synthesis of multiqubit
  {Clifford}+{$T$} circuits}}.
\newblock {\sl \bibinfo{journal}{Physical Review A}}
  \bibinfo{volume}{87}(\bibinfo{number}{3}), p. \bibinfo{pages}{032332},
  \doi{10.1103/PhysRevA.87.032332}.
\newblock \bibinfo{note}{Also available at \arxiv{1212.0506}}.

\bibitemdeclare{article}{glaudell2021optimal}
\bibitem{glaudell2021optimal}
\bibinfo{author}{Andrew~N. \surnamestart Glaudell\surnameend},
  \bibinfo{author}{Neil~J. \surnamestart Ross\surnameend} \&
  \bibinfo{author}{Jacob~M. \surnamestart Taylor\surnameend}
  (\bibinfo{year}{2021}): \emph{\bibinfo{title}{Optimal two-qubit circuits for
  universal fault-tolerant quantum computation}}.
\newblock {\sl \bibinfo{journal}{npj Quantum Information}}
  \bibinfo{volume}{7}(\bibinfo{number}{103}), \doi{10.1038/s41534-021-00424-z}.

\bibitemdeclare{mastersthesis}{Gr2014}
\bibitem{Gr2014}
\bibinfo{author}{Seth E.~M. \surnamestart Greylyn\surnameend}
  (\bibinfo{year}{2014}): \emph{\bibinfo{title}{Generators and relations for
  the group {$\mathrm{U}_4(\mathbb{Z}[1/\sqrt{2},i])$}}}.
\newblock Master's thesis, \bibinfo{school}{Department of Mathematics and
  Statistics, Dalhousie University}.
\newblock \bibinfo{note}{Available at \arxiv{1408.6204}}.

\bibitemdeclare{book}{KLM07}
\bibitem{KLM07}
\bibinfo{author}{Phillip \surnamestart Kaye\surnameend},
  \bibinfo{author}{Raymond \surnamestart Laflamme\surnameend} \&
  \bibinfo{author}{Michele \surnamestart Mosca\surnameend}
  (\bibinfo{year}{2007}): \emph{\bibinfo{title}{An Introduction to Quantum
  Computing}}.
\newblock \bibinfo{publisher}{Oxford University Press},
  \doi{10.1093/oso/9780198570004.001.0001}.

\bibitemdeclare{unpublished}{kbry15}
\bibitem{kbry15}
\bibinfo{author}{Vadym \surnamestart Kliuchnikov\surnameend},
  \bibinfo{author}{Alex \surnamestart Bocharov\surnameend},
  \bibinfo{author}{Martin \surnamestart Roetteler\surnameend} \&
  \bibinfo{author}{Jon \surnamestart Yard\surnameend} (\bibinfo{year}{2015}):
  \emph{\bibinfo{title}{A framework for approximating qubit unitaries}}.
\newblock \bibinfo{note}{Available at \arxiv{1510.03888}}.

\bibitemdeclare{article}{KMM-exact}
\bibitem{KMM-exact}
\bibinfo{author}{Vadym \surnamestart Kliuchnikov\surnameend},
  \bibinfo{author}{Dmitri \surnamestart Maslov\surnameend} \&
  \bibinfo{author}{Michele \surnamestart Mosca\surnameend}
  (\bibinfo{year}{2013}): \emph{\bibinfo{title}{Fast and efficient exact
  synthesis of single-qubit unitaries generated by {Clifford} and {$T$}
  gates}}.
\newblock {\sl \bibinfo{journal}{Quantum Information \& Computation}}
  \bibinfo{volume}{13}(\bibinfo{number}{7-8}), pp. \bibinfo{pages}{607--630},
  \doi{10.26421/QIC13.7-8-4}.
\newblock \bibinfo{note}{Avaiable at \arxiv{1206.5236}}.

\bibitemdeclare{article}{kmm-approx}
\bibitem{kmm-approx}
\bibinfo{author}{Vadym \surnamestart Kliuchnikov\surnameend},
  \bibinfo{author}{Dmitri \surnamestart Maslov\surnameend} \&
  \bibinfo{author}{Michele \surnamestart Mosca\surnameend}
  (\bibinfo{year}{2016}): \emph{\bibinfo{title}{Practical approximation of
  single-qubit unitaries by single-qubit quantum {Clifford} and {T} circuits}}.
\newblock {\sl \bibinfo{journal}{IEEE Transactions on Computers}}
  \bibinfo{volume}{65}(\bibinfo{number}{1}), pp. \bibinfo{pages}{161--172},
  \doi{10.1109/TC.2015.2409842}.
\newblock \bibinfo{note}{Also available at \arxiv{1212.6964}}.

\bibitemdeclare{unpublished}{ky15}
\bibitem{ky15}
\bibinfo{author}{Vadym \surnamestart Kliuchnikov\surnameend} \&
  \bibinfo{author}{Jon \surnamestart Yard\surnameend} (\bibinfo{year}{2015}):
  \emph{\bibinfo{title}{A framework for exact synthesis}}.
\newblock \bibinfo{note}{Available at \arxiv{1504.04350}}.

\bibitemdeclare{article}{Montanaro2017}
\bibitem{Montanaro2017}
\bibinfo{author}{Ashley \surnamestart Montanaro\surnameend}
  (\bibinfo{year}{2017}): \emph{\bibinfo{title}{Quantum circuits and low-degree
  polynomials over ${{\mathbb{F}}_\mathsf{2}}$}}.
\newblock {\sl \bibinfo{journal}{Journal of Physics A}}
  \bibinfo{volume}{50}(\bibinfo{number}{8}), p. \bibinfo{pages}{084002},
  \doi{10.1088/1751-8121/aa565f}.
\newblock \bibinfo{note}{Also available at \arxiv{1607.08473}}.

\bibitemdeclare{article}{r15}
\bibitem{r15}
\bibinfo{author}{Neil~J. \surnamestart Ross\surnameend} (\bibinfo{year}{2015}):
  \emph{\bibinfo{title}{Optimal ancilla-Free {Clifford}+{V} approximation of
  $z$-rotations}}.
\newblock {\sl \bibinfo{journal}{Quantum Information \& Computation}}
  \bibinfo{volume}{15}(\bibinfo{number}{11--12}), pp.
  \bibinfo{pages}{932--950}, \doi{10.26421/QIC15.11-12-4}.
\newblock \bibinfo{note}{Also available at \arxiv{1409.4355}}.

\bibitemdeclare{article}{RS16}
\bibitem{RS16}
\bibinfo{author}{Neil~J. \surnamestart Ross\surnameend} \&
  \bibinfo{author}{Peter \surnamestart Selinger\surnameend}
  (\bibinfo{year}{2016}): \emph{\bibinfo{title}{Optimal ancilla-free
  {Clifford}+{$T$} approximation of $z$-rotations}}.
\newblock {\sl \bibinfo{journal}{Quantum Information \& Computation}}
  \bibinfo{volume}{16}(\bibinfo{number}{11--12}), pp.
  \bibinfo{pages}{901--953}, \doi{10.26421/QIC16.11-12-1}.
\newblock \bibinfo{note}{Also available at \arxiv{1403.2975}}.

\bibitemdeclare{article}{Shi2003}
\bibitem{Shi2003}
\bibinfo{author}{Yaoyun \surnamestart Shi\surnameend} (\bibinfo{year}{2003}):
  \emph{\bibinfo{title}{Both {Toffoli} and Controlled-{NOT} need little help to
  do universal quantum computing}}.
\newblock {\sl \bibinfo{journal}{Quantum Information \& Computation}}
  \bibinfo{volume}{3}(\bibinfo{number}{1}), pp. \bibinfo{pages}{84--92},
  \doi{10.26421/QIC3.1-7}.
\newblock \bibinfo{note}{Also available at \arxiv{quant-ph/0205115}}.

\end{thebibliography}

\end{document}